\pdfoutput=1

\newif\ifcomments 

\commentstrue

\documentclass[11pt]{article}
\usepackage[in]{fullpage}

\usepackage{iftex}
\ifPDFTeX
  \usepackage[utf8]{inputenc}
  \usepackage[noTeX]{mmap}
  \usepackage[T1]{fontenc}
\fi
\ifLuaTeX
  \usepackage{luatex85}
  \usepackage[noTeX]{mmap}
\fi

\usepackage[shortlabels]{enumitem}
\usepackage{bm}
\usepackage{amsfonts}
\usepackage{amsmath}
\usepackage{xcolor}
\usepackage{graphicx}
\usepackage{amsthm}
\usepackage{qtree}
\usepackage{tree-dvips}
\usepackage{float}
\usepackage{hyperref}
\usepackage[nameinlink]{cleveref}
\usepackage{braket}
\usepackage{mathrsfs}
\usepackage{tikz}
\usepackage{verbatim}
\usepackage{tabularx}
\usepackage{booktabs}
\usepackage[labelfont=bf,format=plain,justification=raggedright,singlelinecheck=false]{caption}

\usepackage{algpseudocode}
\usepackage{algorithm}

\hypersetup{colorlinks={true},linkcolor={blue},citecolor=blue}

\newtheorem{theorem}{Theorem}
\newtheorem{lemma}{Lemma}
\newtheorem{claim}{Claim}
\newtheorem*{lemma*}{Lemma}
\newtheorem{fact}{Fact}

\newtheorem{proposition}{Proposition}
\newtheorem{observation}{Observation}
\newtheorem{definition}{Definition}

\newcommand{\adversary}{\mathcal{A}}

\newcommand{\ketbra}[2]{\ket{#1}\!\bra{#2}}
\renewcommand{\cal}[1]{\mathcal{#1}}

\newcommand{\N}{\mathbb{N}}

\newcommand{\Tr}{\mathrm{Tr}}

\newcommand{\reg}[1]{\mathsf{#1}}

\newcommand{\id}{I}

\DeclareMathOperator{\TD}{TD}
\newcommand{\TraceDist}[2]{{\TD\left(#1, #2\right)}}
\newcommand{\td}{\TD}
\newcommand{\norm}[1]{{\left\lVert#1\right\rVert}}

\newcommand{\VerCircuit}{\mathsf{V}}
\newcommand{\VerHat}{\widehat{\mathsf{V}}}
\newcommand{\syn}{\mathsf{Syn}}
\newcommand{\length}[1]{{\lvert {#1} \rvert}}

\newcommand{\keygen}{\mathsf{KeyGen}}
\newcommand{\mint}{\mathsf{Mint}}
\newcommand{\verify}{\mathsf{Ver}}

\newcommand{\oracle}{{\mathcal O}}

\newcommand{\qmalgs}{\left(\keygen^{\oracle},\mint^{\oracle},\verify^{\oracle}\right)}

\newcommand{\accept}{{\sf Accept}}
\newcommand{\reject}{{\sf Reject}}
\newcommand{\negl}{\mathsf{negl}}
\newcommand{\prob}{\mathsf{Pr}}
\newcommand{\bfx}{{\bf x}}

\newcommand{\oqm}{oracle-aided public key quantum money scheme }
\newcommand{\ro}{{\mathcal R}}

\newcommand{\pspace}{\mathsf{PSPACE}}

\newcommand{\statePSPACE}{{\mathsf{statePSPACE}}}

\newcommand{\good}{{\sf Good}}
\newcommand{\simulation}{\mathsf{Sim}}
\newcommand{\register}[1]{{\mathsf{#1}}}

\newcommand{\comp}{\mathsf{Comp}}
\newcommand{\decomp}{\mathsf{Decomp}}
\newcommand{\update}{\mathsf{Upd}}
\newcommand{\compV}[1]{{\widetilde{#1}}}
\newcommand{\decompV}[1]{{#1}}

\newcommand{\test}{\mathsf{Test}}

\usepackage[style=alphabetic,minalphanames=3,maxalphanames=4,maxnames=99,backref=true]{biblatex}
\title{On the (Im)plausibility of Public-Key Quantum Money\\ from Collision-Resistant Hash Functions}
\author{}
\author{Prabhanjan Ananth\thanks{\texttt{prabhanjan@cs.ucsb.edu}}\\ {\small UCSB} \and Zihan Hu\thanks{\texttt{huzh19@mails.tsinghua.edu.cn}}\\ {\small Tsinghua University} \and Henry Yuen\thanks{\texttt{hyuen@cs.columbia.edu}} \\ {\small Columbia University}} 

\date{}

\begin{document}

\maketitle

\begin{abstract}
\noindent Public-key quantum money is a cryptographic proposal for using highly entangled quantum states as currency that is publicly verifiable yet resistant to counterfeiting due to the laws of physics. Despite significant interest, constructing provably-secure public-key quantum money schemes based on standard cryptographic assumptions has remained an elusive goal. Even proposing plausibly-secure candidate schemes has been a challenge.

These difficulties call for a deeper and systematic study of the structure of public-key quantum money schemes and the assumptions they can be based on. Motivated by this, we present the first black-box separation of quantum money and cryptographic primitives. Specifically, we show that collision-resistant hash functions cannot be used as a black-box to construct public-key quantum money schemes where the banknote verification makes classical queries to the hash function. Our result involves a novel combination of state synthesis techniques from quantum complexity theory and simulation techniques, including Zhandry's compressed oracle technique. 

\end{abstract}

\section{Introduction}
\noindent Unclonable cryptography is an emerging area in quantum cryptography that leverages the no-cloning principle of quantum mechanics~\cite{WZ82,Dieks82} to achieve cryptographic primitives that are classically impossible. Over the years, many interesting unclonable primitives have been proposed and studied. These include quantum copy-protection~\cite{aaronson2009pkQM}, one-time programs~\cite{BGS13}, secure software leasing~\cite{ALP21}, unclonable encryption~\cite{BL20}, encryption with certified deletion~\cite{BI20}, encryption with unclonable decryption keys~\cite{GZ20,coladangelo2021hidden}, and tokenized signatures~\cite{BS16}. 
\par One of the oldest and (arguably) the most popular unclonable primitives is quantum money, which was first introduced in a seminal work by Wiesner~\cite{wiesner1983conjugate}. A quantum money scheme enables a bank to issue digital money represented as quantum states. Informally, the security guarantee states that it is computationally infeasible to produce counterfeit digital money states. That is, a malicious user, given one money state, cannot produce two money states that are both accepted by a pre-defined verification procedure. There are two notions we can consider here. The first notion is {\em private-key} quantum money, where the verification procedure is private. That is, in order to check whether a money state is valid, we need to submit the state to the bank which decides its validity. A more useful notion is {\em public-key} quantum money, where anyone can verify the validity of money states. While private-key money schemes have been extensively studied and numerous constructions, including information-theoretic ones, have been proposed, the same cannot be said for public-key quantum money schemes. 
\par Aaronson and Christiano~\cite{AC2013HiddenSubspace} first demonstrated the feasibility of information-theoretically secure public-key quantum money in the oracle model; meaning that all algorithms in the scheme (e.g., the minting and verification algorithms) query a black box oracle during their execution. 
In the standard (i.e., non-oracle) model, there are two types of constructions known for building quantum money: 
\begin{itemize}
    \item In the first category, we have constructions borrowing sophisticated tools from different areas of mathematics, such as knot theory~\cite{FGHLS12}, quaternion algebras~\cite{KSS21} and lattices~\cite{zhandry2021quantumlightning,KJS22}. The constructions in this category have been susceptible to cryptanalytic attacks as demonstrated by a couple of recent works~\cite{Rob21,bilyk2022cryptanalysis}. We are still in the nascent stages of understanding the security of these candidates.  
    \item In the second category, we have constructions based on well-studied (or perhaps \emph{better}-studied) cryptographic primitives. In this category, we have constructions~\cite{zhandry2021quantumlightning,Shm22,Shm22b} based on indistinguishability obfuscation, first initiated by Zhandry~\cite{zhandry2021quantumlightning}.   
\end{itemize}
\noindent We focus on the second category. Constructions from existing primitives, especially from those that can be based on well-studied  assumptions, would position public-key quantum money on firmer foundations. Unfortunately, existing constructions of indistinguishability obfuscation are either post-quantum {\bf in}secure~\cite{AJLMS21,JLS21,JLS22b} or are based on newly introduced cryptographic assumptions~\cite{GP21,BDGM21,WW21,DQVWW21} that have been subjected to cryptanalytic attacks~\cite{HJL21}. 
\par The goal of our work is to understand the feasibility of constructing public-key quantum money from fundamental and well-studied cryptographic primitives. We approach this direction via the lens of black-box separations. Black-box separations have been extensively studied in classical cryptography~\cite{Rud91,Simon98,GKMRV00,RTV04,BM09,DLMM11,GKLM12,BDV17}. We say that a primitive $A$ cannot be constructed from another primitive $B$ in a black-box manner if there exists a computational world (defined by an oracle) where $B$ exists but $A$ does not. Phrased another way, these separations rule out constructions of primitive $A$ where primitive $B$ is used in a black-box manner.
In this case, we say that there is a black-box separation between $A$ and $B$. Black-box separations have been essential in understanding the relationship between different cryptographic primitives. Perhaps surprisingly, they have also served as a guiding light in designing cryptographic constructions. One such example is the setting of identity-based encryption (IBE). A couple of works~\cite{BPRVW08,PRV12} demonstrated the difficulty of constructing IBE from the decisional Diffie Hellman (DDH) assumption using a black-box construction which prompted the work of~\cite{DG17} who used non-black-box techniques to construct IBE from DDH.

\subsection{Our Work}
\paragraph{Black-Box Separations for Unclonable Cryptography.} We initiate the study of black-box separations in unclonable cryptography. In this work, we study a black-box separation between public-key quantum money and (post-quantum secure) collision-resistant hash functions. To the best of our knowledge, our work takes the first step in ruling out certain approaches to constructing public-key quantum money from well-studied assumptions.

\paragraph{Model.} We first discuss the model in which we prove the black-box separation. We consider two oracles with the first being a random oracle ${\cal R}$ (i.e., a uniformly random function) and the second being a ${\sf PSPACE}$ oracle (i.e., one that can solve ${\sf PSPACE}$-complete problems). We investigate the feasibility of quantum money schemes and collision-resistant hash functions in the presence of ${\cal R}$ and ${\sf PSPACE}$. That is, all the algorithms of the quantum money schemes and also the adversarial entities are given access to the oracles ${\cal R}$ and ${\sf PSPACE}$. 

\par There are two ways we can model a quantum algorithm to have access to an oracle. The first is {\em classical access}, where the algorithms in the quantum money scheme can only make classical queries to the oracle; that is, each query to the oracle is measured in the computational basis before forwarding it to the oracle. If an algorithm ${\sf A}$ has classical access to an oracle, say ${\cal U}$, we denote this by ${\sf A}^{{\cal U}}$. The second is {\em quantum access}, where the algorithms can make superposition queries. That is, an algorithm can submit a state of the form $\sum_{x,y} \alpha_{x,y} \ket{x}\ket{y}$ to the oracle ${\cal O}$ and it receives back $\sum_{x,y} \alpha_{x,y} \ket{x}\ket{{\cal O}(x) \oplus y}$. If an algorithm ${\sf A}$ has quantum access to an oracle ${{\cal U}}$, we denote this by ${\sf A}^{\ket{{\cal U}}}$. 

\par Our ultimate goal is to obtain black-box separations in the quantum access model, where the algorithms in the quantum money scheme can query oracles in superposition. However, there are two major obstacles to achieving this. 

First, analyzing the quantum access model in quantum cryptography has been notoriously challenging. For example, it is not yet known how to generalize to the quantum access setting black-box separations between key agreement protocols -- a \emph{classical} cryptographic primitive -- and one-way functions~\cite{IR}. Attempts to tackle special cases have already encountered significant barriers~\cite{ACCFLM22}, and has connections to long-standing conjectures in quantum query complexity (like the Aaronson-Ambainis conjecture~\cite{aaronson2009need}).

Second, we have to contend with the difficulty that quantum money is an \emph{inherently quantum} cryptographic primitive. A black-box separation requires designing an adversary that can effectively clone a quantum banknote given a \emph{single} copy of it. Here one encounters problems of a uniquely quantum nature, such as the No-Cloning Theorem~\cite{WZ82,Dieks82} and the fact that measuring the banknote will in general disturb it. 

We present partial progress towards the ultimate goal stated above by simplifying the problem and focusing exclusively on this second obstacle: we prove black-box separations where the banknote verification algorithm in the quantum money schemes makes \emph{classical} queries to the random oracle $\cal{R}$ (but still can make quantum queries to the $\pspace$ oracle), and the minting algorithm may still make quantum queries to both $\cal{R}$ and $\pspace$ oracles.  As we will see, even this special case of quantum money schemes is already challenging and nontrivial to analyze. 
We believe that our techniques may ultimately be extendable to the general setting (if there indeed exists a black-box impossibility in the general setting!), where all algorithms can make quantum queries to all oracles, and furthermore help prove black-box separations of other quantum cryptographic primitives. 

\paragraph{Main Theorem.} We will state our theorem more formally. A quantum money scheme consists of three quantum polynomial-time (QPT) algorithms, namely $(\keygen,\mint,\verify)$, where $\keygen$ produces a public key-secret key pair, $\mint$ uses the secret key to produce money states and a serial number associated with money states and finally, $\verify$ determines the validity of money states using the public key. We consider {\em oracle-aided} quantum money schemes, where these algorithms have access to a random oracle ${\cal R}$ and a ${\sf PSPACE}$ oracle, defined above.

\begin{theorem}[Informal]
\label{thm:intro:main}
Any public-key quantum money scheme $( \keygen^{\ket{{\cal R}},\ket{{\sf PSPACE}}},\mint^{\ket{{\cal R}},\ket{{\sf PSPACE}}},\allowbreak \verify^{{\cal R},\ket{{\sf PSPACE}}})$ is insecure. 
\end{theorem}

\noindent By insecurity, we mean the following. There exists a quantum polynomial-time (QPT) adversary ${\cal A}$ such that ${\cal A}^{{\cal R},\ket{{\sf PSPACE}}}$, given a money state $({\sf pk},\rho_s,s)$, where ${\sf pk}$ is the public key and $s$ is a serial number, with non-negligible probability, can produce two (possibly entangled) states that both pass the verification checks with respect to the same serial number $s$. The probability is taken over the randomness of ${\cal R}$ and also over the randomness of $\keygen,\mint$ and ${\cal A}$. We note that only $\keygen$ and $\mint$ can have quantum access to ${\cal R}$, while $\verify$ only has classical access. On the other hand, we show that the adversary ${\cal A}$ only needs classical access to ${\cal R}$. 
\par Furthermore, we note that the random oracle $\cal{R}$ constitutes a collision-resistant hash function against QPT adversaries that can make queries to $(\cal{R},\ket{\pspace})$~\cite{zhandry2015note}. We note that $\cal{R}$ still remains collision-resistant even when the adversaries can make \emph{quantum} queries to $\cal{R}$, not just classical ones.  

\paragraph{Implications.}  Our main result rules out a class of public-key quantum money constructions that (a) base their security on collision-resistant hash functions, (b) use the hash functions in a black-box way, and (c) where the verification algorithm makes classical queries to the hash function. Clearly, it would be desirable to generalize the result to the case where the verification algorithm can make quantum queries to the hash function. However, there are some conceptual challenges to going beyond classical verification queries (which we discuss in more detail in~\Cref{sec:qqueries:verify}). 

The class of quantum money schemes in this hybrid classical-quantum query model is quite interesting on its own and a well-motivated setting. For example, in Zhandry's public-key quantum money scheme~\cite{zhandry2021quantumlightning}, the mint procedure only needs classical access to the underlying cryptographic primitives (when the component that uses cryptographic primitives is viewed as a black-box) while the verification procedure makes quantum queries. In the constructions of copy-protection due to Coladangelo et al.~\cite{coladangelo2021hidden,CMP20}, the copy-protection algorithm only makes classical queries to the cryptographic primitives in the case of~\cite{coladangelo2021hidden} and the random oracle in the case of~\cite{CMP20} whereas the evaluation algorithm in both constructions make quantum queries. Finally, in the construction of unclonable encryption in~\cite{AKLLZ22}, all the algorithms only make classical queries to the random oracle. Given these constructions, we believe it is important to understand what is feasible or impossible for unclonable cryptosystems in the hybrid classical-quantum query model.

Secondly, we believe that the hybrid classical-quantum query model is a useful testbed for developing techniques needed for black-box separations, and for gaining insight into the structure of unclonable cryptographic primitives. Even in this special case, there are a number of technical and conceptual challenges to overcome in order to get our black-box separation of \Cref{thm:intro:main}. We believe that the techniques developed in this paper will be a useful starting point for future work in black-box separations in unclonable cryptography.

\paragraph{Other Separations.} As a corollary of our main result, we obtain black-box separations between public-key quantum money and many other well-studied cryptographic primitives such as one-way functions, private-key encryption and digital signatures.
\par Our result also gives a separation between public-key quantum money and collapsing hash functions in the same setting as above; that is, when $\verify$ makes classical queries to ${\cal R}$. This follows from a result due to Unruh~\cite{Unr16} who showed that random oracles are collapsing. Collapsing hash functions are the quantum analogue of collision-resistant hash functions. Informally speaking, a hash function is collapsing if an adversary cannot distinguish a uniform superposition of inputs, say $\ket{\psi}$, mapping to a random output $y$ versus a computational basis state obtained by measuring $\ket{\psi}$ in the computational basis. Zhandry~\cite{zhandry2021quantumlightning} showed that hash functions that are collision-resistant but not collapsing imply the existence of public-key quantum money. Thus our result rules out a class of constructions of quantum money from collapsing functions, improving our understanding of the relationship between them.

\paragraph{Acknowledgments.} We thank anonymous conference referees, Qipeng Liu, Yao Ching Hsieh, and Xingjian Li for their helpful comments. HY is supported by AFOSR award
FA9550-21-1-0040 and NSF CAREER award CCF-2144219.

\section{Our Techniques in a Nutshell}
We present a high-level overview of the techniques involved in proving~\Cref{thm:intro:main}. But first, we will briefly discuss the correctness guarantee of {\em oracle-aided} public-key quantum money schemes. 

\paragraph{Reusability.} In a quantum money scheme $(\keygen,\mint,\verify)$, we require that $\verify$ accepts a state and a serial number produced by $\mint$ with overwhelming probability. However, for all we know, $\verify$, during the verification process, might destroy the state. In general, a more useful correctness definition is reusability, which states that a money state can be repeatedly verified without losing its validity. In general, one can show that the gentle measurement lemma~\cite{Win99} does prove that correctness implies reusability. However, as observed in~\cite{AK22}, this is not the case when $\verify$ has classical access to an oracle. Specifically, $\verify$ has classical access to ${\cal R}$. Hence, we need to explicitly define reusability in this setting. Roughly speaking, we require the following: suppose we execute $\verify$ on a money state $\rho^{(0)}$ produced using $\mint$ and the verification algorithm accepts with probability $\delta$. The residual state is (possibly) a different state $\rho^{(1)}$ which when executed upon by $\verify$ is also accepted with probability close to $\delta$. In fact, even if we run the verification process polynomially many times, the state obtained at the end of the process should still be accepted by $\verify$ with probability close to $\delta$. 

\subsection{Warmup: Insecurity when ${\cal R}$ is absent}\label{sec:IntuitionOfSyn} 
\label{sec:overview:warmup}
Towards developing techniques to prove~\Cref{thm:intro:main}, let us first tackle a simpler statement. Suppose we have a secure public-key quantum money scheme $(\keygen,\mint,\verify)$. This means that any QPT adversary cannot break the security of this scheme. But what about oracle-aided adversaries? In more detail, we ask the following question: {\em Does there exist a QPT algorithm, given quantum access to a ${\sf PSPACE}$ oracle, that violates the security of  $(\keygen,\mint,\verify)$?} That is, given $(s,\rho_s)$, where $s$ is a serial number and $\rho_s$ is a valid banknote produced by $\mint$, it should be able to produce two states, with respect to the same serial number $s$, that are both accepted by the verifier. 
\par Even this seemingly simple question seems challenging! Let us understand why. Classical cryptographic primitives (even post-quantum secure ones) such as encryption schemes or digital signatures can be broken by efficient adversaries who have access to even ${\sf NP}$ oracles. This follows from the fact that we can efficiently reduce the problem of breaking the scheme to the problem of determining membership in a language. For instance, in order to succeed in breaking an encryption scheme, the adversary has to decide whether the instance $({\sf pk},{\sf ct},{\sf m}) \in L$, where ${\sf pk}$ is a public key, ${\sf ct}$ is a ciphertext, ${\sf m}$ is a message and $L$ consists of instances of the form $({\sf pk},{\sf ct},{\sf m})$, where ${\sf ct}$ is an encryption of ${\sf m}$ with respect to the public key ${\sf pk}$. Implicitly, we are using the fact that ${\sf pk},{\sf ct}, {\sf m}$ are binary strings. Emulating a similar approach in the case of quantum money would result in {\em quantum} instances and it is not clear how to leverage ${\sf PSPACE}$, or more generally a decider for any language, to complete the reduction. 

\paragraph{Synthesizing Witness States.} Towards addressing the above question, we 
reduce the task of breaking the security of the quantum money scheme using ${\sf PSPACE}$ to the task of finding states accepted by the verifier in quantum polynomial space. This reduction is enabled by the following observation, due to Rosenthal and Yuen~\cite{IPforStates}: a set of pure states computable by a quantum polynomial space algorithm (which may in general include intermediate measurements) can be synthesized by a QPT algorithm with quantum access to a ${\sf PSPACE}$ oracle. Implicit in the result of~\cite{IPforStates} is the following important point: in order to synthesize the state using the ${\sf PSPACE}$ oracle, it is important that we know the entire description of the quantum polynomial space algorithm generating the pure states. 
\par In more detail, we show the following statement: for every\footnote{Technically, we show a weaker statement which holds for an overwhelming fraction of $({\sf pk},s)$.} verification key ${\sf pk}$, serial number $s$, there exists a pure state\footnote{Technically, we require that the reduced density matrix of $\rho_{{\sf pk},s}$ is accepted by $\verify$.} $\rho_{{\sf pk},s}$ that is accepted by $\verify({\sf pk},s,\cdot)$ with non-negligible probability and moreover, can be generated by a quantum polynomial space algorithm.

\par The first attempt is to follow the classical brute-force search algorithm. Namely, we repeat the following for exponential times: guess a quantum state $\rho$ uniformly at random and if $\rho$ is accepted by $\verify({\sf pk},s,\cdot)$ with non-negligible probability, output $\rho$ and terminate. (Output an arbitrary state if we run out of times.) However, there are two problems with this attempt. Firstly, in general, it's not clear how to calculate the acceptance probability of $\verify({\sf pk},s,\rho)$ in polynomial space ($\rho$ needs exponential bits to represent). Secondly, $\rho$ might be destroyed when we calculate the acceptance probability. 

\par To fix the first problem, we note that an estimation of the acceptance probability is already good enough and it can be done by using a method introduced by Marriott and Watrous~\cite{marriott} (called MW technique). The MW technique allows us to efficiently estimate the acceptance probability of a verification algorithm on a state with only one copy of that state. Furthermore, it does not disturb the state too much in the sense that the expected acceptance probability of the residual state does not decay too significantly, which fixes the second problem. 

\par This brings us to our second attempt. We repeat the following process for exponentially many times: apply MW technique on a maximally mixed state and if the estimated acceptance probability happens to be non-negligible, output the residual state and terminate. (Output an arbitrary state if all the iterations fail.)

\par As the MW technique is efficient, this algorithm only uses polynomial space. Furthermore, intuitively we can get a state that is accepted by $\verify$ with non-negligible acceptance probability given the fact that such a state exists. Because if such state exists, by a simple convexity argument, we can assume that without loss of generality that it's pure. Maximally mixed state can be treated as a uniform mixture of a basis containing that pure state. Thus roughly speaking, we start from that pure state with inverse exponential probability, so we can find it in exponentially many iterations with overwhelming probability. This attempt almost succeeds except that it outputs a \emph{mixed} state in general, but the known approach in \cite{IPforStates} can only deal with pure states. There are two reasons for this. Firstly, we start with a maximally mixed state and secondly, MW technique involves intermediate measurements.

\par Our final attempt makes the following minor changes compared to the second attempt. To fix the first issue, it starts with a maximally entangled state (instead of maximally mixed state) and only operates on half of it. To fix the second issue, it runs the MW process coherently by deferring all the intermediate measurements. Then we will end up with a pure state whose reduced density matrix is the same as the output state of the second attempt.

\subsection{Insecurity in the presence of ${\cal R}$}\label{sec:IntuitionWithR}
So far, we considered the task of violating the security of a quantum money scheme where the algorithms did not have access to any oracle. Let us go back to the oracle-aided quantum money schemes, where, all the algorithms (honest and adversarial) have access to ${\cal R}$ and $\ket{{\sf PSPACE}}$. Our goal is to construct an adversary that violates the security of quantum money schemes. {\em But didn't we just solve this problem?} Recall that when invoking~\cite{IPforStates}, it was crucial that we knew the entire description of the polynomial space algorithm in order to synthesize the state. However, when we are considering oracle-aided verification algorithms, denoted by $\verify^{{\cal R},\ket{\sf PSPACE}}$, we don't have the full description of\footnote{The fact that we don't have the description of ${\cal R}$ is the problem here.} $\verify^{{\cal R},\ket{\sf PSPACE}}$. Thus, we cannot carry out the synthesizing process. 
\par A naive approach is to sample our own oracle ${\cal R}'$ and synthesize the state with respect to $\verify^{{\cal R}',\ket{\sf PSPACE}}$. However, this does not help. Firstly, there is no guarantee that $\verify^{{\cal R}',\ket{\sf PSPACE}}({\sf pk},s,\cdot)$ accepts any state with high enough probability. Without this guarantee, the synthesizing process does not work. For now, let us take for granted that there does exist some witness state $\sigma'$ that is accepted by $\verify^{{\cal R}',\ket{\sf PSPACE}}({\sf pk},s,\cdot)$ with high enough probability. However, there is no guarantee that $\sigma'$ is going to be accepted by $\verify^{{\cal R},\ket{\sf PSPACE}}({\sf pk},s,\cdot)$ with better than negligible probability.  
\par Towards addressing these hurdles, we first focus on a simple case when $\keygen$ and $\mint$ make classical queries to ${\cal R}$ and we later, focus on the quantum queries case.

\subsubsection{$\keygen$ and $\mint$: Classical Queries to ${\cal R}$}\label{sec:overview:ClassicalR} 

\paragraph{Compiling out ${\cal R}$.} Suppose we can magically find a database $D$, using only polynomially many queries to ${\cal R}$, such that all the query-answer pairs made by $\verify$ to ${\cal R}$ are contained in $D$. In this case, there is a QPT adversary ${\cal A}^{{\cal R},\ket{{\sf PSPACE}}}$ that given $({\sf pk},s,\rho_s)$, can find two states $(s,\sigma'_s)$ and $(s,\sigma''_s)$ such that $\verify^{{\cal R},\ket{\sf PSPACE}}$ accepts both the states. ${\cal A}$ does the following: it first finds the database $D$ and constructs another circuit $\verify^{D,\ket{{\sf PSPACE}}}$ such that  $\verify^{D,\ket{\sf PSPACE}}$ runs $\verify^{{\cal R},\ket{{\sf PSPACE}}}$ and when $\verify^{{\cal R},\ket{{\sf PSPACE}}}$ makes a query to ${\cal R}$, the query is answered by $D$. Then, ${\cal A}$ synthesizes two states $(s,\sigma'_s)$ and $(s,\sigma''_s)$, using ${\sf PSPACE}$, such that both the states are accepted by $\verify^{D,\ket{{\sf PSPACE}}}$. By definition of the database $D$, these two states are also accepted by $\verify^{{\cal R},\ket{\sf PSPACE}}$. 
\par Of course, it is wishful for us to hope that we can find a database $D$ by making only polynomially many queries to ${\cal R}$ that is perfectly consistent with the queries made by $\verify$. Instead, we hope to recover a good enough database $D$. In more detail, we aim to recover a database $D$ that captures all the relevant queries made by $\keygen$ and $\mint$. 
\par Let $D_{\keygen}$ and $D_{\mint}$ be the collection of query-answer pairs made by ${\keygen}$ and ${\mint}$ respectively. A query made by $\verify$ is called {\em bad} if this query is in $D_{\keygen} \cup D_{\mint}$ and moreover, this query was not recorded in $D$. If $\verify$ makes {\em bad} queries then the answers returned will likely be inconsistent with ${\cal R}$ and thus, there is no guarantee that $\verify$ will work. Our hope is that the probability of $\verify$ making bad queries is upper bounded by an inverse polynomial. 
\par Once we have such a database $D$, by a similar argument, we can conclude that the states synthesized using $\verify^{D,\ket{{\sf PSPACE}}}$ are also accepted by $\verify^{{\cal R},\ket{{\sf PSPACE}}}$. 
\par {\em But how do we recover this database $D$?} To see how, we will first focus on a simple case before dealing with the general case.  

\paragraph{State-independent database simulation.} Note that the queries made by $\verify$ could potentially depend on its input state. For now, we will assume that the distribution of queries made by $\verify$ is independent of the input state. We will deal with the state-dependent query distributions later. 
\par The first attempt to generate $D$ would be to rely upon techniques introduced by Canetti, Kalai and Paneth~\cite{RemoveROfromObf} who, in a different context -- that of proving impossibility of obfuscation in the random oracle model -- showed how to generate a database that is sufficient to simulate the queries made by the evaluation algorithm. Suppose $(s,\rho_s)$ is the state generated by $\mint$. Then, run $\verify^{{\cal R},\ket{{\sf PSPACE}}}(pk, s,\rho_s)$ a fixed polynomially many times, referred to as {\em test} executions, by querying ${\cal R}$. In each execution of $\verify^{{\cal R},\ket{{\sf PSPACE}}}$, record all the queries made by $\verify$ along with their answers. The union of queries made in all the executions of $\verify$ will be assigned the database $D$. In the context of  obfuscation for classical circuits,~\cite{RemoveROfromObf} argue that, except with inverse polynomial probability, the queries made by the evaluation algorithm can be successfully simulated by $D$. This argument is shown by proving an upper bound on the probability that the evaluation algorithm makes bad queries. 

\par A similar analysis can also be made in our context to argue that $D$ suffices for successful simulation. That is, we can argue that the state we obtain after all the executions of $\verify$ (which could be very different from the state we started off with) can be successfully simulated using $D$. However, it is crucial for our analysis to go through that $D_{\verify}$ (the query-answer pairs made during $\verify$) is {\em independent} of the state input to $\verify$.

\paragraph{State-dependent database simulation.} For all we know, $D_{\verify} $ could indeed depend on the input state. In this case, we can no longer appeal to the argument of~\cite{RemoveROfromObf}. At a high level, the reason is due to the fact that after each execution of $\verify$, the money state could potentially change and this would affect the distribution of $D_{\verify}$ in the further executions of $\verify$ in such a way that the execution of $\verify$ on the final state (which could be different from the input state in the first execution of $\verify$) cannot be simulated using the database $D$. 

Instead, we will rely upon a technique due to~\cite{AK22}, who studied a similar problem in the context of copy-protection. They showed that by randomizing the number of executions, one can argue that the execution of $\verify$ on the state obtained after all the test executions can be successfully simulated using $D$, except with  inverse polynomial probability. That is, suppose the initial state is $(s,\rho_s^{(0)})$ and after running $\verify^{{\cal R},\ket{{\sf PSPACE}}}$, $t$ number of times where $t \xleftarrow{\$} \{0, 1, \cdots, T\}$, let the resulting state be $(s,\rho_s^{(t)})$. Let $D$ be as defined before.  Then, we have the guarantee that $\verify$ accepts $(s,\rho_s^{(t)})$,  except with inverse polynomial probability, even when ${\cal R}$ is simulated using $D$. This is because the sum of the number of bad queries we encounter during $T + 1$ verifications is bounded by $\lvert D_{\keygen} \cup D_{\mint} \rvert$. Then there are only at most $\epsilon^{-1}\lvert D_{\keygen} \cup D_{\mint} \rvert$ points $t' \in \{0, 1, \cdots, T\}$ such that the probability of making a bad query during the $(t' + 1)^{th}$ verification is at least $\epsilon$. So when $T$ is large enough, there is a good chance that we choose $t$ such that the probability of making a bad query during the next verification (i.e. the $(t + 1)^{th}$ verification if you count from the beginning) is small, in which case $D$ can simulate $\cal R$ well.
\par This suggests the following attack on the quantum money scheme. On input a money state $(s,\rho_s)$, do the following:
\begin{itemize}
    \item Run $\verify^{{\cal R},\ket{{\sf PSPACE}}}$, $t$ times, also referred to as test executions. The number of times we need to run $\verify^{{\cal R},\ket{{\sf PSPACE}}}$, namely $t$, is randomized as per~\cite{AK22}. Let $D$ be the set of query-answer pairs made by $\verify$ to ${\cal R}$ during the test executions. Denote $\rho_s^{(t)}$ to be the state obtained after $t$ executions of $\verify$. 
    \item Let $\verify^{D,\ket{{\sf PSPACE}}}$ be the verification circuit as defined earlier. 
    \item Using quantum access to ${\sf PSPACE}$, synthesize two states $(s,\sigma'_s)$ and $(s,\sigma''_s)$, as per~\Cref{sec:overview:warmup}, such that both the states are accepted by $\verify^{D,\ket{{\sf PSPACE}}}$.
    \item Output $(s,\sigma'_s)$ and $(s,\sigma''_s)$. 
\end{itemize}
\noindent From the witness synthesis method, we have the guarantee that  $(s,\sigma'_s)$ and $(s,\sigma''_s)$ are both accepted by $\verify^{D,\ket{{\sf PSPACE}}}$. However, this is not sufficient to prove that the above attack works. Remember that the adversary is supposed to output two states that are both accepted by $\verify^{{\cal R},\ket{{\sf PSPACE}}}$. Unfortunately, there is no guarantee that $\verify^{{\cal R},\ket{{\sf PSPACE}}}$ accepts these two states. Indeed, both $\sigma'_s$ and $\sigma''_s$ could be quite different from $\rho_s^{(t)}$. Hence, the above attack does not work.  

\paragraph{Every mistake we make is progress.} Let us understand why the above attack does not work. Note that as long as $\verify$ does not make any bad query (i.e., a query in $D_{\keygen} \cup D_{\mint}$ but not contained in $D$), it cannot distinguish whether its queries are being simulated by $D$ or ${\cal R}$. However, when $\verify$ is executed on $(s,\sigma'_s)$ or $(s,\sigma''_s)$, we can no longer upper bound the probability that $\verify$ will not make any bad queries. 
\par We modify the above approach as follows: whenever $\verify$ makes bad queries, we can update the database $D$ to contain the bad queries along with the correct answers (i.e., answers generated using ${\cal R}$). Once $D$ is updated, we can synthesize two new states using $\verify^{D,\ket{{\sf PSPACE}}}$. We repeat this process until we have synthesized two states that are accepted by $\verify^{{\cal R},\ket{{\sf PSPACE}}}$. 
\par {\em Is there any guarantee that this process will stop?} Our key insight is that whenever we make a mistake and synthesize states that are not accepted by $\verify^{{\cal R},\ket{{\sf PSPACE}}}$ then we necessarily learn a new query in $D_{\keygen} \cup D_{\mint}$ that is not contained in $D$. Thus, with each mistake, we make progress. Since there are only a polynomial number of queries in $D_{\keygen} \cup D_{\mint}$, we will ultimately end up synthesizing two states that are accepted by $\verify^{{\cal R},\ket{{\sf PSPACE}}}$. 

\paragraph{Our Attack.} With this modification, we have the following attack. On input a money state $(s,\rho_s)$, do the following:
\begin{itemize}
    \item {\em Test phase}: Run $\verify^{{\cal R},\ket{{\sf PSPACE}}}$, $t$ times, also referred to as test executions. The number of times we need to run $\verify^{{\cal R},\ket{{\sf PSPACE}}}$, namely $t$, is randomized as per~\cite{AK22}. Let $D$ be the set of query-answer pairs made by $\verify$ to ${\cal R}$ during the test executions. Denote $\rho_s^{(t)}$ to be the state obtained after $t$ executions of $\verify$. 
    \item {\em Update phase}: Repeat the following polynomially many times. Let $\verify^{D,\ket{{\sf PSPACE}}}$ be the verification circuit as defined earlier. Using quantum access to ${\sf PSPACE}$, synthesize a state $(s,\sigma_s)$ as per~\Cref{sec:overview:warmup}, such that the state is accepted by $\verify^{D,\ket{{\sf PSPACE}}}$. Run $\verify^{{\cal R},\ket{{\sf PSPACE}}}$ on this state and include any new queries made by $\verify$ to ${\cal R}$ in $D$. 
    \item Let $D_1,\ldots,D_{{\rm poly}}$ be the databases obtained after every execution during the update phase.  
    \item Using quantum access to ${\sf PSPACE}$, synthesize two states $(s,\sigma'_s)$ and $(s,\sigma''_s)$ such that both the states are accepted by $\verify^{D_j,\ket{{\sf PSPACE}}}$ for some randomly chosen $j$.
    \item Output $(s,\sigma'_s)$ and $(s,\sigma''_s)$. 
\end{itemize}
\noindent In the technical sections, we analyze the above attack and prove that it works.

\subsubsection{$\keygen$ and $\mint$: Quantum Queries to ${\cal R}$}

The important point to note here is the form of our aforementioned attacker. It only takes advantage of the fact that $\verify$ makes classical queries to $\ro$. When $\keygen$ and $\mint$ make quantum queries to $\ro$ while $\verify$ makes classical queries to $\ro$, we can still run the attacker. What is left is to show that the same attacker works even when $\keygen$ and $\mint$ make quantum queries to $\ro$.

The main difficulty in carrying out the intuitions in \Cref{sec:overview:ClassicalR} to the case where $\keygen$ and $\mint$ make quantum queries to $\ro$ is that it's difficult to define analogue of $D_{\keygen}$ and $D_{\mint}$. To give a flavour of the difficulty, let's first consider two naive attempts. 

The first attempt is to define $D_{\keygen}$ and $D_{\mint}$ to be those query-answer pairs asked (with non-zero amplitudes) during $\keygen$ and $\mint$. However, this attempt suffers from the problem that in this way, $D_{\keygen} \cup D_{\mint}$ can have exponential elements. So even if each time we can make progress in the sense that we recover some new elements in $D_{\keygen} \cup D_{\mint}$, there is no guarantee that the update phase will terminate in polynomial time.

The second attempt is to only include queries that are asked ``heavily'' during $\keygen$ and $\mint$. To be more specific, let $D_{\keygen}$ and $D_{\mint}$ be query-answer pairs asked with inverse polynomial squared amplitudes during $\keygen$ and $\mint$. However, with this plausible definition, the claim does not hold that whenever the acceptance probability of $\verify^{D,\ket{{\sf PSPACE}}}$ is far from that of $\verify^{\ro, \ket{\pspace}}$, then we can recover a query in $D_{\keygen} \cup D_{\mint} - D$, which is a crucial idea underlying our intuitions in \Cref{sec:overview:ClassicalR}. Let us understand why this claim is not true if we adopt this definition of $D_{\keygen}$ and $D_{\mint}$. 
\par Consider the following contrived counterexample $(\keygen^{\ket{\ro}, \ket{\pspace}}, \mint^{\ket{\ro}, \ket{\pspace}}, \verify^{\ro, \ket{\pspace}})$. Suppose there exists a quantum money scheme $(\keygen'^{\ket{\ro}, \ket{\pspace}},\mint'^{\ket{\ro}, \ket{\pspace}},\verify'^{\ro, \ket{\pspace}})$. We modify this scheme into $(\keygen^{\ket{\ro}, \ket{\pspace}}, \mint^{\ket{\ro}, \ket{\pspace}}, \verify^{\ro, \ket{\pspace}})$ as follows: 
\begin{itemize}
    \item $\keygen^{\ket{\ro}, \ket{\pspace}}(1^n)$: outputs the secret key-public key pair of $\keygen'^{\ket{\ro}, \ket{\pspace}}$.
    \item $\mint^{\ket{\ro}, \ket{\pspace}}(sk)$: takes as input $sk$, makes quantum query to $\ro$ on state $\frac{1}{\sqrt{2^n}}\sum_{s = 0}^{2^n - 1}\ket{s}\ket{0}$ to get a state $\frac{1}{\sqrt{2^n}}\sum_{s = 0}^{2^n - 1}\ket{s}\ket{\ro(s)}$, and then measures the first register to get a value $s$. It also runs $\mint'^{\ket{\ro}, \ket{\pspace}}(sk)$ to get a serial number $s'$ along with a state $\rho_{s'}$. It outputs $(s,s')$ as the serial number and $(\ro(s),\rho_{s'})$ as the banknote. 
    \item $\verify^{\ro, \ket{\pspace}}(pk,((s, s'),(h, \rho_{s'})))$: takes as input $pk$ and an alleged banknote $((s,s'), (h,\rho_{s'}))$, makes classical query to $\ro$ on the input $s$ to get $R(s)$ and checks if $h=\ro(s)$. It also checks if $(s', \rho_{s'})$ is a valid money state with respect to $\verify'^{\ro, \ket{\pspace}}$. Accepts if and only if both the checks pass.
\end{itemize}

In the above counterexample, it is possible that there is no query-answer pair that is asked with inverse polynomial squared amplitudes and thus $D_{\keygen} \cup D_{\mint} = \emptyset$. At the beginning $D = \emptyset$ because we have not started to record query-answer pairs. In this case, the acceptance probability of $\verify^{D,\ket{{\sf PSPACE}}}(pk,((s, s'),(\ro(s), \rho_{s'})))$ is smaller than or equal to $\frac{1}{2^m}$ where $m$ is the output length of $\ro$ on input length $n$ while the acceptance probability of $\verify^{\ro,\ket{{\sf PSPACE}}}(pk,((s, s'),(\ro(s), \rho_{s'})))$ is 1 if $(\keygen'^{\ket{\ro}, \ket{\pspace}},\mint'^{\ket{\ro}, \ket{\pspace}},\verify'^{\ro, \ket{\pspace}})$ is (perfectly) correct. However, it's impossible to recover a new query in $D_{\keygen} \cup D_{\mint} - D$ because it's empty.

\paragraph{Purified View.} Our insight is to consider an alternate world called the {\em purified view}. In this alternate world, we run everything coherently; in more detail, we consider a uniform superposition of $\ro$, run $\mint$, $\keygen$ and even the attacker coherently (i.e., no intermediate measurements). If the attacker is successful in this alternate world then he is also successful in the real world where $\ro$ and the queries made by $\verify$ to $\ro$ are measured. We then employ the the compressed oracle technique by Zhandry~\cite{zhandryCompressedOracle} to coherently recover the database of query-answer pairs recorded during $\keygen,\mint$ and relate this with the database recorded during $\verify$. Using an involved analysis, we then show many of the insights from the case when $\keygen,\mint$
make classical queries to $\ro$ can be translated to the quantum query setting. 

\subsubsection{Challenges To Handling Quantum Verification Queries}
\label{sec:qqueries:verify}
It is natural to wonder whether we can similarly use the compressed oracle technique to handle quantum queries made by $\verify$. Unfortunately, there are inherent limitations. Recall that in our attack, the adversary records the verifier's classical query-answer pairs in a database, uses this to produce a \emph{classical} description of a verification circuit (that does not make any queries to the random oracle), and submits the circuit description to a $\pspace$ oracle in order to synthesize a money state. If the verifier instead makes quantum queries, then a natural idea is to use Zhandry's compressed oracle technique to record the quantum queries. However, there are two conceptual challenges to implementing this idea. 

First, in the compressed oracle technique, the queries are being recorded by the \emph{oracle} itself in a
``database register'', and not the adversary in the cryptosystem. In our setting, we are trying to construct an adversary to record the queries, but it does not have access to the oracle's database register. In general, any attempts by the adversary to get some information about the query positions of $\verify$ could potentially disturb the intermediate states of the $\verify$ algorithm; it is then unclear how to use the original guarantees of $\verify$. Another way of saying this is that Zhandry's compressed oracle technique is typically used in the \emph{security analysis} to show limits on the adversary's ability to break some cryptosystem. In our case, we want to use some kind of quantum recording technique in the adversary's \emph{attack}. 
\par Secondly, the natural approach to using the $\pspace$ oracle is to leverage it to synthesize alleged banknotes. However, since the $\pspace$ oracle is a classical function (which may be accessed in superposition), it requires polynomial-length classical strings as input. In our approach, the adversary submits a classical description of a verification circuit with query/answer pairs hardcoded inside. On the other hand if $\verify$ makes quantum queries, it may query exponentially many positions of the random oracle $\cal{R}$ in superposition, and it is unclear how to ``squeeze'' the relevant information about the queries into a polynomial-sized classical string that could be utilized by the $\pspace$ oracle. 

This suggests that we may need a fundamentally new approach to recording quantum queries in order to handle the case when the verification algorithm makes quantum queries.

\subsection{Related Work}

\paragraph{Quantum Money.} The notion of quantum money was first conceived in the paper by Wiesner~\cite{wiesner1983conjugate}. In the same work, a construction of private-key quantum money was proposed. Wiesner's construction has been well studied and its limitations~\cite{Lut10} and security guarantees~\cite{MVW12} have been well understood. Other constructions of private-key quantum money have also been studied. Ji, Liu and Song~\cite{JLS18} construct private-key quantum money from pseudorandom quantum states. Radian and Sattath~\cite{RS22} construct private-key quantum money with classical bank from quantum hardness of learning with errors. 
\par With regards to public-key quantum money, Aaronson and Christiano~\cite{AC2013HiddenSubspace} present a construction of public-key quantum money in the oracle model. Zhandry~\cite{zhandry2021quantumlightning} instantiated this oracle and showed how to construct public-key quantum money based on the existence of post-quantum indistinguishability obfuscation (iO)~\cite{BGIRSVY01}. Recently, Shmueli~\cite{Shm22} showed how to achieve public-key quantum money with classical bank, assuming post-quantum iO and quantum hardness of learning with errors. Constructions~\cite{FGHLS12,KSS21,KJS22} of public-key quantum money from newer assumptions have also been explored although they have been susceptible to quantum attacks~\cite{Rob21,bilyk2022cryptanalysis}.

\paragraph{Black-box Separations in Quantum Cryptography.} So far, most of the existing black-box separations in quantum cryptography have focused on extending black-box separations for classical cryptographic primitives to the quantum setting. Hosoyamada and Yamakawa~\cite{HY20} extend the black-box separation between collision-resistant hash functions and one-way functions~\cite{Simon98} to the quantum setting. Austrin, Chung, Chung, Fu, Lin and Mahmoody~\cite{ACCFLM22} showed a black-box separation between key agreement and one-way functions in the setting when the honest parties can perform quantum computation but only have access to classical communication. Cao and Xue~\cite{CX21} extended classical black-box separations between one-way permutations and one-way functions to the quantum setting. 

\section{Preliminaries}
\label{sec:prelims}

For a string $x$, let $\length{x}$ denote its length. Let $[n]$ denote the set $\{0, 1, \cdots, n - 1\}$ for any positive integer $n$. Define the symmetric difference of two sets $X$ and $Y$ to be the set of elements contained in exactly one of $X$ and $Y$, i.e. $X \Delta Y = (X - Y) \cup (Y - X)$. 

\subsection{Quantum States, Algorithms, and Oracles}

A \emph{register} $\reg{R}$ is a  finite-dimensional complex Hilbert space. If $\reg{A}, \reg{B}, \reg{C}$ are registers, for example, then the concatenation $\reg{A} \reg{B} \reg{C}$ denotes the tensor product of the associated Hilbert spaces. For a linear transformation $L$ and register $\reg R$, we sometimes write $L_{\reg R}$ to indicate that $L$ acts on $\reg R$, and similarly we sometimes write $\rho_{\reg R}$ to indicate that a state $\rho$ is in the register $\reg R$. We write $\Tr(\cdot)$ to denote trace, and $\Tr_{\reg R}(\cdot)$ to denote the partial trace over a register $\reg R$.

For a pure state $\ket\varphi$, we write $\varphi$ to denote the density matrix $\ketbra{\varphi}{\varphi}$. Let $\id$ denote the identity matrix. Let $\td(\rho,\sigma)$ denote the trace distance between two density matrices $\rho,\sigma$.

For a pure state $\ket{\varphi} = \sum_i a_i \ket{i}$ written in computational basis, we write $\ket{\overline{\varphi}} = \sum_i \overline{a_i}\ket{i}$ to denote the conjugate of $\ket{\varphi}$ where $\overline{a_i}$ is the complex conjugate of the complex number $a_i$. The following observation shows that what the maximally entangled state looks like in other basis.

\begin{lemma}\label{lemma:MaximallyEntangledState}
For two registers $\reg{A}$ and $\reg{B}$ of the same dimension $N$, let $\left(\ket{i}\right)_{i = 1}^N$ be the computational basis and $\left(\ket{v_i}\right)_{i = 1}^N$ be an arbitrary basis. Then $\frac{1}{\sqrt N}\sum_{i = 1}^N \ket{i}_{\reg{A}}\ket{i}_{\reg{B}} = \frac{1}{\sqrt N}\sum_{i = 1}^{N} \ket{v_i}_{\reg{A}}\ket{\overline{v_i}}_{\reg{B}}.$ 
\end{lemma}

\begin{proof}
It's easy to show that $\left(\ket{\overline{v_i}}\right)_{i = 1}^N$ also forms a basis. Suppose $\ket{v_j} = \sum_{i = 1}^N a_{j, i}\ket{i}$. Then
\begin{align*}
    \frac{1}{\sqrt N}\sum_{i = 1}^N\ket{i}_{\reg{A}} \ket{i}_{\reg{B}} &= \sum_{j = 1}^N \ketbra{v_j}{v_j}_{\reg{A}} \sum_{j' = 1}^N \ketbra{\overline{v_{j'}}}{\overline{v_{j'}}}_{\reg{B}}\frac{1}{\sqrt N}\sum_{i = 1}^N \ket{i}_{\reg{A}}\ket{i}_{\reg{B}} \\
    &= \frac{1}{\sqrt N}\sum_{j = 1}^N\sum_{j' = 1}^N \sum_{i = 1}^N \overline{a_{j, i}}a_{j', i}\ket{v_j}_{\reg A}\ket{\overline{v_{j'}}}_{\reg B}\\
    &= \frac{1}{\sqrt N}\sum_{j = 1}^N\sum_{j' = 1}^N \braket{v_j \vert v_{j'}}\ket{v_j}_{\reg A}\ket{\overline{v_{j'}}}_{\reg B}\\
    &= \frac{1}{\sqrt N}\sum_{j = 1}^N \ket{v_j}_{\reg A}\ket{\overline{v_{j}}}_{\reg B}
\end{align*}
where we use the fact that $\sum_{j = 1}^N \ketbra{v_j}{v_j}_{\reg{A}}$ and $\sum_{j' = 1}^N \ketbra{\overline{v_{j'}}}{\overline{v_{j'}}}_{\reg{B}}$ are  identity.
\end{proof}

\paragraph{Quantum Circuits}
We specify the model of quantum circuits that we work with in this paper. For convenience we fix the universal gate set $\{ H, \mathit{CNOT}, T \}$ \cite[Chapter 4]{nielsen2000quantum} (although our results hold for any universal gate set consisting of gates with algebraic entries). Quantum circuits can include unitary gates from the aforementioned universal gate set, as well as non-unitary gates that (a) introduce new qubits initialized in the zero state, (b) trace them out, or (c) measure them in the standard basis. We say that a circuit uses space $s$ if the total number of qubits involved at any time step of the computation is at most $s$. The description of a circuit is a sequence of gates (unitary or non-unitary) along with a specification of which qubits they act on.

We call a sequence of quantum circuits $C = (C_x)_{x \in \{0, 1\}^*}$ a \emph{quantum algorithm}. We say that $C$ is \emph{polynomial-time} if there exists a polynomial $p$ such that $C_x$ has size at most $p(\lvert x \rvert)$. We say that $C$ is \emph{polynomial-space} if there exists a polynomial $p$ such that $C_x$ uses at most $p(|x|)$ space.

Let $C = (C_x)_{x \in \{0,1\}^*}$ denote a quantum algorithm. Given a string $x \in \{0,1\}^*$ and a state $\rho$ whose number of qubits matches the input size of the circuit $C_x$, we write $C(x,\rho)$ to denote the output of circuit $C_x$ on input $\rho$. The output will in general be a mixed state as the circuit $C_x$ can perform measurements.

We say that a quantum algorithm $C = (C_x)_{x \in \{0,1\}^*}$ is \emph{time-uniform} (or simply \emph{uniform}) if there exists a polynomial-time Turing machine that on input $x$ outputs the description of $C_{x}$. Similarly we say that $C$ is \emph{space-uniform} if there exists a polynomial-space Turing machine that on input $x$ outputs the description of $C_{x}$.

\paragraph{Oracle Algorithms} 
\emph{Oracle algorithms} are quantum algorithms whose circuits, in addition to having the gates as described above, have the ability to query (perhaps in superposition) a function $O$ (called an \emph{oracle}) which may act on many qubits. This is essentially the same as the standard quantum query model~\cite[Chapter 6]{nielsen2000quantum}, except the circuits may perform non-unitary operations such as measurement, reset, and tracing out. Each oracle call is counted as a single gate towards the size complexity of a circuit.  The notion of time- and space-uniformity for oracle algorithms is the same as with non-oracle algorithms: there is a polynomial-time/polynomial-space Turing machine -- which does \emph{not} have access to the oracle -- that outputs the description of the circuits.

Given an oracle $\oracle = (\oracle_n)_{n \in \N}$ where each $\oracle_n:\{0,1\}^n \to \{0,1\}$ is an $n$-bit boolean function, we write $C^\oracle = (C_x^\oracle)_{x \in \{0,1\}^*}$ to denote an oracle algorithm where each circuit $C_x$ can query any of the functions $(\oracle_n)_{n \in \N}$ (provided that the oracle does not act on more than the number of qubits of $C_x)$.

In this paper we distinguish between classical and quantum queries. We say that an oracle algorithm $C^\oracle$ makes quantum queries if it can query $\oracle$ in superposition; this is akin to the standard query model. We say that $C^\oracle$ makes classical queries if, before every oracle call, the input qubits to the oracle are measured in the standard basis. In this case, the algorithm would be querying the oracle on a \emph{probabilistic mixture} of inputs. For clarity, we write $C^{\ket{\oracle}}$ to denote $C$ making quantum queries, and $C^\oracle$ to denote $C$ making classical queries.

A specific oracle that we consider throughout is the $\pspace$ oracle. What we mean by this is a sequence of functions $(\pspace_n)_{n \in \N}$ where for every $n$, the function $\pspace_n$ decides $n$-bit instances of a $\pspace$ complete language (such as Quantified Satisfiability~\cite{papadimitriou1994computational}).
\par We state the following observation. 

\begin{lemma}\label{lemma:Replace|PSPACE>}
Let $C^{\ket{\pspace}} = (C_x^{\ket{\pspace}})_{x \in \{0,1\}^*}$ denote a polynomial-time oracle algorithm (not necessarily uniform) that makes quantum queries to $\pspace$ and has one-bit classical output. Then there exists a polynomial-space algorithm $D = (D_x)_{x \in \{0,1\}^*}$ such that for all $x \in \{0,1\}^*$, $D_x$ is a unitary and the functionality of $C_x^{\ket{\pspace}}$ is exactly the same as introducing polynomial number of qubits initialized in the zero state, applying unitary $D_x$ and then measuring the first qubit in computational basis to get a classical output. Furthermore if $C$ is uniform, then $D$ is space-uniform.
\end{lemma}

\begin{proof}
 This follows because for a polynomial-time (oracle) algorithm, we can always introduce new qubits only at the beginning and defer measurements and tracing out to the end, and each oracle query in $C_x^{\ket{\pspace}}$ to the $\pspace$ oracle can be computed by first introducing several ancillas initialized in the zero state and then applying a unitary that implements the classical polynomial space algorithm for the PSPACE-complete language and uncomputes all the intermediate results. Furthermore, the description of the unitary can be generated by polynomial-space Turing machines.
\end{proof}

Finally, we will consider \emph{hybrid} oracles $\oracle$ that are composed of two separate oracles $\ro$ and the $\ket{\pspace}$ oracle. In this model, the oracle algorithm $C^\oracle$ makes classical queries to $\ro$, and quantum queries to $\pspace$. We abuse the notation and refer to algorithms having access to hybrid oracles as oracle algorithms. 

\paragraph{State Synthesis} 
We define the following ``state complexity class''. Intuitively it captures the set of quantum states that can be \emph{synthesized} by polynomial-space quantum algorithms.

\begin{definition}[$\statePSPACE$]\label{def:statePSPACE}
	$\statePSPACE$ is the class of all sequences $(\rho_x)_{x \in S}$ for some set $S \subseteq \{0,1\}^*$ (called the \emph{promise}) such that there is a polynomial $p$ where each $\rho_x$ is a density matrix on $p(|x|)$ qubits, and for every polynomial $q$ there exists a space-uniform polynomial-space quantum algorithm $C = (C_{x})_{x \in \{0, 1\}^*}$ such that for all $x \in S$, the circuit $C_x$ takes no inputs and outputs a density matrix $\sigma$ such that $\td(\sigma,\rho_x) \leq \exp(-q(\lvert x \rvert))$. 
	
	We say that the state sequence $(\rho_x)_{x \in S}$ is \emph{pure} if each $\rho_x$ is a pure state $\ketbra{\psi_x}{\psi_x}$; in that case we usually denote the sequence by $(\ket{\psi_x})_{x \in S}$.
\end{definition}

The following theorem says that, for $\statePSPACE$ sequences that are pure, there in fact is a polynomial-\emph{time} oracle algorithm that makes quantum queries to a $\pspace$ oracle to synthesize the state sequence. 

\begin{theorem}[Section 5 of \cite{IPforStates}]\label{GenerateStatePSPACE}
     Let $(\ket{\psi_{x}})_{x \in S}$ be a $\statePSPACE$ family of pure states. Then there exists a polynomial-time oracle algorithm $A^{\ket{\pspace}}$ such that on input $x \in S$, the algorithm outputs a pure state that is $\exp(-\lvert x \rvert)$-close in trace distance to $\ket{\psi_{x}}$. 
\end{theorem}

\subsection{Public Key Quantum Money Schemes}
\label{sec:def:pkqm}

\begin{definition}[Oracle-aided Public Key Quantum Money Schemes]
A \emph{oracle-aided public key quantum money scheme} $\mathcal{S}^\mathcal{O}$ consists of three uniform polynomial-time oracle algorithms $\qmalgs$:
\begin{itemize}
    \item $\keygen^{\oracle}(1^n)$: takes as input a security parameter $n$ in unary notation and generates secret key-public key pair $(sk, pk)$.
    \item $\mint^{\oracle}(sk)$: takes as input $sk$ and mints banknote $\rho_s$ associated with the serial number $s$.
    \item $\verify^{\oracle}(pk,(s,\rho_s))$: takes as inputs $pk$ and an alleged banknote $(s, \rho_s)$ and outputs $\rho'_s \otimes \ketbra{\bfx}{\bfx}$, where $\bfx \in \{\accept,\reject\}$.
\end{itemize}
\noindent For simplicity, when we don't care about the output $\rho_s'$ in $\verify^{\oracle}$, we sometimes denote the event that $\rho'_s \otimes \ketbra{\accept}{\accept}  \leftarrow \verify^{\oracle}(pk,(s,\rho_s))$ as $\verify^{\oracle}(pk,(s,\rho_s))$ accepts.
\end{definition}

\noindent We require the above oracle-aided public key quantum money scheme to satisfy both correctness and security properties.

\subsubsection{Correctness} 
We first consider the traditional definition of correctness considered by prior works. Roughly speaking, correctness states that the verification algorithm accepts the money state produced by the minting algorithm. Later, we consider a stronger notion called reusability which stipulates that the verification process on a valid money outputs another valid money state (not necessarily the same as before).

\begin{definition}[Correctness]
An oracle-aided public key quantum money scheme $\qmalgs$ is \textbf{$\delta$-correct} if the following holds for every $n \in \mathbb{N}$:
$$\prob \left[ \rho'_s \otimes \ketbra{\accept}{\accept}  \leftarrow \verify^{\oracle}(pk,(s,\rho_s))\ :\ \substack{(sk,pk) \leftarrow  \keygen^{\oracle}(1^n)\\ (s,\rho_s) \leftarrow \mint^{\oracle}(sk)} \right] \geq \delta,$$
where the probability is also over the randomness of $\oracle$.

We omit $\delta$ when $\delta \geq 1 - \negl(n)$.
\end{definition}

\paragraph{Reusability.} In this work, we consider quantum money schemes satisfying the stronger notion of reusability.

\begin{definition}[Reusability] \label{def:reusable}
  An oracle-aided public key quantum money scheme $\qmalgs$ is \textbf{$\delta$-reusable} if the following holds for every $n \in \mathbb{N}$ and for every polynomial $q(n)$:
  $$\prob \left[ \rho_s^{(q(n))} \otimes \ketbra{\accept}{\accept}  \leftarrow \verify^{\oracle}(pk,(s,\rho_s^{(q(n)-1)}))\ :\ \substack{(sk,pk) \leftarrow  \keygen^{\oracle}(1^n)\\ \ \\ (s,\rho_s^{(0)}) \leftarrow \mint^{\oracle}(sk)\\ \ \\\forall i \in [q(n)-1],\ \rho_s^{(i+1)} \otimes \ketbra{\bfx}{\bfx} \leftarrow \verify^{\oracle}(pk,(s,\rho_s^{(i)}))} \right] \geq \delta,$$
  where the probability is also over the randomness of $\oracle$.
  
We omit $\delta$ when $\delta \geq 1 - \negl(n)$.
\end{definition}

\noindent In general, gentle measurement lemma~\cite{Win99} can be invoked to prove that correctness generically implies reusability. However, this is not the case in our context. The reason being that the verification algorithm performs intermediate measurements whenever it makes classical queries to an oracle and these measurements cannot be deferred to the end. 

\subsubsection{Security} 

We consider the following security notion. Basically, it says that no efficient adversary can produce two alleged banknotes from one valid banknote with the same serial number.

\begin{definition}[Security]
An oracle-aided public key quantum money scheme $\qmalgs$ is \textbf{$\delta$-secure} if the following holds for every $n \in \mathbb{N}$ and for every uniform polynomial-time oracle algorithm $\adversary^{\oracle}$:
$$\prob \left[ \verify^{\oracle}(pk,(s,\phi_1)) \text{ accepts and } \verify^{\oracle}(pk,(s,\phi_2)) \text{ accepts}\ :\ \substack{(sk,pk) \leftarrow  \keygen^{\oracle}(1^n)\\ (s,\rho_s) \leftarrow \mint^{\oracle}(sk) \\ \phi \leftarrow \adversary^{\oracle}(pk, (s, \rho_s))} \right] \leq \delta,$$
where the probability is also over the randomness of $\oracle$. By $\phi_i$, we mean the reduced density matrix of $\phi$ on the $i^{th}$ register.

We omit $\delta$ when $\delta \leq \negl(n)$.
\end{definition}

\subsection{Jordan's Lemma and Alternating Projections}\label{subsect:MWtechnique}

In this section, we analyze {\em{alternating projection algorithm}}, a tool for estimating the acceptance of the verification algorithm on a state with only one copy of that state, which was introduced by Marriott and Watrous~\cite{marriott} for witness-preserving error reduction. This section follows section~4.1 in \cite{Rewinding} mostly.

For two binary-outcome projective measurements $M_1 = \{\Pi_1, I - \Pi_1\}, M_2 = \{\Pi_2, I - \Pi_2\}$, the alternating projection algorithm applies measurements $M_1$ and $M_2$ alternatively ($\Pi_1, \Pi_2$ corresponds to outcome $1$) until a stopping condition is met. The following lemma can help us analyze the distribution of the outcomes by decomposing it into several small subspaces.

\begin{lemma}[Jordan's Lemma]\label{Lemma:Jordan}
For any two projectors $\Pi_1$, $\Pi_2$, there exists an orthogonal decomposition of the Hilbert space into one-dimensional and two-dimensional subspaces $S_i$ that are invariant under both $\Pi_1$ and $\Pi_2$. Moreover, if $S_i$ is a one-dimensional subspace, then $\Pi_1$ and $\Pi_2$ act as identity or rank-zero projectors inside $S_i$. If $S_i$ is a two-dimensional subspace, then $\Pi_1$ and $\Pi_2$ are rank-one projectors inside $S_i$. To be more specific, there are two unit vectors $\ket{v_i}$ and $\ket{w_i}$ such that inside $S_i$, $\Pi_1$ projects on $\ket{v_i}$ and $\Pi_2$ projects on $\ket{w_i}$.
\end{lemma}

Let measurement $M_{\mathsf{Jor}} = \{P_i\}_i$ where $P_i$ is the projector onto the subspace $S_i$ defined above. Then both $M_1$ and $M_2$ commute with $M_{\mathsf{Jor}}$. Therefore the distribution of outcomes of each $M_1$ and $M_2$ will not change if we insert $M_{\mathsf{Jor}}$ at any point of the alternating projections. We can analyze the distribution of outcome sequence by first applying $M_{\mathsf{Jor}}$ and then applying $M_1$, $M_2$ alternatively. For each two-dimensional subspace $S_i$, denote $p_i :=\lvert\braket{v_i\vert w_i}\rvert^2$. (This can be seen as a quantity that measures the angle between $\Pi_1$ and $\Pi_2$ inside $S_i$.)

For now, let's assume that there are only two-dimensional subspaces in the decomposition. The general case where there exist one-dimensional subspaces is essentially the same and can be handled similarly. Then, $\Pi_1 = \sum_i \ketbra{v_i}{v_i}, \Pi_2 = \sum_i \ketbra{w_i}{w_i}$.\footnote{Generally, for each one-dimensional subspace $S_i$ on which $\Pi_1$ acts as identity, we can set $\ket{v_i}$ to be the vector that spaces $S_i$. Let $A$ be the set of index $i$ such that $\Pi_1$ is not a rank-zero projector inside $S_i$. Then $\Pi_1 = \sum_{i \in A} \ketbra{v_i}{v_i}$. Similarly $\Pi_2 = \sum_{i \in B}\ketbra{w_i}{w_i}$ where $B$ and $\ket{w_i}$ are defined in a similar way.}

Now let's state a result first proven in \cite{marriott} and restated in many later works (e.g., \cite{Rewinding}). 
\begin{proposition}\label{proposition:distribution}
If initially the state is $\ket{v_i}$\footnote{The same holds for each $\ket{v_i} (i \in A)$ generally where we define $p_i = 1$ if $\Pi_1$ and $\Pi_2$ act both as identity in subspace $S_i$ and we define $p_i = 0$ if $\Pi_1$ acts as identity while $\Pi_2$ acts as zero-projector in subspace $S_i$.} and we apply $M_2$, $M_1$ alternatively for $N$ times, then the outcome sequence $b_1, b_2, b_3, \cdots, b_{2N}$ will follow the distribution below
\begin{enumerate}
    \item Set $b_0 = 1$ (because applying $M_1$ to $\ket{v_i}$ will give outcome 1).
    \item For each $j$, we set $b_{j}$ to be $b_{j - 1}$ with probability $p_i$, and $1 - b_{j - 1}$ otherwise.
\end{enumerate}

\noindent Moreover, whenever we measure $M_1$ and get outcome 1, we will go back to state $\ket{v_i}$.
\end{proposition}

Then the fraction of bit-flips in the outcome sequence will be a good estimation of $1 - p_i$ if we start from $\ket{v_i}$.

\subsection{Compressed Oracle Techniques}\label{sec:prelims:compressed_oracle}
In this section, we present some basics of compressed oracle techniques introduced by Zhandry~\cite{zhandryCompressedOracle}.

For a quantum query algorithm $A$ interacting with a random oracle, let's assume that $A$ only queries the random oracle with $n$-bit input and gets $1$-bit output for simplicity. By the deferred measurement principle, without loss of generality we can write $A$ in the form of a sequence of unitaries $U_0, U_f, U_1, U_f, \cdots, U_{k - 1}, U_f, U_k$ where $U_i$ is the unitary that prepares the ${(i + 1)^{th}}$ query of $A$ to $R$ and $U_f$ maps $\ket{x}\ket{y}$ to $\ket{x}\ket{y \oplus f(x)}$ where $f$ is the chosen random function from all the functions with $n$-bit input and $1$-bit output. 

Then the behavior of $A$ when it is interacting with a random oracle can be analyzed in the following \emph{purified view}: 
\begin{itemize}
    \item Initialize register $\reg{A}$ to be the input for $A$ (along with enough ancillas $\ket{0}$) and initialize register $\reg{F}$ to be a uniform superposition of the truth tables of all functions from $[2^n]$ to $\{0, 1\}$ (to be more specific, $\reg{F}$ is initialized to $\frac{1}{\sqrt{2^{2^n}}}\sum_{f \text{ is a $2^n$-bit string}}\ket{f}$ where $\ket{f} = \ket{f(0)}\ket{f(1)}\cdots \ket{f(2^n - 1)}$ and $\ket{f(i)}$ consists of one qubit).
    \item Apply $U_0, U_F, U_1, U_F, \cdots, U_{k - 1}, U_F, U_{k}$ where $U_i$ is acting on $\reg{A}$ and $U_F$ maps $\ket{x}\ket{y}\ket{f}_{\reg{F}}$ to $\ket{x}\ket{y \oplus f(x)}\ket{f}_{\reg{F}}$.
\end{itemize}

In fact, the output (mixed) state of $A$ (we also take the randomness of $f$ into account) equals to the reduced density matrix on the output register of the state we obtain from the above procedure as $U_i, U_F$ commutes with computational basis measurement on $\reg{F}$. More generally, the output (mixed) state of a sequence of algorithms with access to random oracle can also be analyzed in the same way.

\begin{definition}[Fourier basis]
$\ket{\hat{0}} := \ket{+} = \frac{1}{\sqrt{2}}\left(\ket{0} + \ket{1}\right)$.
$\ket{\hat{1}} := \ket{-} = \frac{1}{\sqrt{2}}\left(\ket{0} - \ket{1}\right)$.

One can easily check that $\{\ket{\hat{0}}, \ket{\hat{1}}\}$ is a basis because it's just the result of applying hermitian matrix $H$ to $\ket{0}, \ket{1}$. We call this basis as \emph{Fourier basis}.
\end{definition}

The following fact is simple and easy to check, but crucial in compressed oracle techniques. Roughly speaking, it says that if we see $\mathit{CNOT}$ in Fourier basis, its control bit and target bit swaps.

\begin{fact}\label{fact:FourierBasis}
The operator defined by $\ket{y}\ket{y'} \rightarrow \ket{y \oplus y'}\ket{y'}$ for all $y, y' \in \{0, 1\}$ is the same as the operator defined by $\ket{\widehat{y}}\ket{\widehat{y'}} \rightarrow \ket{\hat{y}}\ket{\widehat{y' \oplus y}}$ for all $y, y' \in \{0, 1\}$.
\end{fact}

By \Cref{fact:FourierBasis}, when we look at the last two registers in Fourier basis, $U_F$ becomes 
$$\ket{x}\ket{\widehat{y}}\ket{\widehat{y_0}} \ket{\widehat{y_1}} \cdots \ket{\widehat{y_{2^n - 1}}} \rightarrow \ket{x}\ket{\widehat{y}}\ket{\widehat{y_0}} \ket{\widehat{y_1}} \cdots \ket{\widehat{y_{x - 1}}} \ket{\widehat{y_x \oplus y}} \ket{\widehat{y_{x + 1}}} \cdots \ket{\widehat{y_{2^n - 1}}}.$$

Initially, $\reg{F}$ is $\ket{\hat{0}}\ket{\hat{0}}\cdots\ket{\hat{0}}$ and each call of $U_F$ only changes one position if we look at the last two registers in Fourier basis. So after $k$ calls of $U_F$, the state can be written as $$\sum_{\substack{a,  y_0, y_1, \cdots, y_{2^n - 1}\\ \text{such that there are at most $k$ non-zero} \\ \text{in  $y_0, y_1, \cdots, y_{2^n - 1}$}}}\alpha_{a, y_0, y_1, \cdots, y_{2^n - 1}}\ket{a}_{\reg{A}}\ket{\widehat{y_0}} \ket{\widehat{y_1}} \cdots \ket{\widehat{y_{2^n - 1}}}.$$

We can record those non-$\hat{0}$ into a database. To be more specific, there exists a unitary that maps those $\ket{\widehat{y_0}} \ket{\widehat{y_1}} \cdots \ket{\widehat{y_{2^n - 1}}}$ (perhaps along with some ancillas) to a database $\ket{x_1}\ket{\widehat{y_{x_1}}} \cdots \ket{x_l}\ket{\widehat{y_{x_l}}}$ (perhaps along with some unused space) where $x_1 < x_2 < \cdots < x_l, \widehat{y_{x_i}} \ne \hat{0}$ and $l \le k$. That is, there exists a unitary that can compress the oracle. Furthermore, the inverse of the unitary can decompress the database back to the oracle.

\paragraph{Chernoff bound} Finally, we state here a variant of the Chernoff bound that we will use.

\begin{theorem}[Chernoff Bound]
Suppose $X_1, X_2, \cdots, X_n$ are independent random variables taking values from $\{0, 1\}$ such that each $X_i = 1$ with probability $p$. Let $\mu = pn$. Then for any $\delta > 0$,
$$\prob\left[\sum_{i = 1}^n X_i \ge (1 + \delta)\mu\right] \le e^{-\delta^2\mu / (2 + \delta)},$$
$$\prob\left[\sum_{i = 1}^n X_i \le (1 - \delta)\mu\right] \le e^{-\delta^2\mu/2}.$$
\end{theorem}

\section{Synthesizing Witness States In Quantum Polynomial Space}
\label{sec:synthesize}

In the classical setting it is easy to see that given a (classical) verifier circuit $V$ (which may make oracle queries to $\pspace$), one can find in polynomial space a witness string $y$ that is accepted by $V$: one can simply perform brute-force search over all strings and check whether $V^\pspace$ accepts $x$. 

In the section, we prove the quantum counterpart, where now the verifier circuit is quantum and can make quantum queries to the $\pspace$ oracle. We show that given the description of such a verifier circuit, with the help of the quantum $\ket{\pspace}$ oracle, we can efficiently \emph{synthesize} a witness state $\rho$ that is accepted by $V$ with probability greater than the desired guarantee (provided that there exists a witness state with acceptance probability greater than the threshold). Formally:

\begin{theorem}\label{theorem:synthesizer}
Let $a$ (called the \emph{guarantee}), $b$ (called the \emph{threshold}) be functions such that $b(n) - a(n) \ge \frac{1}{p(n)}$ for every $n$ where $p$ is a polynomial. Let $\VerCircuit^{\ket{\pspace}}$ denote a uniform oracle algorithm. Then there exists a uniform oracle algorithm $\syn$ (called the \emph{synthesizer}) such that for every $x \in S$,
$$\prob\left[\VerCircuit^{\ket{\pspace}}(x,\syn^{\ket{\pspace}}(x)) \text{ accepts} \right] \ge a(\length{x})$$
where $ S := \left \{x: \max_\rho \prob\left[\VerCircuit^{\ket{\pspace}}(x,\rho) \text{ accepts}\right] \ge b(\length{x}) \right \}$.
\end{theorem}

This theorem follows directly from \Cref{GenerateStatePSPACE} and the following lemma. 

\begin{lemma}\label{MainLemma}
Let $a, b$ be functions such that $b(n) - a(n) \ge \frac{1}{p(n)}$ for every $n$ where $p$ is a polynomial. Let $\VerCircuit^{\ket{\pspace}}$ denote a uniform oracle algorithm, and let $S$ be the corresponding set as in \Cref{theorem:synthesizer}. Then there exists a $\statePSPACE$ family of pure states $(\ket{\psi_{x}})_{x \in S}$ where each state $\ket{\psi_x}$ is bipartite on two registers (labeled $\reg{M}$ and $\reg{E}$) such that for every $x \in S$,
$$\prob\left[\VerCircuit^{\ket{\pspace}}(x,\rho_{\reg{M}}) \text{ accepts} \right] \ge a(\length{x})$$ where $\rho_{\reg{M}}$ is the reduced density matrix of $\ket{\psi_{x}}$ on register $\reg{M}$, i.e. $\rho_{\reg{M}} = \Tr_{\reg{E}}(\psi_x)$.
\end{lemma}

\begin{proof}[Proof of \Cref{theorem:synthesizer}]
Let $a'(n) = a(n) + \exp(-n)$ and $b'(n) = b(n)$, where $a(n),b(n)$ are as given by the conditions in \Cref{theorem:synthesizer}. Applying \Cref{MainLemma} with functions $a'(n),b'(n)$, we obtain a $\statePSPACE$ state sequence $(\ket{\psi_{x}})_{x \in S}$ such that for every $x \in S$,
$$\prob\left[\mathsf{V}^{\ket{\pspace}}(x,\rho_{\reg{M}}) \text{ accepts} \right] \ge a(\length{x}) + \exp(-\length{x})$$
where $ S := \left \{x: \max_\rho \prob\left[\mathsf{V}^{\ket{\pspace}}(x,\rho) \text{ accepts}\right] \ge b(\length{x}) \right \}$.

\Cref{GenerateStatePSPACE} implies that there exists a polynomial-time oracle algorithm $A^{\ket{\pspace}}$ that on input $x \in S$, outputs a pure state $\ket{\varphi_x}$ that is $\exp(-\length{x})$-close to $\ket{\psi_x}$. This implies that the reduced density matrix of $A^{\ket{\pspace}}(x)$ on register $\reg{M}$, which we denote by $\sigma_{\reg{M}}$, is also $\exp(-\length{x})$-close to $\rho_{\reg{M}} = \Tr_{\reg{E}}(\psi_x)$ (this follows from the fact that trace distance is non-increasing when you discard subsystems). Thus for every $x \in S$, we have
$$\prob\left[\mathsf{V}^{\ket{\pspace}}(x,\sigma_{\reg{M}}) \text{ accepts} \right] \ge a(\length{x})$$
because otherwise $\VerCircuit^{\ket{\pspace}}$ would be able to distinguish between $\rho_{\reg{M}}$ and $\sigma_{\reg{M}}$ with more than $\exp(-\length{x})$ bias. 

The synthesizer $\syn$ works as follows: on input $x \in S$ it runs the oracle algorithm $A^{\ket{\pspace}}(x)$ to obtain a pure state $\ket{\varphi_x}$, and then traces out the $\reg{E}$ register and returns the remaining state on the $\reg{M}$ register as output. 
\end{proof}

The remainder of \Cref{sec:synthesize} will be devoted to the proof of \Cref{MainLemma}. We will use the techniques and results from \cite{marriott} (also presented in \Cref{subsect:MWtechnique} for completeness). In \Cref{subsect:description} we present the description of the state family along with the description of a circuit family that generates (an approximation of) the state family. In \Cref{subsect:ProofOfMainLemma} we prove that the state family satisfies the requirements. 

\subsection{Description of the State Family and Circuit Family}\label{subsect:description}

In this section, we implement our ideas from \Cref{sec:IntuitionOfSyn} in a formal way. Recall that our algorithm in \Cref{sec:IntuitionOfSyn} repeatedly does the following (which we will call a {\em{trial}}): start from a maximally entangled state, estimate the acceptance probability coherently using MW technique and if the estimated acceptance probability high, then output the remaining state. Roughly speaking, the target state we aim to generate will be the remaining state after a successful trial (a trial is successful if the estimated acceptance probability is high). Looking ahead, in order to prove \Cref{MainLemma}, we only need to show two things. Firstly, our algorithm actually outputs a good approximation of the target state, so our target state forms a $\statePSPACE$ family; Secondly, our target state will indeed be accepted with high probability.

Now let's start by giving a formal description of the state family.

\paragraph{The state family}

Let $\VerCircuit^{\ket{\pspace}}$ be the uniform oracle algorithm given in the condition of \Cref{MainLemma}. From \Cref{lemma:Replace|PSPACE>}, there exists a space-uniform polynomial-space algorithm $\VerHat = (\VerHat_x)_{x \in \{0, 1\}^*}$ such that $\VerHat_x$ is unitary and the functionality of $\VerCircuit_x^{\ket{\pspace}}$ is exactly the same as introducing $k(\length{x})$ new ancilla qubits in $\ket{0}$, applying unitary $\VerHat_x$ and then measuring the first qubit in computational basis where $k$ is a polynomial. Let $m(\length{x})$ be the number of qubits that $\VerCircuit_x$ takes as input, which is also a polynomial. 

Fix $x \in S, n = \length{x}$. We sometimes omit subscript $x$ when it is clear from the context. For convenience, we write $m(n), k(n)$ as $m, k$ respectively from now on.

Let $\register{M}$ denote the register containing the $m$ input qubits. Let $\register{K}$ denote the register containing the $k$ ancilla qubits. Let $\register{Ans}$ denote the first qubit (i.e. the one that will be measured in computational basis to decide whether $\mathsf{V}_x^{\ket{\pspace}}$ accepts or rejects, and outcome $1$ means accept while outcome $0$ means reject). Let $\register{Aux}$ denote a register containing another $m$ fresh qubits. 

Here we define two binary-outcome projective measurements on $\reg{MK}$. Define $P^1 := \ketbra{0^k}{0^k}_{\register{K}}$, $P^0 := \id_{\reg{MK}} - P^1$ and $P := \{P^0, P^1\}$. Intuitively, $P^1$ corresponds to ``valid input subspace'' (i.e., the ancilla qubits are initialized properly). Define $Q^1 := \VerHat_x^\dagger\ketbra{1}{1}_{\register{Ans}}\VerHat_x, Q^0 := \VerHat_x^\dagger\ketbra{0}{0}_{\register{Ans}}\VerHat_x$ and $Q := \{Q^0, Q^1\}$. Intuitively, $Q^1$ corresponds to the state that will be accepted if we apply $\VerHat_x$ to it and then measure $\reg{Ans}$ in computational basis. So $Q$ checks whether $\VerCircuit_x^{\ket{\pspace}}$ will accept as long as register $\register{K}$ is initialized properly. The following simple observation is implicitly shown in~\cite{marriott}.

\begin{observation}\label{obs:AcceptProb}
The maximum acceptance probability of $\VerCircuit^{\ket{\pspace}}(x, \cdot)$ is exactly the largest eigenvalue of $P^1Q^1P^1$.
\end{observation}
\begin{proof}
First, we show that the maximum eigenvalue of $P^1Q^1P^1$ is upper bounded by the maximum acceptance probability of $\VerCircuit^{\ket{\sf PSPACE}}(x, \cdot)$.

For any pure state $\ket{\phi}$ on register $\reg{MK}$, let $\rho_{\reg{M}} = \frac{1}{\|P^1\ket{\phi}\|^2}\Tr_{\reg{K}}(P^1\ketbra{\phi}{\phi}P^1)$. Then 
\begin{align*}
    \bra{\phi}P^1 Q^1 P^1\ket{\phi} =& \Tr(Q^1P^1\ketbra{\phi}{\phi}P^1) = \|P^1\ket{\phi}\|^2 \Tr(Q^1\frac{1}{\|P^1\ket{\phi}\|^2}P^1\ketbra{\phi}{\phi}P^1)\\ 
    =& \|P^1\ket{\phi}\|^2 \Tr(Q^1(\rho_{\reg{M}} \otimes \ketbra{0^k}{0^k}_{\reg{K}}))\\
    =&\prob\left[\VerCircuit^{\ket{\pspace}}(x, \rho_{\reg{M}}) \text{ accepts}\right]\\
    \le& \max_\rho\prob\left[\VerCircuit^{\ket{\pspace}}(x, \rho) \text{ accepts}\right]
\end{align*}

Second, we show that the maximum eigenvalue of $P^1Q^1P^1$ is lower bounded by the maximum acceptance probability of $\VerCircuit^{\ket{\sf PSPACE}}(x, \cdot)$. By a simple convexity argument, we can assume without loss of generality, the acceptance probability of $\VerCircuit^{\ket{\sf PSPACE}}(x, \cdot)$ achieves its maximum on pure state $\ket{\phi_0}$. Let $\ket{\phi} = \ket{\phi_0}\ket{0^k}$. Then
$$\bra{\phi}P^1 Q^1 P^1\ket{\phi} =  \prob\left[\VerCircuit^{\ket{\pspace}}(x, \phi_0) \text{ accepts}\right]  = \max_\rho\prob\left[\VerCircuit^{\ket{\pspace}}(x, \rho) \text{ accepts}\right]$$

Therefore, this observation holds true.

\end{proof}

We first define a subroutine $\mathsf{Trial}$ where $N := \max\left(\frac{3a + b}{(b - a)^2}\left(m + 2 - \log(b - a)\right), \frac{16b}{(b - a)^2}\right)$ is polynomial in $n$ (recall that $b(n) - a(n) \ge \frac{1}{p(n)}$ where $p$ is a polynomial).

\begin{algorithmic}[1]
    \State {Initialize register $\reg{MAux}$ to be $\frac{1}{\sqrt{2^m}}\sum_{i \in \{0, 1\}^{m}}\ket{i}_{\register{M}}\ket{i}_{\register{Aux}}$}
    \State {Initialize register $\reg{K}$ to be $\ket{0^k}_{\register{K}}$}
    \State {Introduce a new $2N + 1$ qubit register $\register{Y} := \reg{{Y}_0 \cdots {Y}_{2N}}$ initialized to be $\ket{1}\ket{0^{2N}}$}
    \State {Introduce a new register $\register{Cnt}$ initialized in $\ket{0}$}
    \For {$i = 1, 2, \cdots, N$}
        \State {Measure $\reg{MK}$  with $Q$ coherently, store the outcome in $\register{Y}_{2i - 1}$}
        \State {Measure $\reg{MK}$ with $P$ coherently, store the outcome in $\register{Y}_{2i}$}
    \EndFor
    \State {Compute the number of times that $y_j = y_{j - 1}$ in superposition and store the result in $\register{Cnt}$}
    \State {Do the projective measurement $\mathsf{Test} := \{\mathsf{Yes} := \sum_{j \ge N(a + b)} \ketbra{j}{j}_{\register{Cnt}} \otimes \ketbra{1}{1}_{\register{Y}_{2N}}, \mathsf{No} := I - \mathsf{Yes}\}$}
\end{algorithmic}

Define $\reg{E} := \register{K} \otimes \register{Aux} \otimes \register{Y} \otimes \register{Cnt}$ (i.e., all registers except $\reg{M}$).

\begin{definition}[State family $(\ket{\psi_x})_{x \in S}$]\label{def:state_family}
Let $S$ be the set defined in \Cref{MainLemma}. When $x \in S$, let $\ket{\psi_{x}}$ denote the state in register $\register{M} \otimes \register{E}$ after a successful implementation of $\mathsf{Trial}$ (i.e., the outcome of $\mathsf{Test}$ is $\mathsf{Yes}$). When $x$ are clear from the context, we also write it as $\ket{\psi}$.
\end{definition}

Observe that in $\mathsf{Trial}$, we initialize a pure state in register $\register{M} \otimes \register{E}$ (line~1 - line~4), then apply a unitary on it (line~5- line~9) as all measurements are conducted coherently, and finally do a projective measurement (line~10). So the definition above indeed gives us a family of states $(\ket{\psi_{ x}})_{x \in S}$ such that each $\ket{\psi_{x}}$ is a pure state on $l(n)$ qubits where $l$ is a polynomial.

\paragraph{The circuit family} 
Now let's construct a circuit family (or algorithm) that can generate efficiently an approximation of the state family $(\ket{\psi_x})_{x \in S}$. For any polynomial $q$ and approximation factor $\exp(-q(n))$, the circuit $C_x$ operates as follows where $T := 2^{m + 2} q(n)$ is exponential in $n$.

\begin{algorithmic}[1]
\For {$t = 1, \cdots, T$}
    \State{Run $\mathsf{Trial}$}
    \If {it is successful} 
    \State{\Return the state in the register $\register{M} \register{E}$}
    \EndIf
\EndFor
\State {\Return an arbitrary state with $l(n)$ qubits}
\end{algorithmic}

\subsection{Proof of Lemma~\ref{MainLemma}}\label{subsect:ProofOfMainLemma}

In the section, we prove that the pure state family $(\ket{\psi_{x}})_{ x \in S}$ satisfies the requirements in \Cref{MainLemma} by applying known result in \Cref{subsect:MWtechnique}. 

Fix $x \in S$. We associate $\Pi_1$ with $P^1$, $\Pi_2$ with $Q^1$, $M_1$ with $P$ and $M_2$ with $Q$, and adopt the notations in \Cref{subsect:MWtechnique}. Then $P^1Q^1P^1 = \sum_i \ketbra{v_i}{v_i} \sum_i \ketbra{w_i}{w_i} \sum_i \ketbra{v_i}{v_i} = \sum_i p_i \ketbra{v_i}{v_i}$. From $x \in S$ and \Cref{obs:AcceptProb}, we can assume $p_1 = \max_i p_i \ge b$. \footnote{Generally  $P^1Q^1P^1 = \sum_{i \in A \cap B} p_i \ketbra{v_i}{v_i}$. We can assume $p_1 = \max_{i \in A \cap B} p_i \ge b$.}

To begin with, let's prove that the state family $(\ket{\psi_{x}})_{x \in S}$ satisfies the first requirement in \Cref{MainLemma}. That is, it is a $\statePSPACE$ family, which can be approximately generated by the circuit family. Notice that $C_x$ outputs $\ket{\psi_x}$ whenever one of the exponential $\mathsf{Trial}$s succeeds. So we first analyze the success probability of one $\mathsf{Trial}$.

\begin{lemma}\label{Lemma:Prob1Trial}
$\prob\left[\mathsf{Trial} \text{ succeeds}\right] \ge \frac{1}{2^{m + 2}}.$
\end{lemma}

\begin{proof}
The success probability of $\mathsf{Trial}$ doesn't change if we measure each qubit of $\register{Y}$ once the outcome is stored in it because computational basis measurements on $\reg{Y}$ commutes with the operations in line 9 and 10 of $\mathsf{Trial}$. Thus we can think of it as measure $\reg{MK}$ directly with $P$ and $Q$ alternatively, get a classical outcome sequence $y := y_1y_2\cdots y_{2N}$ and return $\mathsf{Yes}$ if $y_{2N} = 1$ and the number of times that $y_j = y_{j - 1}(1 \le j \le 2N)$ is at least $N(a + b)$ (where $y_0 = 1$), which is an alternating projection algorithm. For simplicity, we will denote by $\good$ the set of $y$ that corresponds to outcome $\mathsf{Yes}$. Now let's analyze the probability of $y \in \good$.

An important observation is that the initial state can also be written in forms of $\ket{v_i}$. Because $P^1 = \sum_i \ketbra{v_i}{v_i}$, $\ket{v_i}$ forms a basis for the Hilbert space $\mathcal{H}_{\reg{M}} \otimes \ketbra{0^k}{0^k}_{\reg{K}}$. Let $\ket{u_i}$ be a truncation on $\reg{M}$ of $\ket{v_i}$, i.e. $\ket{v_i}_{\reg{MK}} = \ket{u_i}_{\reg{M}}\ket{0^k}_{\reg{K}}$. Then $\ket{u_i}$ forms a basis for $\mathcal{H}_{\reg{M}}$. Thus by \Cref{lemma:MaximallyEntangledState}, the state
$$\frac{1}{\sqrt{2^m}}\sum_{i = 0}^{2^m - 1}\ket{i}_{\register{M}}\ket{i}_{\register{Aux}} = \frac{1}{\sqrt{2^m}}\sum_{i}\ket{u_i}_{\register{M}}\ket{\overline{u_i}}_{\register{Aux}}$$

Consequently, $\frac{1}{\sqrt{2^m}}\sum_{i = 0}^{2^m - 1}\ket{i}_{\register{M}}\ket{0^k}_{\reg{K}}\ket{i}_{\register{Aux}} = \frac{1}{\sqrt{2^m}}\sum_{i}\ket{v_i}_{\register{MK}}\ket{\overline{u_i}}_{\register{Aux}}.$ \footnote{In the general case, the summation is over $i \in A$. That is, $\frac{1}{\sqrt{2^m}}\sum_{i = 0}^{2^m - 1}\ket{i}_{\register{M}}\ket{0^k}_{\reg{K}}\ket{i}_{\register{Aux}} = \frac{1}{\sqrt{2^m}}\sum_{i \in A}\ket{v_i}_{\register{MK}}\ket{\overline{u_i}}_{\register{Aux}}.$ Thus for each $i \in A$, we will be in subspace $S_i$ with probability $\frac{1}{2^m}$ if we apply $M_{\mathsf{Jor}}$. }

Notice that we can apply $M_{\mathsf{Jor}}$ on register $\reg{MK}$ before the alternating projections without changing the distribution of $y$. Applying $M_{\mathsf{Jor}}$ to the above state, the post-measurement state will be $\ket{v_i}_{\reg{MK}}\ket{\overline{u_i}}_{\reg{Aux}}$ with probability $\frac{1}{2^m}$. And by \Cref{proposition:distribution}, when we start from $\ket{v_i}_{\reg{MK}}\ket{\overline{u_i}}_{\reg{Aux}}$, $y_j = y_{j - 1}$ with probability $p_i$ for each $j$ independently.

In particular, we will start from $\ket{v_1}_{\reg{MK}}\ket{\overline{u_1}}_{\reg{Aux}}$ with probability $\frac{1}{2^m}$. And when we start from $\ket{v_1}_{\reg{MK}}\ket{\overline{u_1}}_{\reg{Aux}}$, $y_j = y_{j - 1}$ with probability $p_1$ for each $j$ independently. This can be seen as performing $2N$ independent coin flips with bias $p_1 \ge b$. And $y \in \good$ if the number of heads (denote as $cnt$) is an even greater than or equal to $N(a + b)$.

By Chernoff bound, $$\prob\left[cnt < (a + b)N \right] \le \exp(-Np_1(1 - \frac{a + b}{2p_1})^2) \le \exp(-N\frac{(b - a)^2}{4b}) \le \frac{1}{4}$$
\begin{align*}
    \prob\left[cnt \text{ is an odd}\right] =& \sum_{j = 0}^{N - 1}\binom{2N}{2j + 1} p_1^{2j+1} (1 - p_1)^{2N - 2j - 1}\\
    =& \frac{1}{2}((p_1 + 1 - p_1)^{2N} - (p_1 - (1 - p_1))^{2N})\\
    \leq& \frac{1}{2}
\end{align*}

Thus by union bound, when the post-measurement state after $M_{\mathsf{Jor}}$ is $\ket{v_1}_{\reg{MK}}\ket{\overline{u_1}}_{\reg{Aux}}$, $y \in \good$ with probability at least $\frac{1}{4}$. 

So $\prob\left[\mathsf{Trial} \text{ succeeds}\right] = \prob\left[y \in \good \right] \ge \frac{1}{2^m}\frac{1}{4} = \frac{1}{2^{m + 2}}.$

\end{proof}

\begin{claim}
$(\ket{\psi_{x}})_{x \in S}$ is a $\statePSPACE\left[S\right]$ family.
\end{claim}

\begin{proof}
We only need to show that our construction $C_x$ satisfies the requirements in \Cref{def:statePSPACE}.

From the construction, $C_x$ can be generated by polynomial space Turing machine on input $x$ and $C_x$ uses at most polynomial space at any time. Thus $C = (C_x)_{x \in \{0, 1\}^*}$ is a space-uniform polynomial-space quantum algorithm. It is obvious from the construction that $C_x$ takes no inputs. The only remaining thing is to prove $C_x$ outputs a good approximation of $\ket{\psi_x}$ when $x \in S$.

Fix $x \in S$. Whenever there is a successful implementation of $\mathsf{Trial}$, $C_x$ will output $\ket{\psi_x}$. Moreover, by  \Cref{Lemma:Prob1Trial},  recall that $T = 2^{m + 2}q(n)$,
$$\prob\left[\text{$T$ independent repetitions of $\mathsf{Trial}$ all fail}\right] = (1 - \prob\left[\mathsf{Trial} \text{ succeeds}\right])^T \le (1 - \frac{1}{2^{m + 2}})^T \le e^{-q(n)}$$

That is, except with probability $e^{-q(n)}$, $C_x$ outputs $\ket{\psi_x}$. 

As a result, the state outputted by $C_x$ is $e^{-q(n)}$-close to $\ket{\psi_x}$ in trace distance.
\end{proof}

The second requirement in \Cref{MainLemma} that $(\ket{\psi_x})_{x \in S}$ needs to satisfy is that the reduced density matrix of $\ket{\psi_x}$ will be accepted by $\VerCircuit^{\ket{\pspace}}(x, \cdot)$ with high probability. This is intuitively correct because the real acceptance probability should not be too far from the estimated acceptance probability. Now let's prove it formally.

\begin{claim}\label{claim:GoodAcceptProb}
For every $x \in S$, $\prob\left[\mathsf{V}_{x}^{\ket{\pspace}}(\rho_{\reg{M}}) \text{ accepts} \right] \ge a$
where $\rho_{\reg{M}}$ is the reduced density matrix of $\psi_{x}$ on register $\reg{M}$.
\end{claim}

\begin{proof}

Fix $x \in S$. We will omit subscripts when it is clear from the context.

Similar with what we did in \Cref{Lemma:Prob1Trial}, this probability doesn't change if in the generation of $\psi_x$ (i.e. $\mathsf{Trial}$), we measure $\reg{Y}$ directly instead (as we only care about the part in $\reg{M}$). Let $\psi'_x$ be the state we obtain from a successful implementation of $\mathsf{Trial}$ if we measure directly instead. Then the reduced density matrices of $\psi_x$ and $\psi'_x$ are the same on register $\reg{M}$.

Notice that $\mathsf{V}_x^{\ket{\pspace}}(\rho_{\reg{M}})$ is just applying $\VerHat_x$ to $\rho_{\reg{M}} \otimes \ketbra{0^k}{0^k}_{\reg{K}}$ and then measuring the first qubit in computational basis, or equivalently, it is just measuring $\rho_{\reg{M}} \otimes \ketbra{0^k}{0^k}_{\reg{K}}$ with $Q$. By the definition, $\psi'_x$ is $\ket{0^k}$ on register $\reg{K}$ (because the outcome $y_{2N}$ should be 1). So the reduced density matrix of $\psi'_x$ on register $\reg{MK}$ is exactly $\rho_{\reg{M}} \otimes \ketbra{0^k}{0^k}_{\reg{K}}$.

Consider the following alternating projection algorithm: 

We start from $\frac{1}{\sqrt{2^m}}\sum_{i \in \{0, 1\}^{m}}\ket{i}_{\register{M}}\ket{0^k}_{\register{K}}\ket{i}_{\register{Aux}}$, apply $Q, P$ to the state alternatively for $N$ times to obtain a classical outcome sequence $y:=y_1y_2\cdots y_{2N}$ and if $y$ meets the requirement (to be more accurate, $y \in \good$ where $\good$ is defined in \Cref{Lemma:Prob1Trial}), we will additionally apply $Q$ to get an outcome $z$ and accept if $z = 1$. 

In the above algorithm, if $y \in \good$, the (mixed) state remaining is exactly $\psi_x'$, whose reduced density matrix on $\reg{MK}$ is $\rho_{\reg M}\otimes \ketbra{0^k}{0^k}_{\reg K}$. Recall that $\mathsf{V}_x^{\ket{\pspace}}(\rho_{\reg{M}})$ is just measuring $\rho_{\reg{M}} \otimes \ketbra{0^k}{0^k}_{\reg{K}}$ with $Q$. Therefore,

$$\prob\left[\mathsf{V}_{x}^{\ket{\pspace}}(\rho_{\reg{M}}) \text{ accepts} \right] = \prob\left[z = 1 \ \vert\ y \in \good\right]$$

From \Cref{subsect:MWtechnique}, $\prob\left[z = 1 \ \vert\ y \in \good\right]$ will not change if we insert $M_{\mathsf{Jor}}$ in front of the alternating projections. So we can calculate it by projecting the initial state $\frac{1}{\sqrt{2^m}}\sum_{i = 0}^{2^m - 1}\ket{i}_{\register{M}}\ket{0^k}_{\reg{K}}\ket{i}_{\register{Aux}}$ to one of the subspaces $S_i$, getting the post-measurement state $\ket{v_i}_{\reg{MK}}\ket{\overline{u_i}}_{\reg{Aux}}$ and then sampling $y$ and $z$ as if we start from this state (here we also use the fact that $\frac{1}{\sqrt{2^m}}\sum_{i = 0}^{2^m - 1}\ket{i}_{\register{M}}\ket{0^k}_{\reg{K}}\ket{i}_{\register{Aux}}$ can be written in the form $\frac{1}{\sqrt{2^m}}\sum_{i}\ket{v_i}_{\register{MK}}\ket{\overline{u_i}}_{\register{Aux}}$). Denote $E_i$ be the event that we get the post-measurement state $\ket{v_i}_{\reg{MK}}\ket{\overline{u_i}}_{\reg{Aux}}$. Then $\prob\left[E_i\right] = \frac{1}{2^m}$\footnote{Generally, $\prob\left[E_i\right] = \frac{1}{2^m}$ for each $i \in A$. And thus all the summations below will be only over $i \in A$. }. Therefore, 

\begin{align*}
  \prob\left[z = 1 \ \vert\ y \in \good\right] =&  \frac{\prob\left[z = 1 \wedge y \in \good\right]}{\prob\left[y \in \good\right]}\\
  =& \frac{\sum_i \prob\left[E_i \right]\prob\left[y \in \good \ \vert \ E_i\right] \prob\left[z = 1 \ \vert \ E_i \wedge y \in \good \right]}{\sum_i \prob\left[E_i\right]\prob\left[y \in \good \ \vert \ E_i\right]}\\
  =& \frac{\sum_i p_i \prob\left[y \in \good \ \vert \ E_i\right] }{\sum_i \prob\left[y \in \good \ \vert \ E_i\right]}\\
\end{align*}

Same as \Cref{Lemma:Prob1Trial}, when we start from $\ket{v_i}_{\reg{MK}}\ket{\overline{u_i}}_{\reg{Aux}}$, the probability of $y \in \good$ is the same as the probability that during $2N$ independent coin flips with bias $p_i$, the number of heads (denote as $cnt$) is an even greater than or equal to $N(a + b)$.

So if $p_i < a$, by Chernoff bound, 
$$\prob\left[y \in \good \ \vert \ E_i\right] \le \prob\left[cnt \ge N(a + b)\right] \le \exp(-2Np_i(\frac{a + b}{2p_i} - 1)^2/(1 + \frac{a + b}{2p_i})) \le \exp(-N\frac{(b - a)^2}{3a + b})$$

From \Cref{Lemma:Prob1Trial}, $\prob\left[y \in \good \ \vert \ E_1\right] \ge \frac{1}{4}$. As a result,

\begin{align*}
    \sum_{p_i \ge a} (p_i - a) \prob\left[y \in \good \ \vert \ E_i\right] - \sum_{p_i < a} (a - p_i)\prob\left[y \in \good \ \vert \ E_i\right] \ge (b - a)\frac{1}{4} - 2^m a \exp(-N\frac{(b - a)^2}{3a + b}) > 0
\end{align*}
where we use the fact that there are only $2^m$ $\ket{v_i}$s because $P^1$ has rank $2^m$.\footnote{For the general case, we only sum over $i \in A$ and $|A|$ equals to the rank of $P^1$, which is $2^m$.}

The above inequality can be rearranged into $\sum_i p_i \prob\left[y \in \good \ \vert \ E_i\right] > a\sum_i \prob\left[y \in \good \ \vert \ E_i\right].$ 

Therefore, $\prob\left[\mathsf{V}_{x}^{\ket{\pspace}}(\rho_{\reg{M}}) \text{ accepts} \right] \ge a$, which ends the proof of this claim.
\end{proof}

\Cref{MainLemma} follows directly from the above two claims.

\section{Insecurity of Oracle-Aided Public-Key Quantum Money}\label{sec:AttackClassicalAccess} 
In this section, we will use the synthesizer from \Cref{sec:synthesize} as a building block to attack the \oqm where $\oracle$ is a hybrid oracle composed of random oracle $\ro$ and $\ket{\pspace}$. Formally:

\begin{theorem}\label{theorem:insecurity}
Reusable and secure \oqm $( \keygen^{{\ro},\ket{\pspace}}, \allowbreak\mint^{{\ro},\ket{\pspace}}, \verify^{{\ro},\ket{\pspace}})$ does not exist where $\ro$ is a random oracle.
\end{theorem}

Informally speaking, our synthesizer in \Cref{theorem:synthesizer} works for uniform oracle algorithm $\VerCircuit^{\ket{\pspace}}$. However, in the \oqm we aim to attack, the verification algorithm has access to random oracle $\ro$ in addition to $\ket{\pspace}$. Inspired by \cite{RemoveROfromObf,AK22}, we try to remove $\ro$ and simulate it with a good database. Based on the ideas in \Cref{sec:IntuitionWithR}, we give the following attacker.

Let $\oracle$ be the hybrid oracle composed of random oracle $\ro$ and $\ket{\pspace}$. For a $\delta_r$-reusable $\delta_s$-secure oracle-aided quantum money scheme $\qmalgs$ where $\delta_r = 0.99, \delta_s = \negl(n)$, denote $l(n)$ to be the number of queries to $\ro$ made by one execution of $\keygen^{\oracle}$ and one execution of $\mint^{\oracle}$. By efficiency of $\verify^{\oracle}$, there exists a uniform oracle algorithm $\VerCircuit^{\ket{\pspace}} = (\VerCircuit_x^{\ket{\pspace}})_{x \in \{0, 1\}^*}$ such that running $\VerCircuit_{(pk, D, s)}^{\ket{\pspace}}(\rho)$ is the same as running $\verify^{\oracle}(pk, (s, \rho))$ where $\ro$ is simulated with database $D$. 

Let $\epsilon = 0.01$, $b = 1 - \sqrt{1 - \delta_r + \epsilon}$, $a = 0.99b$. By \Cref{theorem:synthesizer}, there exists a polynomial-time uniform oracle algorithm $\syn^{\ket{\pspace}}$ which can generate an ``almost optimal'' witness state of $\VerCircuit^{\ket{\pspace}}_{(pk, D, s)}$ with guarantee $a$ and threshold $b$. Now let's construct the adversary $\adversary^{\oracle}$.

\paragraph{Adversary $\adversary^{\oracle}$.}It takes as input a valid banknote $(s, \rho_s)$ and public key $pk$, and behaves as follows.

\begin{enumerate}
    \item Let $t \xleftarrow{\$} \{0,\ldots,\lceil \frac{l}{\epsilon} \rceil - 1\}$. Let $D = \emptyset$, $\rho_s^{(0)} = \rho_s$. Run the following $t$ times. In $i^{th}$ iteration,
    \begin{enumerate}
        \item  $\rho_s^{(i)} \otimes \ketbra{\bfx}{\bfx}  \leftarrow \verify^{\oracle}(pk, s, \rho_s^{(i - 1)})$.
        \item Add query-answer pairs to $\ro$ in item (a) into $D$.
    \end{enumerate}
    \item Denote $D_0 = D$.
    \item For $k = 0, 1, \cdots, N(n) - 1$ where $N(n) = \frac{100l(n)}{\left(1 - \sqrt{1 - \delta_r + \epsilon}\right)^2}$ is polynomial in $n$,
    \begin{enumerate}
        \item $\sigma_{D_{k}} \leftarrow \syn^{\ket{\pspace}}(pk, D_{k}, s)$.
        \item Run $\verify^{\oracle}(pk, (s, \sigma_{D_{k}}))$.
        \item Let $D_{k + 1}$ consist of all the query-answer pairs to $\ro$ in item (b) and the pairs in $D_{k}$. 
    \end{enumerate}
    \item $j \xleftarrow{\$} \{0, 1, \ldots, N(n) - 1\}$.
    \item Output $\phi = \phi_1 \otimes \phi_2$ where $\phi_i  \leftarrow \syn^{\ket{\pspace}}(pk, D_j, s)$ $(i = 1, 2)$.
\end{enumerate}

\paragraph{Analysis of $\adversary^{\oracle}$} Now let's prove that $\adversary^{\oracle}$ outputs what we want. We will use the notations defined in the construction of $\adversary^{\oracle}$.

\begin{theorem}\label{theorem:AnalysisOfA}
Given input $(pk, (s, \rho_s))$ generated by $\keygen^{\oracle}$ and $\mint^{\oracle}$, $\adversary^{\oracle}$ outputs two alleged banknotes associated with the serial number $s$ that will be accepted with high probability. Formally:
$$\prob \left[ \verify^{\oracle}(pk,(s,\phi_1)) \text{ accepts and } \verify^{\oracle}(pk,(s,\phi_2)) \text{ accepts} \right] \geq 1.8\left(1 - \sqrt{1 - \delta_r + \epsilon}\right)^2 - 1,$$
where the probability is over the randomness of $\ro$, the randomness of the generation of the input for $\adversary^{\oracle}$ (that is, the randomness of $\keygen^{\oracle}$ and $\mint^{\oracle}$) and the randomness of our adversary $\adversary^{\oracle}$.
\end{theorem}

\begin{proof}[Proof of \Cref{theorem:AnalysisOfA}]

The proof will be divided into two parts. Informally speaking, in the first part, we will show that for every $k$, $\sigma_{D_k}$ works well on the simulation, i.e. $\VerCircuit_{(pk, D_k, s)}^{\pspace}(\sigma_{D_k})$ accepts with high probability; In the second part, we will show that for every $k$, if $\verify^{\oracle}(pk, (s, \cdot))$ behaves far from $\VerCircuit_{(pk, D_k, s)}^{\pspace}$ on input $\sigma_{D_k}$, then we make progress. Then we will combine the results to prove \Cref{theorem:AnalysisOfA}.

\paragraph{The first part}

The synthesizer $\syn$ in \Cref{theorem:synthesizer} works well provided that good witness state for $\VerCircuit_{(pk, D_k, s)}^{\ket{\pspace}}$ exists. Our candidate for the good witness state is $\rho_s^{(t)}$ as it is accepted by $\verify^{\oracle}$ with high probability by the definition of reusability. We begin by arguing that with high probability, our databases contain necessary information for running verification on $\rho_s^{(t)}$ and thus $\verify$ can not distinguish whether it is interacting with random oracle $\ro$ or the simulated one. Formally:

\begin{claim}\label{claim:dbwithvalidinput}
Let $D_{\keygen}, D_{\mint}$ be the query-answer pairs made during the generation of the input $pk$ and $(s, \rho_s)$ (that is, the execution of $\keygen^{\oracle}$ and $\mint^{\oracle}$). Then
$$\prob \left[\verify^{\oracle}(pk, (s, \rho_s^{(t)})) \text{ queries in $D_{\mint} \cup D_{\keygen} - D$} \right] \le \epsilon,$$
where the probability is over the randomness of $\ro$, the randomness of the generation of the input for $\adversary^{\oracle}$ and the randomness of our adversary $\adversary^{\oracle}$.
\end{claim}

\begin{proof}
We only care about $l(n)$ query positions (those inside $D_{\mint} \cup D_{\keygen}$) and we repeatedly sample $t \xleftarrow{\$} \{0,\ldots,\lceil \frac{l}{\epsilon} \rceil - 1\}$ times. Thus intuitively $D$ should reveal all the positions we care about. Formally,

\begin{align*}
   & \prob\left[\verify^{\oracle}(pk, (s, \rho_s^{(t)})) \text{ queries in $D_{\mint} \cup D_{\keygen} - D$}\right]\\
   & \\
    \le& \sum_{q \in D_{\mint} \cup D_{\keygen}}\prob\left[t = \min_j\left[\verify^{\oracle}(pk,(s,\rho_s^{(j)})) \text{ queries $q$} \right]\right]\\
    \le&\ l \cdot \frac{1}{\lceil \frac{l}{\epsilon} \rceil}\\
    \le& \epsilon
\end{align*}
where the probabilities are only over the randomness of our adversary $\adversary^{\oracle}$ and we use the fact that $t$ is picked uniformly random from $\{0, 1, \ldots, \lceil \frac{l}{\epsilon} \rceil - 1\}$, so it matches $\min_j\left[\verify^{\oracle}(pk,(s,\rho_s^{(j)})) \text{ queries $q$} \right]$ (which may follow some distribution, but is independent of $t$ anyway) with probability less or equal $\frac{1}{\lceil \frac{l}{\epsilon} \rceil}$.

After taking the randomness of $\ro$ and the randomness of the generation of the input for $\adversary^{\oracle}$ into account, we can get the claim.

\end{proof}

The random oracle $\ro$ can be implemented by on-the-fly simulation. Thus $\verify^{\oracle}(pk, (s, \rho_s^{(t)}))$ can be implemented by simulating $\ro$ with database $D_{\mint} \cup D_{\keygen} \cup D$. If $\verify$ doesn't make queries in $(D_{\mint} \cup D_{\keygen} \cup D) \Delta  D = D_{\mint} \cup D_{\keygen} - D$, then it can not distinguish whether $\ro$ is simulated with $D_{\mint} \cup D_{\keygen} \cup D$ or $D$. That is, $D$ is a good database to simulate the verification process on input $\rho_s^{(t)}$ if $\verify$ doesn't make queries inside $D_{\mint} \cup D_{\keygen} - D$. Thus the acceptance probability of $\VerCircuit^{\ket{\pspace}}_{(pk, D, s)}(\rho_s^{(t)})$ should be close to that of $\verify^{\oracle}(pk, (s, \rho_s^{(t)}))$, which is high by the definition of reusability. On average, the performance of the simulation on input $\rho_s^{(t)}$ can only increase if we include more queries into the database. Thus for every $k$, $\rho_s^{(t)}$ should be a good witness state for $\VerCircuit^{\ket{\pspace}}_{(pk, D_k, s)}$. The intuition above is captured by the following claim.

\begin{claim}\label{claim:SimOnRho}
We use the same definition of $D_{\keygen}$ and $D_{\mint}$ as in \Cref{claim:dbwithvalidinput}. $\forall k \in \left[N(n)\right],$
$$\prob \left[\VerCircuit_{(pk, D_k, s)}^{\ket{\pspace}}(\rho_s^{(t)}) \text{ accepts} \right] \geq \delta_r - \epsilon$$ where the probability is over the randomness of $\ro$, the randomness of the generation of the input for $\adversary^{\oracle}$ and the randomness of our adversary $\adversary^{\oracle}$.
\end{claim}

\begin{proof}

This claim follows from \Cref{def:reusable} and \Cref{claim:dbwithvalidinput}. The following probabilities are over the same randomness as the probability in the above claim.

\begin{align*}
    &\prob \left[\verify^{\oracle}(pk, (s, \rho_s^{(t)})) \text{ accepts} \right]\\
    & \\
    = & \prob \left[\verify^{\oracle}(pk, (s, \rho_s^{(t)})) \text{ accepts and queries in $D_{\mint} \cup D_{\keygen} - D_k$}  \right]\\
    & + \prob \left[\verify^{\oracle}(pk, (s, \rho_s^{(t)})) \text{ accepts and never queries in $D_{\mint} \cup D_{\keygen} - D_k$} \right]\\
    & \\
    \le & \prob \left[\verify^{\oracle}(pk, (s, \rho_s^{(t)})) \text{ queries in $D_{\mint} \cup D_{\keygen} - D$} \right]\\
    & + \prob \left[\VerCircuit^{\ket{\pspace}}_{(pk, D_k, s)}(\rho_s^{(t)}) \text{ accepts and never queries in $D_{\mint} \cup D_{\keygen} - D_k$} \right]\\
    & \\
    \le & \epsilon + \prob \left[\VerCircuit^{\ket{\pspace}}_{(pk, D_k, s)}(\rho_s^{(t)}) \text{ accepts} \right]
\end{align*}
where we use the fact that $D \subseteq D_k$ and the fact that we can use on-the-fly simulation to implement $\ro$. As a result, $\verify^{\oracle}(pk, (s, \rho_s^{(t)}))$ can also be seen as simulating $\ro$ with $D_{\keygen} \cup D_{\mint} \cup D_k$, which is different from $D_k$ only on $(D_{\keygen} \cup D_{\mint} \cup D_k) \Delta D_k = D_{\keygen} \cup D_{\mint} - D_k$. Thus $\verify$ can not distinguish whether $\ro$ is simulated with $D_k$ or $D_{\mint} \cup D_{\keygen} \cup D_k$ if it never queries in $D_{\keygen} \cup D_{\mint} - D_k$. The above inequality can be rearranged as
$$\prob \left[\VerCircuit^{\ket{\pspace}}_{(pk, D_k, s)}(\rho_s^{(t)}) \text{ accepts} \right] \ge  \prob \left[\verify^{\oracle}(pk, (s, \rho_s^{(t)})) \text{ accepts} \right] - \epsilon \ge  \delta_r - \epsilon.$$
\end{proof}

 Intuitively, from \Cref{claim:SimOnRho}, for a large fraction of $\VerCircuit_{(pk, D_k, s)}^{\ket{\pspace}}$, good witness state exists. Therefore, our synthesizer can find an ``almost optimal'' one. Formally:

\begin{claim}\label{claim:SimOnSigma}
For every $k \in [N(n)],$
$$\prob\left[\VerCircuit_{(pk, D_k, s)}^{\ket{\pspace}}(\sigma_{D_k}) \text{ accepts} \right] \ge 0.99\left(1 - \sqrt{1 - \delta_r + \epsilon}\right)^2$$ where the probability is over the randomness of $\ro$, the randomness of the generation of the input for $\adversary^{\oracle}$ and the randomness of our adversary $\adversary^{\oracle}$.
\end{claim}

\begin{proof}
The following probabilities are over the same randomness as the probability in the above claim unless otherwise stated.

Define $S := \{(pk, D_k, s) : \max_w\prob\left[\VerCircuit_{(pk, D_k, s)}^{\ket{\pspace}}(w) \text{ accepts}\right] \ge 1 - \sqrt{1 - \delta_r + \epsilon}\}$ where the probability is only over the randomness of $\VerCircuit$. Then by \Cref{claim:SimOnRho} and averaging argument,
$$\prob \left[(pk, D_k, s) \in S \right] \geq 1 - \sqrt{1 - \delta_r + \epsilon}$$

By~\Cref{theorem:synthesizer}, $\forall (pk, D_k, s) \in S$, $\prob\left[\VerCircuit_{(pk, D_k, s)}^{\ket{\pspace}}(\sigma_{D_k}) \text{ accepts}\right] \ge 0.99(1 - \sqrt{1 - \delta_r + \epsilon})$ where the probability is only over the randomness of $\VerCircuit$. Therefore,
$$\prob \left[ \VerCircuit_{(pk, D_k, s)}^{\ket{\pspace}}(\sigma_{D_k}) \text{ accepts} \right] \geq 0.99\left(1 - \sqrt{1 - \delta_r + \epsilon}\right)^2.$$ 
\end{proof}

\paragraph{The second part}

We already know that $\sigma_{D_k}$ is accepted by $\VerCircuit^{\ket{\pspace}}_{(pk, D_k, s)}$ with high probability. The next step is to associate the acceptance probability of $\VerCircuit^{\ket{\pspace}}_{(pk, D_k, s)}$ and that of $\verify^{\oracle}(pk, (s, \cdot))$ on $\sigma_{D_k}$. If the difference of these two terms is large, the simulation with $D_k$ is not good enough. That is, $\verify^{\oracle}(pk, (s, \sigma_{D_k}))$ asks some important queries outside $D_k$. So in this case, $D_{k + 1}$ will contain more important queries and we make progress. Formally:

\begin{claim}\label{claim:UpdateDB}
We use the same notation as above. For every $k \in [N(n)],$
\begin{align*}
&\prob\left[\VerCircuit_{(pk, D_k, s)}^{\ket{\pspace}}(\sigma_{D_k}) \text{ accepts}\right] - \prob\left[\verify^{\oracle}(pk, (s, \sigma_{D_k})) \text{ accepts}\right]\\
\le & \mathbf{E}\left[\lvert D_{k + 1} \cap (D_{\keygen} \cup D_{\mint}) \rvert\right] - \mathbf{E}\left[\lvert D_{k} \cap (D_{\keygen} \cup D_{\mint}) \rvert\right]
\end{align*}
where the probabilities and the expectations are over the randomness of $\ro$, the randomness of the generation of the input for $\adversary^{\oracle}$ and the randomness of our adversary $\adversary^{\oracle}$.
\end{claim}

\begin{proof}
The following probabilities and expectations are over the same randomness of those in the above claim unless otherwise stated.

Similar as the arguments in \Cref{claim:SimOnRho}, $\VerCircuit^{\ket{\pspace}}_{(pk, s, D_k)}(\sigma_{D_k})$ and $\verify^{\oracle}(pk, (s, \sigma_{D_k}))$ behave differently only when they make queries in $(D_{\mint} \cup D_{\keygen} \cup D_k)\Delta D_k = D_{\mint} \cup D_{\keygen} - D_k$. Therefore,
\begin{align*}
    &\prob\left[\VerCircuit_{(pk, D_k, s)}^{\ket{\pspace}}(\sigma_{D_k}) \text{ accepts}\right]\\
    & \\
    =&\prob\left[\VerCircuit_{(pk, D_k, s)}^{\ket{\pspace}}(\sigma_{D_k}) \text{ accepts and queries in $D_{\mint} \cup D_{\keygen} - D_k$}\right]\\
    &+ \prob\left[\VerCircuit_{(pk, D_k, s)}^{\ket{\pspace}}(\sigma_{D_k}) \text{ accepts and never queries in $D_{\mint} \cup D_{\keygen} - D_k$}\right]\\
    & \\
    \le&\prob\left[\VerCircuit_{(pk, D_k, s)}^{\ket{\pspace}}(\sigma_{D_k}) \text{ queries in $D_{\mint} \cup D_{\keygen} - D_k$}\right]\\
    &+ \prob\left[\verify^{\oracle}(pk, (s, \sigma_{D_k})) \text{ accepts and never queries in $D_{\mint} \cup D_{\keygen} - D_k$}\right]\\
    & \\
    =&\prob\left[\verify^{\oracle}(pk, (s, \sigma_{D_k})) \text{ queries in $D_{\mint} \cup D_{\keygen} - D_k$}\right]\\
    &+ \prob\left[\verify^{\oracle}(pk, (s, \sigma_{D_k})) \text{ accepts and never queries in $D_{\mint} \cup D_{\keygen} - D_k$}\right]\\
    & \\
    \le&\prob\left[\left(D_{k + 1} - D_k\right) \cap \left(D_{\keygen} \cup D_{\mint}\right) \ne \emptyset\right] + \prob\left[\verify^{\oracle}(pk, (s, \sigma_{D_k})) \text{ accepts} \right]\\
    & \\
    \le& \mathbf{E}\left[\lvert\left(D_{k + 1} - D_k\right) \cap \left(D_{\keygen} \cup D_{\mint}\right)\rvert\right] + \prob\left[\verify^{\oracle}(pk, (s, \sigma_{D_k})) \text{ accepts} \right]
\end{align*}

Note that $D_{k} \subseteq D_{k + 1}$,
$$\mathbf{E}\left[\lvert\left(D_{k + 1} - D_k\right) \cap \left(D_{\keygen} \cup D_{\mint}\right)\rvert\right] = \mathbf{E}\left[\lvert D_{k + 1} \cap (D_{\keygen} \cup D_{\mint}) \rvert\right] - \mathbf{E}\left[\lvert D_k \cap (D_{\keygen} \cup D_{\mint}) \rvert\right]$$

Therefore, the claim holds true.
\end{proof}

Now let's combine the above results to prove \Cref{theorem:AnalysisOfA}. The probabilities and expectations below are over the randomness of $\ro$, the randomness of the generation of the input for $\adversary^{\oracle}$ and the randomness of our adversary $\adversary^{\oracle}$ (thus over the randomness of $t$ and $j$) unless otherwise stated.

By our construction and the union bound,
\begin{align*}
    &\prob \left[ \verify^{\oracle}(pk,(s,\phi_1)) \text{ accepts and } \verify^{\oracle}(pk,(s,\phi_2)) \text{ accepts} \right]\\
    \geq & 2\prob \left[ \verify^{\oracle}(pk,(s,\sigma_{D_j})) \text{ accepts} \right] - 1\\
    = &\frac{2}{N(n)}\sum_{k = 0}^{N(n) - 1} \prob \left[ \verify^{\oracle}(pk,(s,\sigma_{D_k})) \text{ accepts} \right] - 1
\end{align*}
From \Cref{claim:SimOnSigma} and \Cref{claim:UpdateDB},
\begin{align*}
&\frac{1}{N(n)} \sum_{k = 0}^{N(n) - 1}\prob \left[ \verify^{\oracle}(pk,(s,\sigma_{D_k})) \text{ accepts} \right]\\
\ge &\frac{1}{N(n)} \sum_{k = 0}^{N(n) - 1}\left(\prob \left[ \VerCircuit_{(pk, D_k, s)}^{\ket{\pspace}}(\sigma_{D_k}) \text{ accepts} \right] - \mathbf{E}\left[\lvert D_{k + 1} \cap (D_{\keygen} \cup D_{\mint}) \rvert\right] + \mathbf{E}\left[\lvert D_{k} \cap (D_{\keygen} \cup D_{\mint}) \rvert\right]\right)\\
\ge &\frac{1}{N(n)} \sum_{k = 0}^{N(n) - 1}\prob \left[ \VerCircuit_{(pk, D_k, s)}^{\ket{\pspace}}(\sigma_{D_k}) \text{ accepts} \right] - \frac{1}{N(n)}\mathbf{E}\left[\lvert D_{N(n)} \cap (D_{\keygen} \cup D_{\mint}) \rvert\right]\\
\ge &\frac{1}{N(n)} \sum_{k = 0}^{N(n) - 1}\prob \left[ \VerCircuit_{(pk, D_k, s)}^{\ket{\pspace}}(\sigma_{D_k}) \text{ accepts} \right] - \frac{l(n)}{N(n)}\\
\ge & 0.99\left(1 - \sqrt{1 - \delta_r + \epsilon}\right)^2 - 0.01\left(1 - \sqrt{1 - \delta_r + \epsilon}\right)^2\\
\geq & 0.9\left(1 - \sqrt{1 - \delta_r + \epsilon}\right)^2
\end{align*}

Therefore,
$$\prob \left[ \verify^{\oracle}(pk,(s,\phi_1)) \text{ accepts and } \verify^{\oracle}(pk,(s,\phi_2)) \text{ accepts} \right] \geq 1.8\left(1 - \sqrt{1 - \delta_r + \epsilon}\right)^2 - 1,$$
which ends our proof of \Cref{theorem:AnalysisOfA}.
\end{proof}

\begin{proof}[Proof of \Cref{theorem:insecurity}]
The proposed adversary $\adversary^{\oracle}$ is a valid attack because when $\epsilon = 0.01, \delta_r = 0.99$, 
$$1.8\left(1 - \sqrt{1 - \delta_r + \epsilon}\right)^2 - 1 \ge 1.8 (1 - 0.2)^2 - 1 \ge 0.1,$$ 
which is non-negligible.
\end{proof}

\section{Extensions to Quantum Access}\label{sec:AttackQuantumAccess}

In this section, we will explore a special case where some algorithms can have quantum access to the random oracle. 
We consider reusable secure \oqm $(\keygen^{\ket{\ro}, \ket{\pspace}}, \mint^{\ket{\ro}, \ket{\pspace}}, \verify^{\ro, \ket{\pspace}})$. Formally:

\begin{theorem}\label{theorem:quantumInsecurity}
Reusable and secure \oqm $( \keygen^{{\ket{\ro}},\ket{\pspace}}, \allowbreak\mint^{{\ket{\ro}},\ket{\pspace}}, \verify^{{\ro},\ket{\pspace}})$ does not exist where $\ro$ is a random oracle.
\end{theorem}

Without loss of generality, we can suppose the algorithms only make queries to the random oracle on input length $l(n)$ and receive $1$ bit output where $l$ is a polynomial. (If they make queries to $\ro$ on various input lengths, suppose the maximal input length is $l'(n)$. Let $l(n) = l'(n) + \log l'(n)$. We can modify the algorithms so that their queries on input length $k(n)$ will be made on input length $l(n)$ where the first $k(n)$ bits stores the true query position, the middle $l'(n) - k(n)$ bits are 0, and the last $\log l'(n)$ bits indicates $k(n)$.)

Let $\verify$ make $q(n)$ classical queries to $\ro$. Let $\keygen$ and $\mint$ make $q'(n)$ quantum queries to $\ro$ in total. 
Denote the reusability and the security of the scheme as $\delta_r$ and $\delta_s$ respectively where $\delta_r = 1 - \negl(n), \delta_s = \negl(n)$. When it is clear from the context, we sometimes omit $n$ for simplicity.

It's worth noting that the attacker in \Cref{sec:AttackClassicalAccess} 
doesn't take advantage of the fact that $\keygen$ and $\mint$ there can only make classical queries to $\ro$. 
In fact, the same attacker works even when $\keygen$ and $\mint$ can make quantum queries to $\ro$ (with some modifications on the number of iterations). 
To be more specific, here is our construction of the attacker where $T(n)$, $N(n)$, the guarantee $a$ and the threshold $b$ of $\syn$ will be determined later.

\paragraph{$\adversary^{\ro, \ket{\pspace}}$}It takes as input a valid banknote $(s, \rho_s)$ and public key $pk$, and behaves as follows.

\begin{enumerate}
    \item {\bf Test phase}: Let $t \xleftarrow{\$} \{0, 1, \ldots, T(n) - 1\}$. Let $D = \emptyset$, $\rho_s^{(0)} = \rho_s$. Run the following $t$ times. In $i^{th}$ iteration,
    \begin{enumerate}
        \item  $\rho_s^{(i)} \otimes \ketbra{\bfx}{\bfx}  \leftarrow \verify^{\ro, \ket{\pspace}}(pk, s, \rho_s^{(i - 1)})$.
        \item Add query-answer pairs to $\ro$ in item (a) into $D$.
    \end{enumerate}
    \item {\bf Update phase}: Let $j \xleftarrow{\$} \{0, 1, \ldots, N(n) - 1\}$. Let $D_0 = D$. Run the following $j$ times. In $k^{th}$ iteration,
    \begin{enumerate}
        \item $\sigma_{D_{k - 1}} \leftarrow \syn^{\ket{\pspace}}(pk, D_{k - 1}, s)$.
        \item Run $\verify^{\ro, \ket{\pspace}}(pk, (s, \sigma_{D_{k - 1}}))$.
        \item Let $D_{k}$ consist of all the query-answer pairs to $\ro$ in item (b) and the pairs in $D_{k - 1}$.
    \end{enumerate}
    \item {\bf Synthesize phase}: Output $\phi = \phi_1 \otimes \phi_2$ where $\phi_i  \leftarrow \syn^{\ket{\pspace}}(pk, D_j, s)$ $(i = 1, 2)$.
\end{enumerate}

This description of $\adversary^{\ro, \ket{\pspace}}$ is actually equivalent to our adversary in \Cref{sec:AttackClassicalAccess}. We move the line $j \xleftarrow{\$} \{0, 1, \ldots, N(n) - 1\}$ to the front because it will be easier to analyze. 

What is left is to prove an analogue of \Cref{theorem:AnalysisOfA}. That is, the output states of $\adversary$ will be accepted with high probability.

\begin{theorem}\label{theorem:QuantumAnalysisOfA}
    Given input $(pk, (s, \rho_s))$ generated by $\keygen^{\ket{\ro}, \ket{\pspace}}$ and $\mint^{\ket{\ro}, \ket{\pspace}}$, $\adversary^{\ro, \ket{\pspace}}$ outputs two alleged banknotes associated with the serial number $s$ that will be accepted with high probability. Formally:
    $$\prob \left[ \verify^{\ro, \ket{\pspace}}(pk,(s,\phi_1)) \text{ accepts and } \verify^{\ro, \ket{\pspace}}(pk,(s,\phi_2)) \text{ accepts} \right] \geq 1.8\left(1 - \sqrt{1 - \delta_r + \epsilon}\right)^2 - 1,$$
    where the probability is over the randomness of $\ro$, the randomness of the generation of the input for $\adversary^{\ro, \ket{\pspace}}$ (that is, the randomness of $\keygen^{\ket{\ro}, \ket{\pspace}}$ and $\mint^{\ket{\ro}, \ket{\pspace}}$) and the randomness of our adversary $\adversary^{\ro, \ket{\pspace}}$.
\end{theorem}

Similar to \Cref{theorem:AnalysisOfA}, we will show that the verification on $\sigma_{D_j} \leftarrow \syn^{\ket{\pspace}}(pk, D_j, s)$ accepts with high probability and then prove the theorem by union bound.

In \Cref{sec:AttackClassicalAccess}, we crucially rely on the fact that whenever we make a mistake, we make progress in the sense that we recover a query inside $D_{\keygen} \cup D_{\mint} - D_k$. 
However, now $\keygen, \mint$ can make quantum queries. As a result, $\keygen$ and $\mint$ could ``touch'' exponentially many positions. 
Fortunately, the compressed oralce technique introduced by Zhandry~\cite{zhandryCompressedOracle} can be seen as a quantum analogue of recording queries into a database.
Basically, if we run all the algorithms in the \emph{purified view} and see the register containing the oracle (labeled $\reg{F}$) in Fourier basis, then all except polynomial positions are $\ket{\hat 0}$ after polynomial quantum queries, and thus the register can be compressed using a unitary. In this work, in order to better mimic $D_k$ and $D_{\keygen} \cup D_{\mint} - D_k$ in \Cref{sec:AttackClassicalAccess}, we take advantage of the fact that $\verify$ only makes classical queries. To be more specific, we will maintain a register to store a database for all the classical queries and only record those non-$\ket{\hat 0}$ positions outside the database into $\reg{F}$. These two registers will be our analogue of $D_k$ and $D_{\keygen} \cup D_{\mint} - D_k$. We will elaborate on this idea in \Cref{subsec:compress}.

\subsection{A Purified View of the Algorithms}\label{subsec:purifiedView}

From \Cref{sec:prelims:compressed_oracle}, for any sequence of algorithms that only make queries to the random oracle on input length $l(n)$ and receive 1 bit output, we can analyze the output using a pure state that we obtain by running all the algorithms in the \emph{purified view} instead. By \emph{purified view}, we mean that we will purify the execution of the algorithms in the following way:
\begin{itemize}
    \item We will introduce a register $\reg{F}$ that contains the truth table of the oracle. Before the execution of the first algorithm, it is initialized to a uniform superposition of all the possible truth tables of the oracle, i.e. $\ket{\hat 0}^{\otimes 2^l}$.
    \item Instead of quantum query to $\ro$, we apply a unitary $U_Q: \ket{x}_{\reg{Q}}\ket{y}_{\reg{A}}\ket{f}_{\reg{F}} \rightarrow \ket{x}_{\reg{Q}}\ket{y \oplus f(x)}_{\reg{A}}\ket{f}_{\reg{F}}$ where $\reg{Q}$ stores the query position and $\reg{A}$ is for the answer bit. (The subscript $Q$ in $U_Q$ is for Quantum queries.)
    \item Instead of computational basis measurements, we apply $\mathit{CNOT}$ to copy it to a fresh ancilla.
\end{itemize}

Without loss of generality, we can suppose for any classical query to $\ro$, the register for query answer is always set to $\ket{0}$ before the query. Notice that a classical query to $\ro$ is equivalent to a computational basis measurement on the query position followed by a quantum query to $\ro$. An extra computational basis measurement on the answer of the query won't change the view. So a classical query in the purified view can be treated as applying the unitary $$U_C: \ket{x}_{\reg{Q}}\ket{0}_{\reg{A}}\ket{f}_{\reg{F}}\ket{D_{\ro}}_{\reg{D}_{\ro}} \rightarrow \ket{x}_{\reg{Q}}\ket{f(x)}_{\reg{A}}\ket{f}_{\reg{F}}\ket{D_{\ro}, (x, f(x))}_{\reg{D}_{\ro}}$$
where $\reg{D}_{\ro}$ is a register that we will use to purify the computational basis measurements in the classical queries. By $\ket{D_{\ro}}$, we mean a sequence of pairs $\ket{(x_1, z_1), (x_2, z_2), \cdots (x_k, z_k)}$ where $x_1, x_2, \cdots, x_k$ are not necessary to be distinct but if $x_i = x_j$, then we have the guarantee that $z_i = z_j$. Here $\reg{D}_{\ro}$ has enough space. That is, by $\ket{(x_1, z_1), \cdots (x_k, z_k)}$, we actually mean $\ket{(x_1, z_1), \cdots (x_k, z_k), \bot, \cdots, \bot}$ where $\bot$ is a special symbol that represents empty. Despite not being standard, we sometimes call $D_{\ro}$ database. (The subscript $C$ in $U_C$ is for Classical queries.)

Another convenient way to think of the purified view is to treat the execution of the algorithms as an interaction between two parties, the algorithm and the oracle. The oracle will maintain two private registers $\reg{F}$ and $\reg{D}_{\ro}$ (and also some ancillas initialized to be $\ket{0}$). If the algorithm is allowed to make quantum queries to $\ro$, the algorithm will submit $\reg{Q}\reg{A}$ to the oracle, and then the oracle will apply $U_Q$ and send $\reg{Q}\reg{A}$ back to the algorithm. If the algorithm is only allowed to make classical queries, the algorithm will submit $\reg{Q}$ to the oracle, and then the oracle will put a fresh ancilla on $\reg{A}$, apply $U_C$ and send $\reg{Q}\reg{A}$ to the algorithm. 

We will use $U_{\keygen, n}, U_{\mint, n}, U_{\verify, n}$ and $U_{\syn, n}$ to denote the unitary corresponding to the purified version of $\keygen$, $\mint$, $\verify$ and $\syn$ on security number $n$ respectively. Then $U_{\keygen, n}, U_{\mint, n}$ and $U_{\verify, n}$ are all in the form of preparing the first query and then repeatively answering the query by applying $U_Q$ or $U_C$ and preparing the next query (or the final output if there is no further query). In particular, we will write $U_{\verify, n}$ as $U_{q(n)}U_CU_{q(n) - 1}\cdots U_C U_0$. We will omit the subscript $n$ when it is clear from the context. 

Let $U_{\verify}':= U_qU_R U_{q- 1}\cdots U_R U_0$ where $U_R$ (the subscript $R$ is for Recording) is a unitary that in addition to a classical query $U_C$, it records the query-answer pair into a database maintained by $\adversary$. That is, $$U_R: \ket{x}_{\reg{Q}}\ket{0}_{\reg{A}}\ket{D_{\adversary}}_{\reg{D}_{\adversary}}\ket{f}_{\reg{F}}\ket{D_{\ro}}_{\reg{D}_{\ro}} \rightarrow \ket{x}_{\reg{Q}}\ket{f(x)}_{\reg{A}}\ket{D_{\adversary}, (x, f(x))}_{\reg{D}_{\adversary}}\ket{f}_{\reg{F}}\ket{D_{\ro}, (x, f(x))}_{\reg{D}_{\ro}}$$ where $\reg{D}_{\adversary}$ is the register that stores the database maintained by $\adversary$. Again by $\ket{D_{\adversary}}$ and $\ket{D_{\ro}}$, we mean a sequence of query-answer pairs where the query positions are not necessary to be distinct, but the pairs are consistent. $\reg{D}_{\adversary}$ has enough space. 

It's easy to see that $U_{\verify}'$ corresponds to running $\verify^{\ro, \ket{\pspace}}$ while the adversary records the query-answer pairs made by $\verify^{\ro, \ket{\pspace}}$. 

Then in the purified view, $\adversary^{\ro, \ket{\pspace}}$ is the following (we will denote by $U_{\adversary}$):

\begin{enumerate}
    \item Given input public key, serial number, the alleged banknote along with the register containing the truth table of the oracle, introduce a register $\reg{T}$ initialized to be $\frac{1}{\sqrt{T(n)}}\sum_{t = 0}^{T(n) - 1}\ket{t}_{\reg{T}}$ and introduce a register $\reg{J}$ initialized to be $\frac{1}{\sqrt{N(n)}}\sum_{j = 0}^{N(n) - 1}\ket{j}_{\reg{J}}$.
    \item {\bf{Test phase}}: Conditioned on the content in $\reg{T}$ is $t$, apply $U_{\verify}'$ on the banknote for $t$ times in sequential. (Or equivalently apply unitary $U_{\test} := \sum_{t = 0}^{T(n) - 1}U_{\verify}'^t\otimes \ketbra{t}{t}_{\reg{T}}$ where $U_{\verify}'^t$ means applying $U_{\verify}'$ for $t$ times.)
    \item {\bf Update phase}: Conditioned on the content in $\reg{J}$ is $j$, apply the following for $j$ times:
    \begin{enumerate}
        \item Apply $U_{\syn}$ on all the query-answer pairs we learn so far (i.e. the contents in $\reg{D}_{\adversary}$).
        \item Apply $U_{\verify}'$ on the state synthesized in item (a).
    \end{enumerate}
    (Or equivalently apply unitary $U_{\update} := \sum_{j = 0}^{N(n) - 1}(U_{\verify}'U_{\syn})^j \otimes \ketbra{j}{j}_{\reg{J}}$ where $(U_{\verify}'U_{\syn})^j$ means alternatively applying $U_{\syn}$ and $U_{\verify}'$ for $j$ times.)
    \item {\bf Synthesize phase}: Apply $U_{\syn_1}$ and $U_{\syn_2}$ on the query-answer pairs in $\reg{D}_{\adversary}$ to obtain two alleged banknotes where $U_{\syn_1}$ and $U_{\syn_2}$ are $U_{\syn}$ that acts on different registers.
\end{enumerate}

Then the acceptance probability of the following algorithms on the corresponding states (taking the randomness of $\ro$ into account) can be analyzed by running the following sequence of algorithms in the purified view.

\paragraph{$\bm{\verify^{\ro, \ket{\pspace}}(pk, (s, \cdot))}$ on $\bm{\rho_s^{(t)}}$}

The acceptance probability is the probability that 
\begin{enumerate}
    \item Run the algorithms $\keygen^{\ket{\ro}, \ket{\pspace}}, \mint^{\ket{\ro}, \ket{\pspace}}$ and $\adversary^{\ro, \ket{\pspace}}$ in a purified view sequentially where the input of $\keygen$ is the security parameter in unary notion, the input of $\mint$ is the output register corresponding to secret key of $\keygen$, and the input of $\adversary$ is the output register of $\mint$ and the register for the public key.
    \item Furthermore run $\verify^{\ro, \ket{\pspace}}$ in a purified view where the input of $\verify$ is the register for the public key, the register for the serial number and the register for $\rho_s^{(t)}$ (It's inside working space register of $\adversary$).
    \item Measure the outcome register of the above $\verify^{\ro, \ket{\pspace}}$ and obtain $\accept$.
\end{enumerate}

\paragraph{$\bm{\verify^{D_j, \ket{\pspace}}(pk, (s, \cdot))}$ on $\bm{\rho_s^{(t)}}$} 

Here $\verify^{D_j, \ket{\pspace}}$ means running $\verify^{\ro, \ket{\pspace}}$ where the queries to $\ro$ is answered by the database $D_j$ (all the query-answer pairs in $\reg{D}_{\adversary}$).

Define a unitary $$U_D: \ket{x}_{\reg{Q}}\ket{0}_{\reg{A}}\ket{D_j}_{\reg{D}_{\adversary}} \rightarrow \begin{cases}
    \ket{x}_{\reg{Q}}\ket{D_j(x)}_{\reg{A}}\ket{D_j, (x, D_j(x))}_{\reg{D}_{\adversary}}, & x \in D_j\\
    \ket{x}_{\reg{Q}}\sum_{z = 0}^1 \frac{1}{\sqrt 2} \ket{z}_{\reg{A}}\ket{D_j, (x, z)}_{\reg{D}_{\adversary}}, & x \notin D_j 
    \end{cases}$$
where by $\ket{D_j}$, we mean a sequence of query-answer pairs $\ket{(x_1, z_1), (x_2, z_2), \cdots (x_k, z_k)}$ where the query positions $x_1, x_2, \cdots, x_k$ are not necessary to be distinct, but the pairs are consistent. By $x \in D_j$, we mean there exists $z$ such that $(x, z)$ is a pair in $D_j$ and we will denote this $z$ as $D_j(x)$. By $x \notin D_j$, we mean for all $z$, $(x, z)$ is not a pair in $D_j$. (The subscript $D$ in $U_D$ is for simulating with Database.)

Then applying $U_D$ is exactly answering the query $x$ using $D_j$ (If $x$ is in the database, then answer the query using $D_j$; Otherwise, give a random answer while recording this query-answer pair into the database for later use). Thus the purified version of $\verify^{D_j, \ket{\pspace}}$ is $$U_{\simulation} := U_q U_D U_{q - 1}\cdots U_D U_0.$$

So the acceptance probability of $\verify^{D_j, \ket{\pspace}}(pk, (s, \cdot))$ on $\rho_s^{(t)}$ is the probability that 
\begin{enumerate}
    \item Run the first step in the case ${\verify^{\ro, \ket{\pspace}}(pk, (s, \cdot))}$ on ${\rho_s^{(t)}}$.
    \item Furthermore run $\verify^{D_j, \ket{\pspace}}$ in a purified view where the input of $\verify$ is the register for the public key, the register for the serial number, and the register for $\rho_s^{(t)}$ (The input is the same as the case ${\verify^{\ro, \ket{\pspace}}(pk, (s, \cdot))}$ on ${\rho_s^{(t)}}$).
    \item Measure the outcome register of the above $\verify^{D_j, \ket{\pspace}}$ and obtain $\accept$.
\end{enumerate}

\paragraph{$\bm{\verify^{\ro, \ket{\pspace}}(pk, (s, \cdot))}$ on $\bm{\sigma_{D_j}}$}

The acceptance probability is the probability that 
\begin{enumerate}
    \item Run the first step in the case ${\verify^{\ro, \ket{\pspace}}(pk, (s, \cdot))}$ on ${\rho_s^{(t)}}$.
    \item Furthermore run $\verify^{\ro, \ket{\pspace}}$ in a purified view where the input of $\verify$ is the register for the public key, the register for the serial number, and the register for $\sigma_{D_j}$ (It's the first register of the output state of $\adversary$).
    \item Measure the outcome register of the above $\verify^{\ro, \ket{\pspace}}$ and obtain $\accept$.
\end{enumerate}

\paragraph{$\bm{\verify^{D_j, \ket{\pspace}}(pk, (s, \cdot))}$ on $\bm{\sigma_{D_j}}$}

The acceptance probability is the probability that 
\begin{enumerate}
    \item Run the first step in the case ${\verify^{\ro, \ket{\pspace}}(pk, (s, \cdot))}$ on ${\rho_s^{(t)}}$.
    \item Furthermore run $\verify^{D_j, \ket{\pspace}}$ in a purified view where the input of $\verify$ is the register for the public key, the register for the serial number, and the register for $\sigma_{D_j}$ (The input is the same as the case ${\verify^{\ro, \ket{\pspace}}(pk, (s, \cdot))}$ on ${\sigma_{D_j}}$).
    \item Measure the outcome register of the above $\verify^{D_j, \ket{\pspace}}$ and obtain $\accept$.
\end{enumerate}

\subsection{Compress and Decompress}\label{subsec:compress}

Intuitively, $\ket{\hat 0}$ position in $\reg{F}$ is a uniform superposition of the range and it is unentangled with all other things, so it can be seen as choosing a value from the range uniformly at random independently, which is exactly what the simulation does. It is an analog of those positions that are never asked during the sequence of algorithms in the purely classical query case. 

In this subsection, we will show how to extract an analog of $D_k$ and $D_\keygen \cup D_\mint - D_k$ from the pure state. Roughly speaking, the recorded classical queries are an analog of $D_k$ and we will compress the register $\reg{F}$ to extract those non-$\ket{\hat 0}$ positions outside $D_k$ to form our analog of $D_\keygen \cup D_\mint - D_k$.

As it's easier to write down and analyze the inverse operation of compress, we first give a formal description of decompress unitary $\decomp$. Recall that $\reg{D}_{\ro}$ stores all the classical queries to $\ro$. $\reg{F}$ is a register for the random function. Define 
$$\decomp: \ket{D_F}_{\reg{F}} \ket{D_{\ro}}_{\reg{D}_{\ro}} \rightarrow \ket{f_0, f_1, \cdots, f_{2^l - 1}}_{\reg{F}} \ket{D_{\ro}}_{\reg{D}_{\ro}}$$
where $\ket{D_F}_{\reg F}$ can be written as a sequence of pairs $\ket{(x_1, \widehat{y_1}), (x_2, \widehat{y_2}), \cdots, (x_{k'}, \widehat{y_{k'}})}_{\reg{F}}$, $\ket{D_{\ro}}_{\reg{D}_{\ro}}$ can be written as a sequence of pairs $\ket{(x_1', z_1), (x_2', z_2), \cdots, (x_k', z_k)}$ and the input $\ket{D_F}_{\reg{F}} \ket{D_{\ro}}_{\reg{D}_{\ro}}$ satisfies 
\begin{itemize}
    \item if $x_i' = x_j'$, then $z_i = z_j$;
    \item $x_1 < x_2 < \cdots < x_{k'}, \widehat{y_i} \ne \widehat{0}$;
    \item $\forall i, j, x_j \ne x_i'$;
\end{itemize}
and the output satisfies
\begin{itemize}
    \item If $x_j' = i$, then $f_i = z_j$;
    \item If $x_j = i$, then $f_i = \widehat{y_j}$;
    \item If $\forall j, x_j \ne i, x_j' \ne i$, then $f_i = \widehat{0}$.
\end{itemize} 

Roughly speaking, we fill $f_0, f_1, \cdots, f_{2^l - 1}$ by looking at the pairs $(x_1', z_1), (x_2', z_2), \cdots, (x_k', z_k)$ and $(x_1, \widehat{y_1}), (x_2, \widehat{y_2}), \cdots, (x_{k'}, \widehat{y_{k'}})$. And we fill all the remaining positions with $\hat{0}$.

Here our register $\reg{F}$ also have enough space. As our random function only has one-bit outputs, $z_1, z_2, \cdots, z_k, y_1, y_2, \cdots, y_{k'} \in \{0, 1\}$. Recall that $\ket{\hat{0}}=\ket{+} = \frac{1}{\sqrt{2}}(\ket{0} + \ket{1}), \ket{\hat{1}}=\ket{-} = \frac{1}{\sqrt{2}}(\ket{0} - \ket{1})$. One can check that each two inputs in the above form are orthogonal and they are mapped to orthogonal outputs. So we can define the outputs of other inputs that are not in the above form so that $\decomp$ is a unitary.

\paragraph{More Notations} For simplicity, when we write $D_{\ro}$, we mean a sequence of consistent pairs $(x_1', z_1), (x_2', z_2), \cdots, (x_k', z_k)$ by default. By $x \in D_{\ro}$, we mean $\exists i$, $x_i' = x$. By $D_{\ro}(x)$ where $x \in D_{\ro}$, we mean $z_i$ where $x_i' = x$. When we write $D_F$, we mean a sequence $(x_1, \widehat{y_1}), (x_2, \widehat{y_2}), \cdots, (x_{k'}, \widehat{y_{k'}})$ that satisfies the second item of the input requirements above. By $x \in D_F$, we mean $\exists i$, $x_i = x$. By $D_F(x)$ where $x \in D_F$, we mean $\widehat{y_i}$ where $x_i = x$. By $\widehat{D_F(x)}$ where $x \in D_F$, we mean $y_i$ where $x_i = x$. By $D_F - x$, we mean the sequence we obtain after deleting $x_i, \widehat{y_i}$ from the sequence $D_F$ where $x_i = x$. Define $D_{\ro} \cap D_F = \{x: x \in D_{\ro} \text{ and } x \in D_F\}$.

The inverse operation of the above unitary $\decomp$ is compress, which can take our database $D_{\ro}$ and the truth table in register $\reg{F}$ as inputs and compress them into two databases $D_{\ro}$ and $D_F$. Define it as $\comp := \decomp^{\dagger}$. These two unitaries enable us to change our view between the decompressed one (a database for classical queries and a truth table) and the compressed one (a database for classical queries and another database). Here is a picture to illustrate this. $\ket{\theta}$ is an arbitrary state without compression. $\decompV{U}$ is a unitary in the decompressed view (It takes a state without compression as input and outputs a state without compression). Then $\compV{U} := \comp \decompV{U} \decomp$ is a compressed view version of $\decompV{U}$ (It takes a state after compression as input and outputs a state after compression). From now on, when we write unitary $\compV{\cdot}$, we mean it is in the compressed view.

\begin{center}
\begin{tikzpicture}[node distance=2.8cm, auto]\label{figure:compress}
\node (Input) {$\ket{\theta}$};
\node (CompressedInput)[right of=Input] {$\comp \ket{\theta}$};
\node (Output) [below of=Input] {$\decompV{U}\ket{\theta}$};
\node (CompressedOutput) [right of=Output] {$\comp \decompV{U} \ket{\theta}$};

\draw[transform canvas={yshift=0.5ex},->] (Input) --(CompressedInput) node[above,midway] {\tiny $\comp$};
\draw[transform canvas={yshift=-0.5ex},->](CompressedInput) -- (Input) node[below,midway] {\tiny $\decomp$}; 
\draw[->](Input) to node[swap] {$\decompV{U}$}(Output);
\draw[transform canvas={yshift=0.5ex},->] (Output) --(CompressedOutput) node[above,midway] {\tiny $\comp$};
\draw[transform canvas={yshift=-0.5ex},->](CompressedOutput) -- (Output) node[below,midway] {\tiny $\decomp$}; 
\draw[->](CompressedInput) to node {$\compV{U} = \comp \decompV{U} \decomp$}(CompressedOutput);
\end{tikzpicture}
\end{center}

Readers can treat $D_{\ro}$ as an analog of database $D_k$ in \Cref{sec:AttackClassicalAccess} and treat $D_F$ as an analog of databases $D_{\keygen} \cup D_{\mint} - D_k$ in \Cref{sec:AttackClassicalAccess}. Roughly speaking, $D_{\ro}$ stores our classical queries. The query positions in $D_F$ are those asked by $\keygen$ and $\mint$, but not recorded in $D_{\ro}$. We can understand this analog better after taking a look of the following query unitaries in the compressed view. 

Recall the unitaries $U_C$, $U_R$ and $U_D$ from \Cref{subsec:purifiedView}. They are for classical query, classical query while recording and simulated classical query respectively. Their compressed versions are the unitaries 
$\compV{U_C} := \comp \decompV{U_C} \decomp$,
$\compV{U_R} := \comp \decompV{U_R} \decomp$ and
$\compV{U_D} := \comp \decompV{U_D} \decomp$.

From the description of $\decomp$, we can get for $D_F \cap D_{\ro} = \emptyset,$
\begin{align*}
    &\compV{U_C}(\ket{x}_{\reg{Q}}\ket{0}_{\reg{A}}\ket{D_F}_{\reg{F}}\ket{D_{\ro}}_{\reg{D}_{\ro}})\\
    =&\begin{cases}
        \ket{x}_{\reg{Q}}\ket{D_{\ro}(x)}_{\reg{A}}\ket{D_F}_{\reg{F}}\ket{D_{\ro}, (x, D_{\ro}(x))}_{\reg{D}_{\ro}} &x \in D_{\ro}\\
        \ket{x}_{\reg{Q}}\frac{1}{\sqrt 2}\sum_{z = 0}^1\ket{z}_{\reg{A}}\ket{D_F}_{\reg{F}}\ket{D_{\ro}, (x, z)}_{\reg{D}_{\ro}} &x \notin D_{\ro}, x \notin D_F\\
        \ket{x}_{\reg{Q}}\frac{1}{\sqrt 2}\sum_{z = 0}^1 (-1)^{z\widehat{D_F(x)}} \ket{z}_{\reg{A}}\ket{D_F - x}_{\reg{F}}\ket{D_{\ro}, (x, z)}_{\reg{D}_{\ro}} &x \notin D_{\ro}, x \in D_F
    \end{cases}
\end{align*}
\begin{align*}
    &\compV{U_R}(\ket{x}_{\reg{Q}}\ket{0}_{\reg{A}}\ket{D_{\adversary}}_{\reg{D}_{\adversary}}\ket{D_F}_{\reg{F}}\ket{D_{\ro}}_{\reg{D}_{\ro}})\\
    =&\begin{cases}
        \ket{x}_{\reg{Q}}\ket{D_{\ro}(x)}_{\reg{A}}\ket{D_{\adversary}, (x, D_{\ro}(x))}_{\reg{D}_{\adversary}}\ket{D_F}_{\reg{F}}\ket{D_{\ro}, (x, D_{\ro}(x))}_{\reg{D}_{\ro}} &x \in D_{\ro}\\
        \ket{x}_{\reg{Q}}\frac{1}{\sqrt 2}\sum_{z = 0}^1\ket{z}_{\reg{A}}\ket{D_{\adversary}, (x, z)}_{\reg{D}_{\adversary}}\ket{D_F}_{\reg{F}}\ket{D_{\ro}, (x, z)}_{\reg{D}_{\ro}} &x \notin D_{\ro}, x \notin D_F\\
        \ket{x}_{\reg{Q}}\frac{1}{\sqrt 2}\sum_{z = 0}^1 (-1)^{z\widehat{D_F(x)}} \ket{z}_{\reg{A}} \ket{D_{\adversary}, (x, z)}_{\reg{D}_{\adversary}} \ket{D_F - x}_{\reg{F}}\ket{D_{\ro}, (x, z)}_{\reg{D}_{\ro}} &x \notin D_{\ro}, x \in D_F
    \end{cases}
\end{align*}

Since $U_D$ does not act on $\reg{F}$ and $\reg{D}_{\ro}$, $\compV{U_D} = U_D$. Thus 
\begin{align*}
    &\compV{U_D}(\ket{x}_{\reg{Q}}\ket{0}_{\reg{A}}\ket{D_{\adversary}}_{\reg{D}_{\adversary}})\\
    =& \begin{cases}
    \ket{x}_{\reg{Q}}\ket{D_{\adversary}(x)}_{\reg{A}}\ket{D_{\adversary}, (x, D_{\adversary}(x))}_{\reg{D}_{\adversary}}, & x \in D_{\adversary}\\
    \ket{x}_{\reg{Q}}\frac{1}{\sqrt 2}\sum_{z = 0}^1 \ket{z}_{\reg{A}}\ket{D_{\adversary}, (x, z)}_{\reg{D}_{\adversary}}, & x \notin D_{\adversary} 
        \end{cases}
\end{align*}

By the description of $\compV{U_C}$ and $\compV{U_R}$, whenever we ask a classical query on input $x \in D_{\ro}$, we answer it with our database $D_{\ro}$ and record $(x, D_{\ro}(x))$ for another time; whenever we ask a classical query on input $x \notin D_{\ro} \cup D_F$, we answer it with a random $z$ and record $(x, z)$ in our database for later use; whenever we ask a classical query on input $x \in D_F$, we actually copy the answer from the $D_F$, record it, and remove $x$ from $D_F$. The above three cases is analogous to the classical on-the-fly simulation where the query to $\ro$ inside $D_k$ can be answered by $D_k$, the query to $\ro$ inside $D_{\keygen} \cup D_{\mint} - D_k$ should be answered consistently by $D_{\keygen} \cup D_{\mint} - D_k$ and we can give a random answer to the query to outside $(D_{\keygen} \cup D_{\mint} - D_k) \cup D_k$. It is worth pointing out that $\compV{U_C}$ and $\compV{U_R}$ maintain the property that $D_{\ro} \cap D_F$ is empty (analogous to $\left(D_{\keygen}\cup D_{\mint} - D_k\right) \cap D_k = \emptyset$).

Furthermore, recall that $\decompV{U_{\verify}} = U_q \decompV{U_C} U_{q - 1} \cdots U_1 \decompV{U_C} U_0$ and $U_i$ does not act on $\reg{F}$ or $\reg{D}_{\ro}$. Thus $\compV{U_i} = U_i$,  $\compV{U_{\verify}} = \compV{U_q} \compV{U_C} \compV{U_{q - 1}} \cdots \compV{U_1} \compV{U_C} \compV{U_0} = U_q \compV{U_C} U_{q - 1} \cdots U_1 \compV{U_C} U_0.$ Similarly, recall that the purified version of $\verify^{D_j, \ket{\pspace}}$ is $U_{\simulation} = U_q U_D U_{q - 1}\cdots U_D U_0$. Thus $\compV{U_{\simulation}} = U_q \compV{U_D} U_{q - 1}\cdots \compV{U_D} U_0 = U_{\simulation}$ (Because $\compV{U_D} = U_D$).

\subsection{Analysis of $\adversary^{\ro, \ket{\pspace}}$}\label{subsec:QuantumAnalysis}

Here we analyze the acceptance probability of $\verify^{\ro, \ket{\pspace}}$ on the output of our $\adversary^{\ro, \ket{\pspace}}$. We reuse our ideas in \Cref{sec:AttackClassicalAccess}. Readers are encouraged to refer to our notation table in \Cref{NotationTable} when confused about the notations. 

The following proposition can be seen as an analog of \Cref{claim:UpdateDB} except that we work on a general state instead of solely analyzing $\sigma_{D_k}$. It basically says that when the behavior of $\verify^{\ro, \ket{\pspace}}$ (corresponding to $\compV{U_{\verify}} = U_q\compV{U_C}\cdots \compV{U_C}U_0$) is far from a simulation $\verify^{D, \ket{\pspace}}$ (corresponding to $\compV{U_{\simulation}} = U_qU_D\cdots U_DU_0$), then the number of pairs in $\reg{F}$ (analogous to $\length{D_{\keygen} \cup D_{\mint} - D_k}$) will drop a lot after the verification. The intuition is that roughly speaking, $\compV{U_C}$ and $U_D$ only behave differently when given a query position $x \in D_F$, in which case $x$ will be excluded from $D_F$ after applying $\compV{U_C}$. So it results in a decrement of the number of pairs in $\reg{F}$. Formally:

\begin{proposition}\label{claim:quantumprogress}
Denote $O$ to be the observable corresponding to the number of pairs in $\reg{F}$. To be more specific, $O = \sum_{D_F}\length{D_F}\ketbra{D_F}{D_F}_{\reg{F}}$ where $\length{D_F}$ is the number of pairs in $D_F$.

For a state $\ket{\phi}$ in the following form (i.e. it's in the compressed view and the contents in $\reg{D}_{\ro}$ and $\reg{D}_{\adversary}$ are the same),
$$\ket{\phi} = \sum_{\substack{pk, s, m, D, D_F, g \\ \text{s.t. }D \cap D_F = \emptyset}}\alpha_{pk, s, m, D, D_F, g}\ket{pk}_{\reg{Pk}}\ket{s}_{\reg{S}}\ket{m}_{\reg{M}}\ket{D}_{\reg{D}_{\adversary}}\ket{D_F}_{\reg{F}}\ket{D}_{\reg{D}_{\ro}}\ket{g}_{\reg{G}}.$$

Let $\prob\left[\compV{U_{\verify}} \text{ accepts when running on $\ket{\phi}$}\right]$ be the acceptance probability of $\compV{U_{\verify}}$ when the public key, the serial number and the alleged money state are in $\reg{Pk}, \reg{S}$ and $\reg{M}$ respectively (It also equals to the acceptance probability of $U_{\verify}$ on $\decomp \ket{\phi}$ because $\compV{U_{\verify}}\ket{\phi} = \comp U_{\verify} (\decomp \ket{\phi})$ and $\comp$ does not act on the output bit). 

Let $\prob\left[\compV{U_{\simulation}} \text{ accepts when running on $\ket{\phi}$}\right]$ be the acceptance probability of $\compV{U_{\simulation}}$ when the public key, the serial number, the alleged money state and the database for simulating the random oracle are in $\reg{Pk}, \reg{S}, \reg{M}$ and $\reg{D}_{\adversary}$ respectively (It also equals to the acceptance probability of $U_{\simulation}$ on $\decomp \ket{\phi}$).
\begin{align*}
    &\left|\prob\left[\compV{U_{\verify}} \text{ accepts when running on $\ket{\phi}$}\right] - \prob\left[\compV{U_{\simulation}} \text{ accepts when running on $\ket{\phi}$}\right]\right|\\
    \le &\TraceDist{\Tr_{\reg{F}\reg{D}_{\adversary}\reg{D}_{\ro}}(\compV{U_{\verify}} \ketbra{\phi}{\phi}\compV{U_{\verify}}^{\dagger})}
    {\Tr_{\reg{F}\reg{D}_{\adversary}\reg{D}_{\ro}}(\compV{U_{\simulation}} \ketbra{\phi}{\phi}\compV{U_{\simulation}}^{\dagger})}\\
    \le & 6 \sqrt{q\left(\Tr(O\ketbra{\phi}{\phi}) - \Tr(O\compV{U_{\verify}}\ketbra{\phi}{\phi}\compV{U_{\verify}}^{\dagger})\right)}
\end{align*}
\end{proposition}

\begin{proof}

The first inequality follows immediately from the fact that we can measure a qubit (not in $\reg{F}\reg{D}_{\adversary}\reg{D}_{\ro}$) of $\compV{U_{\verify}}\ket{\phi}$ and $\compV{U_{\simulation}}\ket{\phi}$ to obtain whether they accept and the fact that we can not distinguish two states with probability greater than their trace distance.

As for the second inequality, recall that $\compV{U_{\verify}} = U_q\compV{U_C}\cdots \compV{U_C}U_0$ and $\compV{U_{\simulation}} = U_qU_D\cdots U_DU_0$. Let $U_D'$ be the same as $U_D$ except that it uses the contents in $D_{\ro}$ for simulation instead of the contents in $D_{\adversary}$. To be more specific, 
\begin{align*}
    U_D'(\ket{x}_{\reg{Q}}\ket{0}_{\reg{A}}\ket{D_{\ro}}_{\reg{D}_{\ro}})
    = \begin{cases}
        \ket{x}_{\reg{Q}}\ket{D_{\ro}(x)}_{\reg{A}}\ket{D_{\ro}, (x, D_{\ro}(x))}_{\reg{D}_{\ro}}, & x \in D_{\ro}\\
        \ket{x}_{\reg{Q}}\frac{1}{\sqrt 2}\sum_{z = 0}^1 \ket{z}_{\reg{A}}\ket{D_{\ro}, (x, z)}_{\reg{D}_{\ro}}, & x \notin D_{\ro} 
        \end{cases}
\end{align*}

Define $\ket{\phi_j} = \compV{U_C}U_{j - 1}\cdots \compV{U_C} U_0\ket{\phi}$ where $0 \le j \le q$. In order to analyze the difference of $\compV{U_{\verify}}$ and $\compV{U_{\simulation}}$ on $\ket{\phi}$, it's enough to analyze the difference between one true query and one simulated query. Formally, 
\begin{align*}
    &\TraceDist{\Tr_{\reg{F}\reg{D}_{\adversary}\reg{D}_{\ro}}(\compV{U_{\verify}} \ketbra{\phi}{\phi}\compV{U_{\verify}}^{\dagger})}
    {\Tr_{\reg{F}\reg{D}_{\adversary}\reg{D}_{\ro}}(\compV{U_{\simulation}} \ketbra{\phi}{\phi}\compV{U_{\simulation}}^{\dagger})}\\
    = &\TraceDist{\Tr_{\reg{F}\reg{D}_{\adversary}\reg{D}_{\ro}}(U_q\ketbra{\phi_q}{\phi_q}U_q^{\dagger})}
    {\Tr_{\reg{F}\reg{D}_{\adversary}\reg{D}_{\ro}}(U_qU_D'\cdots U_D' U_0\ketbra{\phi_{0}}{\phi_{0}}U_0^{\dagger}U_D'^{\dagger}\cdots U_D'^{\dagger}U_q^{\dagger})}\\
    \le & {\sum_{j = 0}^{q - 1} \TraceDist{\Tr_{\reg{F}\reg{D}_{\adversary}\reg{D}_{\ro}}(U_qU_D'\cdots U_{j + 1} \ketbra{\phi_{j + 1}}{\phi_{j + 1}}U_{j + 1}^{\dagger}U_D'^{\dagger}\cdots U_q^{\dagger})}{\Tr_{\reg{F}\reg{D}_{\adversary}\reg{D}_{\ro}}(U_qU_D'\cdots U_j\ketbra{\phi_{j}}{\phi_{j}}U_j^{\dagger}U_D'^{\dagger}\cdots U_q^{\dagger})}}\\
    \le & \sum_{j = 0}^{q - 1} \TraceDist{\compV{U_C}U_j\ketbra{\phi_{j}}{\phi_{j}}U_j^{\dagger}\compV{U_C}^{\dagger}}{U_D'U_j\ketbra{\phi_{j}}{\phi_{j}}U_j^{\dagger}U_D'^{\dagger}}
\end{align*}
where we use the fact that $\ket{\phi}$ have the same contents on $\reg{D}_{\adversary}$ and $\reg{D}_{\ro}$ and thus $\Tr_{\reg{F}\reg{D}_{\adversary}\reg{D}_{\ro}}(\compV{U_{\simulation}} \ketbra{\phi}{\phi}\compV{U_{\simulation}}^{\dagger})$ equals to $\Tr_{\reg{F}\reg{D}_{\adversary}\reg{D}_{\ro}}(U_qU_D'\cdots U_D' U_0\ketbra{\phi_{0}}{\phi_{0}}U_0^{\dagger}U_D'^{\dagger}\cdots U_D'^{\dagger}U_q^{\dagger})$.

$\compV{U_C}$ and $U_D'$ act differently only when the query position is inside $D_F$. So the difference between one true query and one simulated query can be bounded by the weight of queries inside $D_F$, which equals the decrement of the number of pairs in $\reg{F}$ after the query. Formally, we give the following lemma, whose proof is deferred to \Cref{appendix:proof1}.

\begin{lemma}\label{lemma:diffofquery}
$\TraceDist{\compV{U_C}U_j\ketbra{\phi_{j}}{\phi_{j}}U_j^{\dagger}\compV{U_C}^{\dagger}}{U_D'U_j\ketbra{\phi_{j}}{\phi_{j}}U_j^{\dagger}U_D'^{\dagger}} \le 6\sqrt{\Tr(O\ketbra{\phi_j}{\phi_j}) - \Tr(O\ketbra{\phi_{j + 1}}{\phi_{j + 1}})}.$
\end{lemma}

We can insert \Cref{lemma:diffofquery} into the above inequality and obtain
\begin{align*}
&\TraceDist{\Tr_{\reg{F}\reg{D}_{\adversary}\reg{D}_{\ro}}(\compV{U_{\verify}} \ketbra{\phi}{\phi}\compV{U_{\verify}}^{\dagger})}
{\Tr_{\reg{F}\reg{D}_{\adversary}\reg{D}_{\ro}}(\compV{U_{\simulation}} \ketbra{\phi}{\phi}\compV{U_{\simulation}}^{\dagger})}\\
\le & \sum_{j = 0}^{q - 1} 6\sqrt{\Tr(O\ketbra{\phi_j}{\phi_j}) - \Tr(O\ketbra{\phi_{j + 1}}{\phi_{j + 1}})}\\
\le & 6\sqrt{q\sum_{j = 0}^{q - 1}\left(\Tr(O\ketbra{\phi_j}{\phi_j}) - \Tr(O\ketbra{\phi_{j + 1}}{\phi_{j + 1}})\right)} \text{   (Cauchy–Schwarz inequality)}\\
= & 6 \sqrt{q\left(\Tr(O\ketbra{\phi}{\phi}) - \Tr(O\compV{U_C}U_{q - 1}\cdots \compV{U_C}U_0\ketbra{\phi}{\phi}U_0^{\dagger}\compV{U_C}^{\dagger}\cdots U_{q - 1}^{\dagger}\compV{U_C}^{\dagger})\right)}\\
= & 6 \sqrt{q\left(\Tr(O\ketbra{\phi}{\phi}) - \Tr(OU_q\compV{U_C}\cdots \compV{U_C}U_0\ketbra{\phi}{\phi}U_0^{\dagger}\compV{U_C}^{\dagger}\cdots \compV{U_C}^{\dagger}U_q^{\dagger})\right)}\\
= & 6 \sqrt{q\left(\Tr(O\ketbra{\phi}{\phi}) - \Tr(O \compV{U_{\verify}}\ketbra{\phi}{\phi}\compV{U_{\verify}}^{\dagger})\right)}
\end{align*}
where we use again the fact that $U_q$ does not act on $\reg{F}$ and thus $O$ commutes with $U_q$.

\end{proof}

The next claim is an analog of \Cref{claim:SimOnRho}. Basically, it argues that on average, $\rho_s^{(t)}$ is a good witness state for the simulation with database $D_j$ even when $\keygen$ and $\mint$ can make quantum queries to $\ro$.

The intuition of the proof is the following: from \Cref{subsec:purifiedView}, the difference between the acceptance probabilities of $\verify^{\ro, \ket{\pspace}}(pk, (s, \cdot))$ and $\verify^{D_j, \ket{\pspace}}(pk, (s, \cdot))$ on $\rho_s^{(t)}$ is the difference between running $\verify^{\ro, \ket{\pspace}}$ and $\verify^{D_j, \ket{\pspace}}$ in the purified view on the same state (i.e. the difference between applying $U_{\verify}$ and $U_{\simulation}$ on the same state), which can be transformed to the compressed view and by \Cref{claim:quantumprogress} can be bounded by the decrement of the number of pairs in $\reg{F}$ after the verification. Roughly speaking, the decrement equals to the number of pairs in $D_F$ asked during the verification. But we randomize $t$, so running another verification on $\rho_s^{(t)}$ should not decrease the number of pairs in $\reg{F}$ too much on average. Formally:

\begin{claim}\label{claim:QuantumSimOnRho}
    $$\left|\prob \left[\verify^{\ro, \ket{\pspace}}(pk, (s, \rho_s^{(t)})) \text{ accepts} \right] - \prob \left[\verify^{D_j, \ket{\pspace}}(pk, (s, \rho_s^{(t)})) \text{ accepts} \right]\right| \le 6\sqrt{\frac{q(n)q'(n)}{T(n)}}$$ where the probabilities are taken over the randomness of $\ro$, the randomness of the inputs to the adversary and the randomness of the adversary (so $t$ and $j$ are both randomized).
\end{claim}

\begin{proof}

From \Cref{subsec:purifiedView}, these two probabilities can be analyzed by running a sequence of algorithms in the purified view. As $\comp\decomp = I$, we can insert $\comp$ and $\decomp$ between the algorithms. So these two probabilities can also be analyzed by running the sequence of algorithms in the compressed view.

Let $\ket{\phi}$ be the whole pure state we obtain by applying the unitaries $\compV{U_{\keygen}} = \comp U_{\keygen} \decomp$, $\compV{U_{\mint}} = \comp U_{\mint} \decomp$ and $\compV{U_{\adversary}} = \comp U_{\adversary} \decomp$ to the state $\ket{1^n}\ket{\emptyset}_{\reg{D}_{\adversary}}\ket{\emptyset}_{\reg{F}}\ket{\emptyset}_{\reg{D}_{\ro}}$ along with enough ancillas (i.e. the compressed version of the first step of ``${\verify^{\ro, \ket{\pspace}}(pk, (s, \cdot))}$ on ${\rho_s^{(t)}}$'' in \Cref{subsec:purifiedView}). Let $\reg{M}$ be the register corresponding to $\rho_s^{(t)}$. 

It's easy to see that every classical query is recorded by the adversary. That is, $\ket{\phi}$ has the same contents in $\reg{D}_{\adversary}$ and $\reg{D}_{\ro}$. From \Cref{claim:quantumprogress}, 
\begin{align*}
    &\left|\prob \left[\verify^{\ro, \ket{\pspace}}(pk, (s, \rho_s^{(t)})) \text{ accepts} \right] - \prob \left[\verify^{D_j, \ket{\pspace}}(pk, (s, \rho_s^{(t)})) \text{ accepts} \right]\right|\\
    \le & 6 \sqrt{q\left(\Tr(O\ketbra{\phi}{\phi}) - \Tr(O \compV{U_{\verify}}\ketbra{\phi}{\phi}\compV{U_{\verify}}^{\dagger})\right)}
\end{align*}

Denote $\compV{U_\update}$ to be unitary that describes our update phase in the compressed view. Formally,
$\compV{U_{\update}} = \comp U_{\update} \decomp = \sum_{j = 0}^{N(n) - 1}(\compV{U_{\verify}'}\compV{U_{\syn}})^j \otimes \ketbra{j}{j}_{\reg{J}}$ where $\compV{U_{\verify}'}$ is running on the state synthesized by $\compV{U_{\syn}}$.

Let $\ket{\psi}$ be the whole pure state we obtain by applying the unitaries $\compV{U_{\keygen}}$, $\compV{U_{\mint}}$ and $\compV{U_{\test}} := \comp U_{\test}\decomp$ to the state $\ket{1^n}\ket{\emptyset}_{\reg{D}_{\adversary}}\ket{\emptyset}_{\reg{F}}\ket{\emptyset}_{\reg{D}_{\ro}} \frac{1}{\sqrt{T(n)}}\sum_{t = 0}^{T(n) - 1}\ket{t}_{\reg{T}}\frac{1}{\sqrt{N(n)}}\sum_{j = 0}^{N(n) - 1}\ket{j}_{\reg{J}}$ along with enough ancillas (i.e. running the compressed version of the first step of ``${\verify^{\ro, \ket{\pspace}}(pk, (s, \cdot))}$ on ${\rho_s^{(t)}}$'' until the end of the \textbf{test} phase of $\adversary$). That is to say, $\ket{\phi} = U_{\syn_2}U_{\syn_1}\compV{U_{\update}}\ket{\psi}$ where $U_{\syn_i}$ synthesizes the $i^{th}$ alleged banknote. Notice that $U_{\syn_i}$ only reads the public key, the serial number and acts on $\reg{D}_{\adversary}$ and some fresh ancillas. So $U_{\syn_i}$ commutes with $O$ and $\compV{U_{\verify}}$. (Recall that $\compV{U_{\verify}}$ is the one that does not record queries for $\adversary$. That is, it does not act on $\reg{D}_{\adversary}$.) Thus 
\begin{align*}
    &\Tr(O\ketbra{\phi}{\phi}) - \Tr(O \compV{U_{\verify}}\ketbra{\phi}{\phi}\compV{U_{\verify}}^{\dagger})\\
    =& \Tr(O\compV{U_{\syn_2}}\compV{U_{\syn_1}}\compV{U_{\update}} \ketbra{\psi}{\psi} \compV{U_{\update}}^{\dagger} \compV{U_{\syn_1}}^{\dagger}\compV{U_{\syn_2}}^{\dagger}) - \Tr(O \compV{U_{\verify}} \compV{U_{\syn_2}}\compV{U_{\syn_1}}\compV{U_{\update}} \ketbra{\psi}{\psi} \compV{U_{\update}}^{\dagger}\compV{U_{\syn_1}}^{\dagger}\compV{U_{\syn_2}}^{\dagger}\compV{U_{\verify}}^{\dagger})\\
    = &\Tr(O\compV{U_{\update}} \ketbra{\psi}{\psi} \compV{U_{\update}}^{\dagger}) - \Tr(O \compV{U_{\verify}}\compV{U_{\update}} \ketbra{\psi}{\psi} \compV{U_{\update}}^{\dagger}\compV{U_{\verify}}^{\dagger})
\end{align*}

The next step is to bound the decrement of the pairs in $\reg{F}$ during the verification on $\compV{U_{\update}}\ket{\psi}$. However, the update phase between $\compV{U_{\verify}}$ and the randomized number of verifications in the test phase may bring some trouble. So we give the following lemma, which is the counterpart of the fact that in the classical case, the number of queries in $D_{\mint} \cup D_{\keygen} - D_k$ asked during a verification is no more than the number of queries in $D_{\mint} \cup D_{\keygen} - D$ asked during a verification.

\begin{lemma}\label{RemoveUpd}
    We use the same notation as above. Then
    \begin{align*}
        &\Tr(O\compV{U_{\update}} \ketbra{\psi}{\psi} \compV{U_{\update}}^{\dagger}) - \Tr(O \compV{U_{\verify}}\compV{U_{\update}} \ketbra{\psi}{\psi} \compV{U_{\update}}^{\dagger}\compV{U_{\verify}}^{\dagger})\\
        \le & \Tr(O\ketbra{\psi}{\psi}) - \Tr(O \compV{U_{\verify}} \ketbra{\psi}{\psi} \compV{U_{\verify}}^{\dagger})
    \end{align*}
\end{lemma}

The proof of the lemma is deferred to \Cref{appendix:proof2}.

Now let's bound the decrement of the pairs in $\reg{F}$ during the verification. From \Cref{RemoveUpd},  
\begin{align*}
    &\left|\prob \left[\verify^{\ro, \ket{\pspace}}(pk, (s, \rho_s^{(t)})) \text{ accepts} \right] - \prob \left[\verify^{D_j, \ket{\pspace}}(pk, (s, \rho_s^{(t)})) \text{ accepts} \right]\right|\\
    \le & 6 \sqrt{q\left(\Tr(O\compV{U_{\update}} \ketbra{\psi}{\psi} \compV{U_{\update}}^{\dagger}) - \Tr(O \compV{U_{\verify}}\compV{U_{\update}} \ketbra{\psi}{\psi} \compV{U_{\update}}^{\dagger}\compV{U_{\verify}}^{\dagger})\right)}\\
    \le & 6 \sqrt{q\left(\Tr(O\ketbra{\psi}{\psi}) - \Tr(O \compV{U_{\verify}} \ketbra{\psi}{\psi} \compV{U_{\verify}}^{\dagger})\right)}\\
    = & 6 \sqrt{q\left(\Tr(O\ketbra{\psi}{\psi}) - \Tr(O \compV{U_{\verify}'} \ketbra{\psi}{\psi} \compV{U_{\verify}'}^{\dagger})\right)}
\end{align*}
where we use the fact that $\compV{U_{\verify}'}$ is the same as $\compV{U_{\verify}}$ except that it records the query-answer pairs for $\adversary$.

Note that we can also write $\ket{\psi}$ as $\frac{1}{\sqrt{T(n)}}\sum_{t = 0}^{T(n) - 1}\ket{\psi^{(t)}}\ket{t}_{\reg{T}}$ where $\ket{\psi^{(t)}}$ is the state after we run $t$ iterations in the test phase. Then $\ket{\psi^{(t + 1)}} = \compV{U_{\verify}'}\ket{\psi^{(t)}}$. As $O$ and $\compV{U_{\verify}'}$ do not run on $\reg{T}$, tracing out $\reg{T}$ does not change the quantity. Therefore,

\begin{align*}
    &\left|\prob \left[\verify^{\ro, \ket{\pspace}}(pk, (s, \rho_s^{(t)})) \text{ accepts} \right] - \prob \left[\verify^{D_j, \ket{\pspace}}(pk, (s, \rho_s^{(t)})) \text{ accepts} \right]\right|\\
    \le & 6 \sqrt{q\left(\Tr(O\ketbra{\psi}{\psi}) - \Tr(O \compV{U_{\verify}'} \ketbra{\psi}{\psi} \compV{U_{\verify}'}^{\dagger})\right)}\\
    = & 6 \sqrt{q\left(\Tr(O\Tr_{\reg{T}}(\ketbra{\psi}{\psi})) - \Tr(O \compV{U_{\verify}'} \Tr_{\reg{T}}(\ketbra{\psi}{\psi}) \compV{U_{\verify}'}^{\dagger})\right)}\\
    = & 6 \sqrt{q\left( \frac{1}{T(n)}\sum_{t = 0}^{T(n) - 1}\Tr\left(O\ketbra{\psi^{(t)}}{\psi^{(t)}}\right) - \frac{1}{T(n)}\sum_{t = 0}^{T(n) - 1}\Tr\left(O\compV{U_{\verify}'}\ketbra{\psi^{(t)}}{\psi^{(t)}}\compV{U_{\verify}'}^{\dagger}\right)\right)}\\
    = & 6 \sqrt{q\left( \frac{1}{T(n)}\sum_{t = 0}^{T(n) - 1}\Tr\left(O\ketbra{\psi^{(t)}}{\psi^{(t)}}\right) - \frac{1}{T(n)}\sum_{t = 0}^{T(n) - 1}\Tr\left(O\ketbra{\psi^{(t + 1)}}{\psi^{(t + 1)}}\right)\right)}\\
    \le & 6 \sqrt{\frac{q}{T(n)}\Tr\left(O\ketbra{\psi^{(0)}}{\psi^{(0)}}\right)}\\
    \le & 6\sqrt{\frac{qq'}{T(n)}}
\end{align*}
where we use the property of compressed oracle techniques that after $q'(n)$ quantum queries, there are at most $q'(n)$ non-$\hat{0}$ elements in $\reg{F}$. $\ket{\psi^{(0)}}$ is just the state we obtain after we run $\keygen$ and $\mint$ (so there are at most $q'(n)$ quantum queries) and then apply the unitary $\comp$. So there are at most $q'(n)$ pairs in $\reg{F}$ of $\ket{\psi^{(0)}}$. Hence $\Tr\left(O\ketbra{\psi^{(0)}}{\psi^{(0)}}\right) \le q'(n)$.

\end{proof}

The next claim is an analog of summing over $k$ of \Cref{claim:UpdateDB} to get a bound for how well $\verify^{D_j, \ket{\pspace}}$ behaves on $\sigma_{D_j}$. The intuition of the proof is similar to \Cref{claim:QuantumSimOnRho}.

\begin{claim}\label{claim:quantumSim}
    $$\left|\prob \left[\verify^{\ro, \ket{\pspace}}(pk, (s, \sigma_{D_j})) \text{ accepts} \right] - \prob \left[\verify^{D_j, \ket{\pspace}}(pk, (s, \sigma_{D_j})) \text{ accepts} \right]\right| \le 6\sqrt{\frac{q(n)q'(n)}{N(n)}}$$ where the probabilities are taken over the randomness of $\ro$, the randomness of the inputs to the adversary and the randomness of the adversary (so $t$ and $j$ are both randomized).
\end{claim}

\begin{proof}
Similar to the proof of \Cref{claim:QuantumSimOnRho}, these two probabilities can be analyzed by running a sequence of algorithms specificed in \Cref{subsec:purifiedView} in the compressed view. Let $\ket{\phi}$ be the same state as in \Cref{claim:QuantumSimOnRho} except that now $\reg{M}$ is the first output register of $\adversary$ (corresponding to the first synthesized state $\sigma_{D_j}$). Let $\ket{\psi}$ be the state we obtain by running the compressed version of the first step of ``${\verify^{\ro, \ket{\pspace}}(pk, (s, \cdot))}$ on ${\rho_s^{(t)}}$'' in~\Cref{subsec:purifiedView} until the end of the \textbf{update} phase of $\adversary$. That is to say, $\ket{\phi} = U_{\syn_2}U_{\syn_1}\ket{\psi}$. Then $\ket{\phi}$ has the same contents in $\reg{D}_{\adversary}$ and $\reg{D}_{\ro}$.

Recall that $\compV{U_{\verify}}$ and $\compV{U_{\verify}'}$ both run the verification for the state in $\reg{M}$, i.e. the first synthesized state $\sigma_{D_j}$. Recall the fact that $\compV{U_{\verify}'}$ is the same as $\compV{U_{\verify}}$ except that it also records query-answer pairs for $\adversary$ and that $O$ only acts on $\reg{F}$. From \Cref{claim:quantumprogress}, 
\begin{align*}
    &\left|\prob \left[\verify^{\ro, \ket{\pspace}}(pk, (s, \sigma_{D_j})) \text{ accepts} \right] - \prob \left[\verify^{D_j, \ket{\pspace}}(pk, (s, \sigma_{D_j})) \text{ accepts} \right]\right|\\
    \le & 6 \sqrt{q\left(\Tr(O\ketbra{\phi}{\phi}) - \Tr(O \compV{U_{\verify}}\ketbra{\phi}{\phi}\compV{U_{\verify}}^{\dagger})\right)}\\
    = & 6 \sqrt{q\left(\Tr(OU_{\syn_2}U_{\syn_1}\ketbra{\psi}{\psi}U_{\syn_1}^{\dagger}U_{\syn_2}^{\dagger}) - \Tr(O \compV{U_{\verify}}U_{\syn_2}U_{\syn_1}\ketbra{\psi}{\psi}U_{\syn_1}U_{\syn_2}^{\dagger}\compV{U_{\verify}}^{\dagger})\right)}\\
    = & 6 \sqrt{q\left(\Tr(O\ketbra{\psi}{\psi}) - \Tr(O \compV{U_{\verify}}U_{\syn_1}\ketbra{\psi}{\psi}U_{\syn_1}^{\dagger}\compV{U_{\verify}}^{\dagger})\right)}\\
    = & 6 \sqrt{q\left(\Tr(O\ketbra{\psi}{\psi}) - \Tr(O \compV{U_{\verify}'}U_{\syn_1}\ketbra{\psi}{\psi}U_{\syn_1}^{\dagger}\compV{U_{\verify}'}^{\dagger})\right)}
\end{align*}
This is because $U_{\syn_2}$ acts on the second output register, and thus it commutes with $\compV{U_{\verify}}$ (it verifies the first output state). Moreover, both $U_{\syn_1}$ and $U_{\syn_2}$ do not act on $\reg{F}$, and thus commute with $O$.

Note that we can write $\ket{\psi} = \frac{1}{\sqrt{N(n)}}\sum_{j= 0}^{N(n) - 1}\ket{\psi^{(j)}}\ket{j}_{\reg{J}}$ where $\ket{\psi^{(j)}}$ is the state we obtain after we run a randomized number of iterations in the test phase and $j$ iterations in the update phase. Then $\ket{\psi^{(j + 1)}} = \compV{U_{\verify}'}U_{\syn}\ket{\psi^{(j)}}$. Notice that $O$, $U_{\syn_1}$ and $\compV{U_{\verify}'}$ do not run on $\reg{J}$. So tracing out $\reg{J}$ won't change the value. Therefore,
\begin{align*}
        &\left|\prob \left[\verify^{\ro, \ket{\pspace}}(pk, (s, \sigma_{D_j})) \text{ accepts} \right] - \prob \left[\verify^{D_j, \ket{\pspace}}(pk, (s, \sigma_{D_j})) \text{ accepts} \right]\right|\\
        \le&  6 \sqrt{q\left(\Tr(O\ketbra{\psi}{\psi}) - \Tr(O \compV{U_{\verify}'}U_{\syn_1}\ketbra{\psi}{\psi}U_{\syn_1}^{\dagger}\compV{U_{\verify}'}^{\dagger})\right)}\\
        = &  6 \sqrt{q\left(\Tr(O\Tr_{\reg{J}}(\ketbra{\psi}{\psi})) - \Tr(O \compV{U_{\verify}'}U_{\syn_1}\Tr_{\reg{J}}(\ketbra{\psi}{\psi})U_{\syn_1}^{\dagger}\compV{U_{\verify}'}^{\dagger})\right)}\\
        =  & 6 \sqrt{q\left(\frac{1}{N(n)}\sum_{j = 0}^{N(n) - 1}\Tr\left(O \ketbra{\psi^{(j)}}{\psi^{(j)}}\right) - \frac{1}{N(n)}\sum_{j = 0}^{N(n) - 1}\Tr\left(O\compV{U_{\verify}'} U_{\syn_1}\ketbra{\psi^{(j)}}{\psi^{(j)}} U_{\syn_1}^{\dagger} \compV{U_{\verify}'}^{\dagger}\right)\right)}\\
        =  & 6 \sqrt{q\left(\frac{1}{N(n)}\sum_{j = 0}^{N(n) - 1}\Tr\left(O \ketbra{\psi^{(j)}}{\psi^{(j)}}\right) - \frac{1}{N(n)}\sum_{j = 0}^{N(n) - 1}\Tr\left(O\ketbra{\psi^{(j + 1)}}{\psi^{(j + 1)}}\right)\right)}\\
        \le  & 6 \sqrt{\frac{q}{N(n)}\Tr\left(O \ketbra{\psi^{(0)}}{\psi^{(0)}}\right)}\\
        \le & 6 \sqrt{\frac{qq'}{N(n)}}
\end{align*}
where we use the property of compressed oracle techniques that after $q'(n)$ quantum queries, there are at most $q'(n)$ non-$\hat{0}$ elements in $\reg{F}$. $\ket{\psi^{(0)}}$ is just the state we obtain after we run $\keygen$ and $\mint$ (so there are at most $q'(n)$ quantum queries) and run the test phase of $\adversary$ (only classical queries with recording) in the compressed view. By the description of $\compV{U_R}$, the number of pairs in $\reg{F}$ can not increase when we make classical queries and record them. So there are at most $q'(n)$ pairs in $\reg{F}$ of $\ket{\psi^{(0)}}$. Hence $\Tr\left(O\ketbra{\psi^{(0)}}{\psi^{(0)}}\right) \le q'(n)$.
\end{proof}

Now let's combine the above results to prove \Cref{theorem:QuantumAnalysisOfA}.

\begin{proof}
Let $\epsilon = 0.01$, $b = 1 - \sqrt{1 - \delta_r + \epsilon}$, $a = 0.99b$, $T(n) = \frac{36q(n)q'(n)}{\epsilon^2}$ and $N(n) = \frac{q(n)q'(n)}{\epsilon^2(1 - \sqrt{1 - \delta_r + \epsilon})^4}$. 

From \Cref{claim:QuantumSimOnRho}, $\prob \left[\verify^{D_j, \ket{\pspace}}(pk, (s, \rho_s^{(t)})) \text{ accepts} \right] \ge \delta_r - \epsilon$. Hence using the same argument in \Cref{claim:SimOnSigma}, $\prob \left[\verify^{D_j, \ket{\pspace}}(pk, (s, \sigma_{D_j})) \text{ accepts} \right] \ge 0.99(1 - \sqrt{1 - \delta_r + \epsilon})^2$. From \Cref{claim:quantumSim}, $\prob \left[\verify^{\ro, \ket{\pspace}}(pk, (s, \sigma_{D_j})) \text{ accepts} \right] \ge 0.99(1 - \sqrt{1 - \delta_r + \epsilon})^2 - 6\sqrt{\frac{q(n)q'(n)}{N(n)}} \ge 0.9(1 - \sqrt{1 - \delta_r + \epsilon})^2.$ Thus by union bound, the outputs of our adversary pass two verifications simultaneously with probability at least $1.8(1 - \sqrt{1 - \delta_r + \epsilon})^2 - 1$, which is non-negligible when $\delta_r \ge 0.99$.
\end{proof}

That is, the adversary we construct gives a valid attack to $( \keygen^{{\ket{\ro}},\ket{\pspace}}, \mint^{{\ket{\ro}},\ket{\pspace}}, \allowbreak\verify^{{\ro},\ket{\pspace}})$ where $\delta_r \ge 0.99$ and $\delta_s = \negl(n)$, which establishes \Cref{theorem:quantumInsecurity}.

\printbibliography

@inproceedings{DG17,
  title={Identity-based encryption from the Diffie-Hellman assumption},
  author={D{\"o}ttling, Nico and Garg, Sanjam},
  booktitle={Annual International Cryptology Conference},
  pages={537--569},
  year={2017},
  organization={Springer}
}

@inproceedings{DQVWW21,
  title={Succinct LWE sampling, random polynomials, and obfuscation},
  author={Devadas, Lalita and Quach, Willy and Vaikuntanathan, Vinod and Wee, Hoeteck and Wichs, Daniel},
  booktitle={Theory of Cryptography Conference},
  pages={256--287},
  year={2021},
  organization={Springer}
}

@inproceedings{AKLLZ22,
  title={On the feasibility of unclonable encryption, and more},
  author={Ananth, Prabhanjan and Kaleoglu, Fatih and Li, Xingjian and Liu, Qipeng and Zhandry, Mark},
  booktitle={Annual International Cryptology Conference},
  pages={212--241},
  year={2022},
  organization={Springer}
}

@article{CMP20,
  title={Quantum copy-protection of compute-and-compare programs in the quantum random oracle model},
  author={Coladangelo, Andrea and Majenz, Christian and Poremba, Alexander},
  journal={arXiv preprint arXiv:2009.13865},
  year={2020}
}

@inproceedings{Simon98,
  title={Finding collisions on a one-way street: Can secure hash functions be based on general assumptions?},
  author={Simon, Daniel R},
  booktitle={International Conference on the Theory and Applications of Cryptographic Techniques},
  pages={334--345},
  year={1998},
  organization={Springer}
}

@inproceedings{RTV04,
  title={Notions of reducibility between cryptographic primitives},
  author={Reingold, Omer and Trevisan, Luca and Vadhan, Salil},
  booktitle={Theory of Cryptography Conference},
  pages={1--20},
  year={2004},
  organization={Springer}
}

@inproceedings{BM09,
  title={Merkle puzzles are optimal—an o (n 2)-query attack on any key exchange from a random oracle},
  author={Barak, Boaz and Mahmoody-Ghidary, Mohammad},
  booktitle={Annual International Cryptology Conference},
  pages={374--390},
  year={2009},
  organization={Springer}
}

@misc{AK22,
      author = {Prabhanjan Ananth and Fatih Kaleoglu},
      title = {A Note on Copy-Protection from Random Oracles},
      howpublished = {Cryptology ePrint Archive, Paper 2022/1109},
      year = {2022},
      note = {\url{https://eprint.iacr.org/2022/1109}},
      url = {https://eprint.iacr.org/2022/1109}
}

@inproceedings{Unr16,
  title={Computationally binding quantum commitments},
  author={Unruh, Dominique},
  booktitle={Annual International Conference on the Theory and Applications of Cryptographic Techniques},
  pages={497--527},
  year={2016},
  organization={Springer}
}

@article{CX21,
  title={Being a permutation is also orthogonal to one-wayness in quantum world: Impossibilities of quantum one-way permutations from one-wayness primitives},
  author={Cao, Shujiao and Xue, Rui},
  journal={Theoretical Computer Science},
  volume={855},
  pages={16--42},
  year={2021},
  publisher={Elsevier}
}

@article{zhandry2015note,
  title={A note on the quantum collision and set equality problems},
  author={Zhandry, Mark},
  journal={Quantum Information \& Computation},
  volume={15},
  number={7-8},
  pages={557--567},
  year={2015},
  publisher={Rinton Press, Incorporated Paramus, NJ}
}

@article{aaronson2009need,
  title={The need for structure in quantum speedups},
  author={Aaronson, Scott and Ambainis, Andris},
  journal={arXiv preprint arXiv:0911.0996},
  year={2009}
}

@article{ACCFLM22,
  title={On the Impossibility of Key Agreements from Quantum Random Oracles},
  author={Austrin, Per and Chung, Hao and Chung, Kai-Min and Fu, Shiuan and Lin, Yao-Ting and Mahmoody, Mohammad},
  journal={Cryptology ePrint Archive},
  year={2022}
}

@inproceedings{HY20,
  title={Finding collisions in a quantum world: quantum black-box separation of collision-resistance and one-wayness},
  author={Hosoyamada, Akinori and Yamakawa, Takashi},
  booktitle={International Conference on the Theory and Application of Cryptology and Information Security},
  pages={3--32},
  year={2020},
  organization={Springer}
}

@inproceedings{DLMM11,
  title={On the black-box complexity of optimally-fair coin tossing},
  author={Dachman-Soled, Dana and Lindell, Yehuda and Mahmoody, Mohammad and Malkin, Tal},
  booktitle={Theory of Cryptography Conference},
  pages={450--467},
  year={2011},
  organization={Springer}
}

@inproceedings{GKLM12,
  title={On black-box reductions between predicate encryption schemes},
  author={Goyal, Vipul and Kumar, Virendra and Lokam, Satya and Mahmoody, Mohammad},
  booktitle={Theory of Cryptography Conference},
  pages={440--457},
  year={2012},
  organization={Springer}
}

@inproceedings{BDV17,
  title={Structure vs. hardness through the obfuscation lens},
  author={Bitansky, Nir and Degwekar, Akshay and Vaikuntanathan, Vinod},
  booktitle={Annual International Cryptology Conference},
  pages={696--723},
  year={2017},
  organization={Springer}
}

@inproceedings{GKMRV00,
  title={The relationship between public key encryption and oblivious transfer},
  author={Gertner, Yael and Kannan, Sampath and Malkin, Tal and Reingold, Omer and Viswanathan, Mahesh},
  booktitle={Proceedings 41st Annual Symposium on Foundations of Computer Science},
  pages={325--335},
  year={2000},
  organization={IEEE}
}

@inproceedings{Rud91,
  title={The Use of Interaction in Public Cryptosystems.},
  author={Rudich, Steven},
  booktitle={Annual International Cryptology Conference},
  pages={242--251},
  year={1991},
  organization={Springer}
}

@article{PRV12,
  title={How powerful are the DDH hard groups?},
  author={Papakonstantinou, Periklis A and Rackoff, Charles W and Vahlis, Yevgeniy},
  journal={Cryptology ePrint Archive},
  year={2012}
}

@inproceedings{BPRVW08,
  title={On the impossibility of basing identity based encryption on trapdoor permutations},
  author={Boneh, Dan and Papakonstantinou, Periklis and Rackoff, Charles and Vahlis, Yevgeniy and Waters, Brent},
  booktitle={2008 49th Annual IEEE Symposium on Foundations of Computer Science},
  pages={283--292},
  year={2008},
  organization={IEEE}
}

@book{papadimitriou1994computational,
  author = {Papadimitriou, Christos H.},
  publisher = {Addison-Wesley},
  title = {Computational {C}omplexity},
  year = 1994
}

@inproceedings{RemoveROfromObf,
  title={On obfuscation with random oracles},
  author={Canetti, Ran and Kalai, Yael Tauman and Paneth, Omer},
  booktitle={Theory of Cryptography Conference},
  pages={456--467},
  year={2015},
  organization={Springer}
}

@article{marriott,
  title={Quantum arthur--merlin games},
  author={Marriott, Chris and Watrous, John},
  journal={computational complexity},
  volume={14},
  number={2},
  pages={122--152},
  year={2005},
  publisher={Springer}
}

@article{WZ82,
  title={A single quantum cannot be cloned},
  author={Wootters, William K and Zurek, Wojciech H},
  journal={Nature},
  volume={299},
  number={5886},
  pages={802--803},
  year={1982},
  publisher={Nature Publishing Group}
}

@inproceedings{AJLMS21,
  title={Indistinguishability obfuscation without multilinear maps: new paradigms via low degree weak pseudorandomness and security amplification},
  author={Ananth, Prabhanjan and Jain, Aayush and Lin, Huijia and Matt, Christian and Sahai, Amit},
  booktitle={Annual International Cryptology Conference},
  pages={284--332},
  year={2019},
  organization={Springer}
}

@inproceedings{JLS21,
  title={Indistinguishability obfuscation from well-founded assumptions},
  author={Jain, Aayush and Lin, Huijia and Sahai, Amit},
  booktitle={Proceedings of the 53rd Annual ACM SIGACT Symposium on Theory of Computing},
  pages={60--73},
  year={2021}
}

@inproceedings{GP21,
  title={Indistinguishability obfuscation from circular security},
  author={Gay, Romain and Pass, Rafael},
  booktitle={Proceedings of the 53rd Annual ACM SIGACT Symposium on Theory of Computing},
  pages={736--749},
  year={2021}
}

@inproceedings{BGIRSVY01,
  title={On the (im) possibility of obfuscating programs},
  author={Barak, Boaz and Goldreich, Oded and Impagliazzo, Rusell and Rudich, Steven and Sahai, Amit and Vadhan, Salil and Yang, Ke},
  booktitle={Annual international cryptology conference},
  pages={1--18},
  year={2001},
  organization={Springer}
}

@article{RS22,
  title={Semi-quantum money},
  author={Radian, Roy and Sattath, Or},
  journal={Journal of Cryptology},
  volume={35},
  number={2},
  pages={1--70},
  year={2022},
  publisher={Springer}
}

@inproceedings{JLS18,
  title={Pseudorandom quantum states},
  author={Ji, Zhengfeng and Liu, Yi-Kai and Song, Fang},
  booktitle={Annual International Cryptology Conference},
  pages={126--152},
  year={2018},
  organization={Springer}
}

@inproceedings{MVW12,
  title={Optimal counterfeiting attacks and generalizations for Wiesner’s quantum money},
  author={Molina, Abel and Vidick, Thomas and Watrous, John},
  booktitle={Conference on Quantum Computation, Communication, and Cryptography},
  pages={45--64},
  year={2012},
  organization={Springer}
}

@article{Lut10,
  title={An online attack against Wiesner's quantum money},
  author={Lutomirski, Andrew},
  journal={arXiv preprint arXiv:1010.0256},
  year={2010}
}

@inproceedings{HJL21,
  title={Counterexamples to new circular security assumptions underlying iO},
  author={Hopkins, Sam and Jain, Aayush and Lin, Huijia},
  booktitle={Annual International Cryptology Conference},
  pages={673--700},
  year={2021},
  organization={Springer}
}

@article{BDGM21,
  title={Factoring and pairings are not necessary for iO: circular-secure LWE suffices},
  author={Brakerski, Zvika and D{\"o}ttling, Nico and Garg, Sanjam and Malavolta, Giulio},
  journal={Cryptology ePrint Archive},
  year={2020}
}

@inproceedings{WW21,
  title={Candidate obfuscation via oblivious LWE sampling},
  author={Wee, Hoeteck and Wichs, Daniel},
  booktitle={Annual International Conference on the Theory and Applications of Cryptographic Techniques},
  pages={127--156},
  year={2021},
  organization={Springer}
}

@inproceedings{JLS22b,
  title={Indistinguishability Obfuscation from LPN over, DLIN, and PRGs in NC},
  author={Jain, Aayush and Lin, Huijia and Sahai, Amit},
  booktitle={Annual International Conference on the Theory and Applications of Cryptographic Techniques},
  pages={670--699},
  year={2022},
  organization={Springer}
}

@inproceedings{Shm22,
  title={Public-key Quantum money with a classical bank},
  author={Shmueli, Omri},
  booktitle={Proceedings of the 54th Annual ACM SIGACT Symposium on Theory of Computing},
  pages={790--803},
  year={2022}
}

@article{Shm22b,
  title={Semi-Quantum Tokenized Signatures},
  author={Shmueli, Omri},
  journal={Cryptology ePrint Archive},
  year={2022}
}

@inproceedings{Rob21,
  title={Security analysis of quantum lightning},
  author={Roberts, Bhaskar},
  booktitle={Annual International Conference on the Theory and Applications of Cryptographic Techniques},
  pages={562--567},
  year={2021},
  organization={Springer}
}

@article{BS16,
  title={Quantum tokens for digital signatures},
  author={Ben-David, Shalev and Sattath, Or},
  journal={arXiv preprint arXiv:1609.09047},
  year={2016}
}

@inproceedings{BGS13,
  title={Quantum one-time programs},
  author={Broadbent, Anne and Gutoski, Gus and Stebila, Douglas},
  booktitle={Annual Cryptology Conference},
  pages={344--360},
  year={2013},
  organization={Springer}
}

@article{ALP21,
  title={Secure Software Leasing},
  author={Ananth, Prabhanjan and La Placa, Rolando L},
  journal={Eurocrypt},
  year={2021}
}

@article{GZ20,
  title={Unclonable decryption keys},
  author={Georgiou, Marios and Zhandry, Mark},
  journal={Cryptology ePrint Archive},
  year={2020}
}

@inproceedings{BI20,
  title={Quantum encryption with certified deletion},
  author={Broadbent, Anne and Islam, Rabib},
  booktitle={Theory of Cryptography Conference},
  pages={92--122},
  year={2020},
  organization={Springer}
}

@InProceedings{BL20,
  author =	{Anne Broadbent and S{\'e}bastien Lord},
  title =	{{Uncloneable Quantum Encryption via Oracles}},
  booktitle =	{15th Conference on the Theory of Quantum Computation, Communication and Cryptography (TQC 2020)},
  pages =	{4:1--4:22},
  series =	{Leibniz International Proceedings in Informatics (LIPIcs)},
  ISSN =	{1868-8969},
  year =	{2020},
  volume =	{158},
  editor =	{Steven T. Flammia},
  publisher =	{Schloss Dagstuhl--Leibniz-Zentrum f{\"u}r Informatik},
  address =	{Dagstuhl, Germany},
  URN =		{urn:nbn:de:0030-drops-120639},
  doi =		{10.4230/LIPIcs.TQC.2020.4}
}

@article{Dieks82,
  title={Communication by EPR devices},
  author={Dieks, DGBJ},
  journal={Physics Letters A},
  volume={92},
  number={6},
  pages={271--272},
  year={1982},
  publisher={Elsevier}
}

@inproceedings{coladangelo2021hidden,
  title={Hidden cosets and applications to unclonable cryptography},
  author={Coladangelo, Andrea and Liu, Jiahui and Liu, Qipeng and Zhandry, Mark},
  booktitle={Annual International Cryptology Conference},
  pages={556--584},
  year={2021},
  organization={Springer}
}

@article{IPforStates,
  title={Interactive Proofs for Synthesizing Quantum States and Unitaries},
  author={Rosenthal, Gregory and Yuen, Henry},
  journal={arXiv preprint arXiv:2108.07192},
  year={2021}
}

@inproceedings{FGHLS12,
  title={Quantum money from knots},
  author={Farhi, Edward and Gosset, David and Hassidim, Avinatan and Lutomirski, Andrew and Shor, Peter},
  booktitle={Proceedings of the 3rd Innovations in Theoretical Computer Science Conference},
  pages={276--289},
  year={2012}
}

@article{KSS21,
  title={Quantum money from quaternion algebras},
  author={Kane, Daniel M and Sharif, Shahed and Silverberg, Alice},
  journal={arXiv preprint arXiv:2109.12643},
  year={2021}
}

@article{KJS22,
  title={Publicly verifiable quantum money from random lattices},
  author={Khesin, Andrey Boris and Lu, Jonathan Z and Shor, Peter W},
  journal={arXiv preprint arXiv:2207.13135},
  year={2022}
}

@article{Win99,
  title={Coding theorem and strong converse for quantum channels},
  author={Winter, Andreas},
  journal={IEEE Transactions on Information Theory},
  volume={45},
  number={7},
  pages={2481--2485},
  year={1999},
  publisher={IEEE}
}

@incollection {Rewinding,
    AUTHOR = {Chiesa, Alessandro and Ma, Fermi and Spooner, Nicholas and
              Zhandry, Mark},
     TITLE = {Post-quantum succinct arguments: breaking the quantum
              rewinding barrier},
 BOOKTITLE = {2021 {IEEE} 62nd {A}nnual {S}ymposium on {F}oundations of
              {C}omputer {S}cience---{FOCS} 2021},
     PAGES = {49--58},
 PUBLISHER = {IEEE Computer Soc., Los Alamitos, CA},
      YEAR = {[2022] \copyright 2022},
   MRCLASS = {68Q12},
  MRNUMBER = {4399668},
       DOI = {10.1109/FOCS52979.2021.00014},
       URL = {https://doi.org/10.1109/FOCS52979.2021.00014},
}

@book {nielsen2000quantum,
    AUTHOR = {Nielsen, Michael A. and Chuang, Isaac L.},
     TITLE = {Quantum computation and quantum information},
 PUBLISHER = {Cambridge University Press, Cambridge},
      YEAR = {2000},
     PAGES = {xxvi+676},
      ISBN = {0-521-63235-8; 0-521-63503-9},
   MRCLASS = {81P68 (68-01 68-02 68Q05 81-02)},
  MRNUMBER = {1796805},
}

@incollection {IR,
    AUTHOR = {Impagliazzo, Russell and Rudich, Steven},
     TITLE = {Limits on the provable consequences of one-way permutations},
 BOOKTITLE = {Advances in cryptology---{CRYPTO} '88 ({S}anta {B}arbara,
              {CA}, 1988)},
    SERIES = {Lecture Notes in Comput. Sci.},
    VOLUME = {403},
     PAGES = {8--26},
 PUBLISHER = {Springer, Berlin},
      YEAR = {1990},
   MRCLASS = {94A60 (03D80 68P25 68Q05)},
  MRNUMBER = {1046378},
MRREVIEWER = {Joan Boyar},
       DOI = {10.1007/0-387-34799-2\_2},
       URL = {https://doi.org/10.1007/0-387-34799-2_2},
}

@article{wiesner1983conjugate,
  title={Conjugate coding},
  author={Wiesner, Stephen},
  journal={ACM Sigact News},
  volume={15},
  number={1},
  pages={78--88},
  year={1983},
  publisher={ACM New York, NY, USA}
}

@incollection {aaronson2009pkQM,
    AUTHOR = {Aaronson, Scott},
     TITLE = {Quantum copy-protection and quantum money},
 BOOKTITLE = {24th {A}nnual {IEEE} {C}onference on {C}omputational
              {C}omplexity},
     PAGES = {229--242},
 PUBLISHER = {IEEE Computer Soc., Los Alamitos, CA},
      YEAR = {2009},
   MRCLASS = {68Q12},
  MRNUMBER = {2932470},
       DOI = {10.1109/CCC.2009.42},
       URL = {https://doi.org/10.1109/CCC.2009.42},
}

@article {zhandry2021quantumlightning,
    AUTHOR = {Zhandry, Mark},
     TITLE = {Quantum lightning never strikes the same state twice. {O}r:
              quantum money from cryptographic assumptions},
   JOURNAL = {J. Cryptology},
  FJOURNAL = {Journal of Cryptology. The Journal of the International
              Association for Cryptologic Research},
    VOLUME = {34},
      YEAR = {2021},
    NUMBER = {1},
     PAGES = {Paper No. 6, 56},
      ISSN = {0933-2790},
   MRCLASS = {94A60},
  MRNUMBER = {4196012},
       DOI = {10.1007/s00145-020-09372-x},
       URL = {https://doi.org/10.1007/s00145-020-09372-x},
}

@article{bilyk2022cryptanalysis,
  title={Cryptanalysis of Three Quantum Money Schemes},
  author={Bilyk, Andriyan and Doliskani, Javad and Gong, Zhiyong},
  journal={arXiv preprint arXiv:2205.10488},
  year={2022}
}

@article {AC2013HiddenSubspace,
    AUTHOR = {Aaronson, Scott and Christiano, Paul},
     TITLE = {Quantum money from hidden subspaces},
   JOURNAL = {Theory Comput.},
  FJOURNAL = {Theory of Computing. An Open Access Journal},
    VOLUME = {9},
      YEAR = {2013},
     PAGES = {349--401},
   MRCLASS = {81P94 (94A60)},
  MRNUMBER = {3040602},
MRREVIEWER = {Aharon Brodutch},
       DOI = {10.4086/toc.2013.v009a009},
       URL = {https://doi.org/10.4086/toc.2013.v009a009},
}

@misc{zhandryCompressedOracle,
      author = {Mark Zhandry},
      title = {How to Record Quantum Queries, and Applications to Quantum Indifferentiability},
      howpublished = {Cryptology ePrint Archive, Paper 2018/276},
      year = {2018},
      note = {\url{https://eprint.iacr.org/2018/276}},
      url = {https://eprint.iacr.org/2018/276}
}

\appendix
\section{Missing Proof of \Cref{lemma:diffofquery}}\label{appendix:proof1}

The proof consists of two parts. In the first part, we will bound $\TraceDist{\compV{U_C}U_j\ketbra{\phi_{j}}{\phi_{j}}U_j^{\dagger}\compV{U_C}^{\dagger}}{U_D'U_j\ketbra{\phi_{j}}{\phi_{j}}U_j^{\dagger}U_D'^{\dagger}}$ by the weight of queries inside $D_F$. In the second part, we will show that the weight equals the decrement of the number of pairs in $\reg{F}$ after the query.

\begin{proof}

Let's begin the first part with some notations. 

For $0 \le j \le q - 1$, $U_j$ prepares the $(j + 1)^{th}$ query, thus we can write $U_j\ket{\phi_j}$ as follows,
$$U_j\ket{\phi_j} = \sum_{\substack{x, D_F, D_{\ro}, h\\ \text{s.t. }D_F \cap D_{\ro} = \emptyset}}\alpha_{x, D_F, D_{\ro}, h}\ket{x}_{\reg{Q}}\ket{0}_{\reg{A}}\ket{D_F}_{\reg{F}}\ket{D_{\ro}}_{\reg{D}_{\ro}}\ket{h}_{\reg{H}}$$ 
where $\reg{Q}$ is the next query position, $\reg{A}$ is for the query answer and $\reg{H}$ contains all other registers including the public key, serial number, $\reg{D}_{\adversary}$, etc. We can divide these terms into two categories, one is $x \notin D_F$, and the other one is $x \in D_F$. That is, $U_j\ket{\phi_j} = \sqrt{1 - \alpha}\ket{u} + \sqrt{\alpha}\ket{v}$ where $\alpha:= \sum_{\substack{x, D_F, D_{\ro}, h\\ \text{s.t. }D_F \cap D_{\ro} = \emptyset, x \in D_F}}\left|\alpha_{x, D_F, D_{\ro}, h}\right|^2$ is the weight of queries inside $D_F$,
$$\sqrt{1 - \alpha} \ket{u} = \sum_{\substack{x, D_F, D_{\ro}, h\\ \text{s.t. }D_F \cap D_{\ro} = \emptyset, x \notin D_F}}\alpha_{x, D_F, D_{\ro}, h}\ket{x}_{\reg{Q}}\ket{0}_{\reg{A}}\ket{D_F}_{\reg{F}}\ket{D_{\ro}}_{\reg{D}_{\ro}}\ket{h}_{\reg{H}}$$
$$\sqrt{\alpha} \ket{v} = \sum_{\substack{x, D_F, D_{\ro}, h\\ \text{s.t. }D_F \cap D_{\ro} = \emptyset, x \in D_F}}\alpha_{x, D_F, D_{\ro}, h}\ket{x}_{\reg{Q}}\ket{0}_{\reg{A}}\ket{D_F}_{\reg{F}}\ket{D_{\ro}}_{\reg{D}_{\ro}}\ket{h}_{\reg{H}}$$

Recall that 
\begin{align*}
    &\compV{U_C}(\ket{x}_{\reg{Q}}\ket{0}_{\reg{A}}\ket{D_F}_{\reg{F}}\ket{D_{\ro}}_{\reg{D}_{\ro}})\\
    =&\begin{cases}
        \ket{x}_{\reg{Q}}\ket{D_{\ro}(x)}_{\reg{A}}\ket{D_F}_{\reg{F}}\ket{D_{\ro}, (x, D_{\ro}(x))}_{\reg{D}_{\ro}} &x \in D_{\ro}\\
        \ket{x}_{\reg{Q}}\frac{1}{\sqrt 2}\sum_{z = 0}^1\ket{z}_{\reg{A}}\ket{D_F}_{\reg{F}}\ket{D_{\ro}, (x, z)}_{\reg{D}_{\ro}} &x \notin D_{\ro}, x \notin D_F\\
        \ket{x}_{\reg{Q}}\frac{1}{\sqrt 2}\sum_{z = 0}^1 (-1)^{z\widehat{D_F(x)}} \ket{z}_{\reg{A}}\ket{D_F - x}_{\reg{F}}\ket{D_{\ro}, (x, z)}_{\reg{D}_{\ro}} &x \notin D_{\ro}, x \in D_F
    \end{cases}
\end{align*}

So when $x \notin D_F$, $\compV{U_C}$ and $U_D'$ act exactly the same. As a result, $\compV{U_C}\ket{u} = U_D'\ket{u}$. Therefore,
\begin{align*}
&\TraceDist{\compV{U_C}U_j\ketbra{\phi_{j}}{\phi_{j}}U_j^{\dagger}\compV{U_C}^{\dagger}}{U_D'U_j\ketbra{\phi_{j}}{\phi_{j}}U_j^{\dagger}U_D'^{\dagger}}\\
= &\norm{\compV{U_C}U_j\ketbra{\phi_{j}}{\phi_{j}}U_j^{\dagger}\compV{U_C}^{\dagger} - U_D'U_j\ketbra{\phi_{j}}{\phi_{j}}U_j^{\dagger}U_D'^{\dagger}}_1\\
= &\|\compV{U_C}\left((1 - \alpha)\ketbra{u}{u} + \alpha\ketbra{v}{v} + \sqrt{\alpha(1 - \alpha)}\left(\ketbra{u}{v} + \ketbra{v}{u}\right)\right)\compV{U_C}^{\dagger} \\&- U_D'\left((1 - \alpha)\ketbra{u}{u} + \alpha\ketbra{v}{v} + \sqrt{\alpha(1 - \alpha)}\left(\ketbra{u}{v} + \ketbra{v}{u}\right)\right)U_D'^{\dagger}\|_1\\
= &\|\compV{U_C}\left(\alpha\ketbra{v}{v} + \sqrt{\alpha(1 - \alpha)}\left(\ketbra{u}{v} + \ketbra{v}{u}\right)\right)\compV{U_C}^{\dagger} - U_D'\left(\alpha\ketbra{v}{v} + \sqrt{\alpha(1 - \alpha)}\left(\ketbra{u}{v} + \ketbra{v}{u}\right)\right)U_D'^{\dagger}\|_1\\
\le & 2 \norm{\alpha\ketbra{v}{v} + \sqrt{\alpha(1 - \alpha)}\left(\ketbra{u}{v} + \ketbra{v}{u}\right)}_1\\
\le & 2(\alpha + 2\sqrt{\alpha(1 - \alpha)})\\
\le & 6\sqrt{\alpha}
\end{align*}
where we use properties of trace norm of matrices $\norm{A + B}_1 \le \norm{A}_1 + \norm{B}_1, \norm{UAU^{\dagger}}_1 = \norm{A}_1$ and $\norm{A}_1 = \sum_{i}\lvert \lambda_i \rvert$ if $A = \sum_i \lambda_i \ketbra{s_i}{v_i}$ where $\ket{s_i}, \ket{v_i}$ are two bases.

Now let's move to the second part. We will show that $\alpha$ equals to the difference of the number of pairs in $\reg{F}$ on $\ket{\phi_j}$ and $\ket{\phi_{j + 1}} = \compV{U_C}U_j\ket{\phi_{j}}$.

First, let's calculate the number of pairs in $\reg{F}$ on state $\ket{\phi_j}$. 

Notice that $U_j$ does not act on $\reg{F}$, so it commutes with $O$. Then
$\Tr(O\ketbra{\phi_j}{\phi_j})= \Tr(OU_j\ketbra{\phi_j}{\phi_j}U_j^{\dagger})$. 
Recall that $$U_j\ket{\phi_j} = \sum_{\substack{x, D_F, D_{\ro}, h\\ \text{s.t. }D_F \cap D_{\ro} = \emptyset}}\alpha_{x, D_F, D_{\ro}, h}\ket{x}_{\reg{Q}}\ket{0}_{\reg{A}}\ket{D_F}_{\reg{F}}\ket{D_{\ro}}_{\reg{D}_{\ro}}\ket{h}_{\reg{H}}$$
In the summation, the terms are orthogonal to each other. Thus the expected number of pairs in $\reg{F}$ of $U_j\ket{\phi_j}$ is $\sum_{\substack{x, D_F, D_{\ro}, h\\ \text{s.t. }D_F \cap D_{\ro} = \emptyset}}\length{D_F}\left|\alpha_{x, D_F, D_{\ro}, h}\right|^2$.
So $$\Tr(O\ketbra{\phi_j}{\phi_j}) = \sum_{\substack{x, D_F, D_{\ro}, h\\ \text{s.t. }D_F \cap D_{\ro} = \emptyset}}\length{D_F}\left|\alpha_{x, D_F, D_{\ro}, h}\right|^2$$

Next, let's calculate the number of pairs in $\reg{F}$ on state $\ket{\phi_{j + 1}}$.

By definition of $\ket{\phi_{j+1}}$ and $\compV{U_C}$
\begin{align*}
&\ket{\phi_{j + 1}}\\
= & \compV{U_C}U_j\ket{\phi_j}\\
= & \sum_{\substack{x, D_F, D_{\ro}, h\\ \text{s.t. }D_F \cap D_{\ro} = \emptyset, x \in D_{\ro}}}\alpha_{x, D_F, D_{\ro}, h}\ket{x}_{\reg{Q}}\ket{D_{\ro}(x)}_{\reg{A}}\ket{D_F}_{\reg{F}}\ket{D_{\ro}, (x, D_{\ro}(x))}_{\reg{D}_{\ro}}\ket{h}_{\reg{H}}\\
&+ \sum_{\substack{x, D_F, D_{\ro}, h\\ \text{s.t. }D_F \cap D_{\ro} = \emptyset, x \notin D_{\ro}, x \notin D_F}}\alpha_{x, D_F, D_{\ro}, h}\frac{1}{\sqrt 2}\sum_{z = 0}^1\ket{x}_{\reg{Q}}\ket{z}_{\reg{A}}\ket{D_F}_{\reg{F}}\ket{D_{\ro}, (x, z)}_{\reg{D}_{\ro}}\ket{h}_{\reg{H}}\\
&+ \sum_{\substack{x, D_F, D_{\ro}, h\\ \text{s.t. }D_F \cap D_{\ro} = \emptyset, x \in D_F}}\alpha_{x, D_F, D_{\ro}, h}\frac{1}{\sqrt 2}\sum_{z = 0}^1(-1)^{z\widehat{D_F(x)}}\ket{x}_{\reg{Q}}\ket{z}_{\reg{A}}\ket{D_F - x}_{\reg{F}}\ket{D_{\ro}, (x, z)}_{\reg{D}_{\ro}}\ket{h}_{\reg{H}}
\end{align*}

$\compV{U_C}$ is a unitary, so the terms $\ket{x}_{\reg{Q}}\ket{D_{\ro}(x)}_{\reg{A}}\ket{D_F}_{\reg{F}}\ket{D_{\ro}, (x, D_{\ro}(x))}_{\reg{D}_{\ro}}\ket{h}_{\reg{H}}$ where $D_F \cap D_{\ro} = \emptyset$ and $x \in D_{\ro}$, $\frac{1}{\sqrt 2}\sum_{z = 0}^1\ket{x}_{\reg{Q}}\ket{z}_{\reg{A}}\ket{D_F}_{\reg{F}}\ket{D_{\ro}, (x, z)}_{\reg{D}_{\ro}}\ket{h}_{\reg{H}}$ where $D_F \cap D_{\ro} = \emptyset$ and $x \notin D_{\ro} \cup D_F$, and $\frac{1}{\sqrt 2}\sum_{z = 0}^1(-1)^{z\widehat{D_F(x)}}\ket{x}_{\reg{Q}}\ket{z}_{\reg{A}}\ket{D_F - x}_{\reg{F}}\ket{D_{\ro}, (x, z)}_{\reg{D}_{\ro}}\ket{h}_{\reg{H}}$ where $D_F \cap D_{\ro} = \emptyset$ and $x \in D_F$ in the above summation are orthogonal to each other. 

Thus the expected number of pairs in $\reg{F}$ of $\ket{\phi_{j + 1}}$ is
\begin{align*}
&\Tr(O\ketbra{\phi_{j + 1}}{\phi_{j + 1}})\\
= &\sum_{\substack{x, D_F, D_{\ro}, h\\ \text{s.t. }D_F \cap D_{\ro} = \emptyset, x \notin D_F}}\length{D_F}\left|\alpha_{x, D_F, D_{\ro}, h}\right|^2 + \sum_{\substack{x, D_F, D_{\ro}, h\\ \text{s.t. }D_F \cap D_{\ro} = \emptyset, x \in D_F}}\length{D_F - x}\left|\alpha_{x, D_F, D_{\ro}, h}\right|^2\\
= & \sum_{\substack{x, D_F, D_{\ro}, h\\ \text{s.t. }D_F \cap D_{\ro} = \emptyset}}\length{D_F}\left|\alpha_{x, D_F, D_{\ro}, h}\right|^2 - \sum_{\substack{x, D_F, D_{\ro}, h\\ \text{s.t. }D_F \cap D_{\ro} = \emptyset, x \in D_F}}\left|\alpha_{x, D_F, D_{\ro}, h}\right|^2\\
= & \Tr(O\ketbra{\phi_j}{\phi_j}) - \alpha
\end{align*}

So $\alpha = \Tr(O\ketbra{\phi_j}{\phi_j}) - \Tr(O\ketbra{\phi_{j + 1}}{\phi_{j + 1}})$. That is to say, the weight of queries outside $D_F$ equals the decrement of the number of pairs in $\reg{F}$ after the query.

Combine the above two parts, and we can obtain
$$\TraceDist{\compV{U_C}U_j\ketbra{\phi_{j}}{\phi_{j}}U_j^{\dagger}\compV{U_C}^{\dagger}}{U_D'U_j\ketbra{\phi_{j}}{\phi_{j}}U_j^{\dagger}U_D'^{\dagger}} \le 6\sqrt{\Tr(O\ketbra{\phi_j}{\phi_j}) - \Tr(O\ketbra{\phi_{j + 1}}{\phi_{j + 1}})}.$$
\end{proof}

\section{Missing Proof of \Cref{RemoveUpd}}\label{appendix:proof2}
We prove~\Cref{RemoveUpd} by first showing that for two parallel queries, we can remove one of them without decreasing the value and then extending the result to $\compV{U_\verify}$ on $\rho_s^{(t)}$ and $\compV{U_{\update}}$ (each has polynomial queries). 
\begin{proof}
Let $\compV{U_R}$ act on the first query position register $\reg{Q}_1$, the first query answer register $\reg{A}_1, \reg{D}_{\adversary}, \reg{F}$ and $\reg{D}_{\ro}$ while $\compV{U_C}$ acts on the second query position register $\reg{Q}_2$, the second query position register $\reg{A}_2, \reg{F}$ and $\reg{D}_{\ro}$.

We first show that for any state $$\ket{\varphi} =  \sum_{\substack{x_1, x_2, D_{\adversary}, D_F, D_{\ro}, h\\ \text{s.t. }D_F \cap D_{\ro} = \emptyset}}\alpha_{x_1, x_2, D_{\adversary}, D_F, D_{\ro}, h}\ket{x_1}_{\reg{Q}_1}\ket{0}_{\reg{A}_1}\ket{x_2}_{\reg{Q}_2}\ket{0}_{\reg{A}_2}\ket{D_{\adversary}}_{\reg{D}_{\adversary}}\ket{D_F}_{\reg{F}}\ket{D_{\ro}}_{\reg{D}_{\ro}}\ket{h}_{\reg{H}},$$ we have the inequality
$$\Tr(O\compV{U_R} \ketbra{\varphi}{\varphi} \compV{U_R}^{\dagger}) - \Tr(O \compV{U_C}\compV{U_R} \ketbra{\varphi}{\varphi} \compV{U_R}^{\dagger}\compV{U_C}^{\dagger}) \le  \Tr(O\ketbra{\varphi}{\varphi}) - \Tr(O \compV{U_C} \ketbra{\varphi}{\varphi} \compV{U_C}^{\dagger}).$$

In fact, from the same argument in \Cref{lemma:diffofquery}, $\Tr(O\ketbra{\varphi}{\varphi}) - \Tr(O \compV{U_C} \ketbra{\varphi}{\varphi} \compV{U_C}^{\dagger})$ is exactly the probability that we get outcome $(x_2, D_F)$ such that $x_2 \in D_F$ when we measure the registers $\reg{Q}_2$ and $\reg{D}_F$ on state $\ket{\varphi}$. That is 
$$\Tr(O\ketbra{\varphi}{\varphi}) - \Tr(O \compV{U_C} \ketbra{\varphi}{\varphi} \compV{U_C}^{\dagger}) = \sum_{\substack{x_1, x_2, D_{\adversary}, D_F, D_{\ro}, h\\ \text{s.t. }D_F \cap D_{\ro} = \emptyset, x_2 \in D_F}}\left|\alpha_{x_1, x_2, D_{\adversary}, D_F, D_{\ro}, h}\right|^2.$$
    
Similarly, $\Tr(O\compV{U_R} \ketbra{\varphi}{\varphi} \compV{U_R}^{\dagger}) - \Tr(O \compV{U_C}\compV{U_R} \ketbra{\varphi}{\varphi} \compV{U_R}^{\dagger}\compV{U_C}^{\dagger})$ is exactly the probability that we get outcome $(x_2, D_F)$ such that $x_2 \in D_F$ when we measure the registers $\reg{Q}_2$ and $\reg{D}_F$ on state $\compV{U_R}\ket{\varphi}$. Notice that we can write $\compV{U_R}\ket{\varphi}$ as 
$$\sum_{\substack{x_1, x_2, D_{\adversary}, D_F, D_{\ro}, h\\ \text{s.t. }D_F \cap D_{\ro} = \emptyset}}\alpha_{x_1, x_2, D_{\adversary}, D_F, D_{\ro}, h}\ket{x_2}_{\reg{Q}_2}\ket{0}_{\reg{A}_2}\ket{x_1}_{\reg{Q}_1}U_{x_1, D_F}(\ket{0}_{\reg{A}_1}\ket{D_{\adversary}}_{\reg{D}_{\adversary}}\ket{D_{\ro}}_{\reg{D}_{\ro}})\ket{D_F - x_1}_{\reg{F}}\ket{h}_{\reg{H}}$$ where $U_{x_1, D_F}$ is a unitary that only depends on $x_1, D_F$. $\compV{U_R}$ is a unitary, so the terms in the above summation are orthogonal. Thus
$$\Tr(O\compV{U_R} \ketbra{\varphi}{\varphi} \compV{U_R}^{\dagger}) - \Tr(O \compV{U_C}\compV{U_R} \ketbra{\varphi}{\varphi} \compV{U_R}^{\dagger}\compV{U_C}^{\dagger}) = \sum_{\substack{x_1, x_2, D_{\adversary}, D_F, D_{\ro}, h\\ \text{s.t. }D_F \cap D_{\ro} = \emptyset, x_2 \in D_F - x_1}}\left|\alpha_{x_1, x_2, D_{\adversary}, D_F, D_{\ro}, h}\right|^2.$$
As a result,
$\Tr(O\compV{U_R} \ketbra{\varphi}{\varphi} \compV{U_R}^{\dagger}) - \Tr(O \compV{U_C}\compV{U_R} \ketbra{\varphi}{\varphi} \compV{U_R}^{\dagger}\compV{U_C}^{\dagger}) \le  \Tr(O\ketbra{\varphi}{\varphi}) - \Tr(O \compV{U_C} \ketbra{\varphi}{\varphi} \compV{U_C}^{\dagger}).$ That is to say, an extra query $\compV{U_R}$ on another part of the state can only decrease the chance of making a bad query in $\compV{U_C}$ because that extra query can only make the set of bad queries smaller.
    
$\compV{U_{\verify}}$ and $\compV{U_{\update}}$ are composed of $\compV{U_C}$ and $\compV{U_R}$. In fact, the above argument can also be extended to $\compV{U_{\verify}}$ and $\compV{U_{\update}}$ to capture our intuition that $\compV{U_{\update}}$ can only decrease the number of bad queries made during $\compV{U_{\verify}}$ because $\compV{U_{\update}}$ can only make the set of bad queries smaller. 

We will first show a fixed number of iterations of the update phase can only decrease the number of bad queries made during $\compV{U_{\verify}}$ and then show it holds for $\compV{U_{\update}}$.

Let $\ket{\psi}$ be the state $\ket{\psi}$ as in~\Cref{claim:QuantumSimOnRho}. We can write it as $\ket{\psi_0} \otimes \frac{1}{\sqrt{N(n)}}\sum_{j = 0}^{N(n) - 1}\ket{j}_{\reg{J}}$.

For any $j \ge 1$, for unitary $\compV{U_{\verify}'}$ that acts on the synthesized state and records the query for $\adversary$, and unitary $\compV{U_{\verify}}$ that acts on $\rho_s^{(t)}$, denote $\ket{\mu} = \compV{U_{\syn}}(\compV{U_{\verify}'}\compV{U_{\syn}})^{j - 1}\ket{\psi_0}$. Let $U_i^{(1)}$ prepare the next query for $\rho_s^{(t)}$ and $U_i^{(2)}$ prepare the next query for the synthesized state. For every $i$, denote $U_{i, C}^{(1)} = U_i^{(1)}\compV{U_C} \cdots U_0^{(1)}$ and $U_{i, R}^{(2)} = U_i^{(2)}\compV{U_R}\cdots U_0^{(2)}$. Apply the above inequality, and we can get that
\begin{align*}
&\Tr(O(\compV{U_{\verify}'}\compV{U_{\syn}})^j\ketbra{\psi_0}{\psi_0} (\compV{U_{\syn}}^{\dagger}\compV{U_{\verify}'}^{\dagger})^j) - \Tr(O \compV{U_{\verify}}(\compV{U_{\verify}'}\compV{U_{\syn}})^j \ketbra{\psi_0}{\psi_0} (\compV{U_{\syn}}^{\dagger}\compV{U_{\verify}'}^{\dagger})^j\compV{U_{\verify}}^{\dagger})\\
= &\Tr(O U_{q, R}^{(2)} \ketbra{\mu}{\mu} {U_{q, R}^{(2)}}^{\dagger}) - \Tr(O U_{q, C}^{(1)} U_{q, R}^{(2)} \ketbra{\mu}{\mu} {U_{q, R}^{(2)}}^{\dagger} {U_{q, C}^{(1)}}^{\dagger})\\
= &\Tr(O U_q^{(2)}\compV{U_R} U_{q - 1, R}^{(2)} \ketbra{\mu}{\mu} {U_{q - 1, R}^{(2)}}^{\dagger}\compV{U_R}^{\dagger} {U_q^{(2)}}^{\dagger})\\ 
&- \Tr(O U_q^{(1)}U_q^{(2)}\compV{U_C} \compV{U_R} U_{q - 1, C}^{(1)} U_{q - 1, R}^{(2)} \ketbra{\mu}{\mu} {U_{q - 1, R}^{(2)}}^{\dagger} {U_{q - 1, C}^{(1)}}^{\dagger}\compV{U_R}^{\dagger}  \compV{U_C}^{\dagger} {U_q^{(2)}}^{\dagger}{U_q^{(1)}}^{\dagger})\\
= &\Tr(O\compV{U_R} U_{q - 1, R}^{(2)} \ketbra{\mu}{\mu} {U_{q - 1, R}^{(2)}}^{\dagger}\compV{U_R}^{\dagger}) - \Tr(O\compV{U_C} \compV{U_R} U_{q - 1, C}^{(1)} U_{q - 1, R}^{(2)} \ketbra{\mu}{\mu} {U_{q - 1, R}^{(2)}}^{\dagger} {U_{q - 1, C}^{(1)}}^{\dagger}\compV{U_R}^{\dagger} \compV{U_C}^{\dagger})\\
\le &\Tr(O U_{q - 1, R}^{(2)} \ketbra{\mu}{\mu} {U_{q - 1, R}^{(2)}}^{\dagger}) - \Tr(O\compV{U_C} U_{q - 1, C}^{(1)} U_{q - 1, R}^{(2)} \ketbra{\mu}{\mu} {U_{q - 1, R}^{(2)}}^{\dagger} {U_{q - 1, C}^{(1)}}^{\dagger} \compV{U_C}^{\dagger})\\
= &\Tr(O U_{q - 1, R}^{(2)} \ketbra{\mu}{\mu} {U_{q - 1, R}^{(2)}}^{\dagger}) - \Tr(OU_{q, C}^{(1)} U_{q - 1, R}^{(2)} \ketbra{\mu}{\mu} {U_{q - 1, R}^{(2)}}^{\dagger} {U_{q, C}^{(1)}}^{\dagger})
\end{align*}
That is, we can delete a query for the synthesized state without decreasing the number of bad queries made during $\compV{U_{\verify}}$ on $\rho_s^{(t)}$. Repetitively removing the queries for the synthesized state, we can get that
\begin{align*}
&\Tr(O(\compV{U_{\verify}'}\compV{U_{\syn}})^j\ketbra{\psi_0}{\psi_0} (\compV{U_{\syn}}^{\dagger}\compV{U_{\verify}'}^{\dagger})^j) - \Tr(O \compV{U_{\verify}}(\compV{U_{\verify}'}\compV{U_{\syn}})^j \ketbra{\psi_0}{\psi_0} (\compV{U_{\syn}}^{\dagger}\compV{U_{\verify}'}^{\dagger})^j\compV{U_{\verify}}^{\dagger})\\
\le & \Tr(O U_{0, R}^{(2)} \ketbra{\mu}{\mu} {U_{0, R}^{(2)}}^{\dagger}) - \Tr(OU_{q, C}^{(1)} U_{0, R}^{(2)} \ketbra{\mu}{\mu} {U_{0, R}^{(2)}}^{\dagger} {U_{q, C}^{(1)}}^{\dagger})\\
= & \Tr(O \ketbra{\mu}{\mu} ) - \Tr(O \compV{U_{\verify}} \ketbra{\mu}{\mu} {\compV{U_{\verify}}}^{\dagger})
\end{align*}
where we use the fact that $U_i^{(1)}$ commutes with $U_j^{(2)}$ and $\compV{U_C}$, $U_i^{(2)}$ commutes with $\compV{U_R}$, and both $U_i^{(1)}$ and $U_j^{(2)}$ commute with $O$.

We successfully removed one $\compV{U_{\verify}'}$ on the synthesized state. We can do it until we remove all of $\compV{U_{\verify}'}$ on the synthesized state. That is,
\begin{align*}
&\Tr(O(\compV{U_{\verify}'}\compV{U_{\syn}})^j\ketbra{\psi_0}{\psi_0} (\compV{U_{\syn}}^{\dagger}\compV{U_{\verify}'}^{\dagger})^j) - \Tr(O \compV{U_{\verify}}(\compV{U_{\verify}'}\compV{U_{\syn}})^j \ketbra{\psi_0}{\psi_0} (\compV{U_{\syn}}^{\dagger}\compV{U_{\verify}'}^{\dagger})^j\compV{U_{\verify}}^{\dagger})\\
\le &\Tr(O\compV{U_{\syn}}(\compV{U_{\verify}'}\compV{U_{\syn}})^{j - 1}\ketbra{\psi_0}{\psi_0} (\compV{U_{\syn}}^{\dagger}\compV{U_{\verify}'}^{\dagger})^{j - 1}{\compV{U_{\syn}}}^{\dagger})\\
& - \Tr(O \compV{U_{\verify}} \compV{U_{\syn}}(\compV{U_{\verify}'}\compV{U_{\syn}})^{j - 1}\ketbra{\psi_0}{\psi_0} (\compV{U_{\syn}}^{\dagger}\compV{U_{\verify}'}^{\dagger})^{j - 1}{\compV{U_{\syn}}}^{\dagger} \compV{U_{\verify}}^{\dagger})\\
=&\Tr(O(\compV{U_{\verify}'}\compV{U_{\syn}})^{j - 1}\ketbra{\psi_0}{\psi_0} (\compV{U_{\syn}}^{\dagger}\compV{U_{\verify}'}^{\dagger})^{j - 1}) - \Tr(O \compV{U_{\verify}}(\compV{U_{\verify}'}\compV{U_{\syn}})^{j - 1}\ketbra{\psi_0}{\psi_0} (\compV{U_{\syn}}^{\dagger}\compV{U_{\verify}'}^{\dagger})^{j - 1}\compV{U_{\verify}}^{\dagger})\\
\le & \cdots \\
\le & \Tr(O\ketbra{\psi_0}{\psi_0}) - \Tr(O \compV{U_{\verify}}\ketbra{\psi_0}{\psi_0}\compV{U_{\verify}}^{\dagger})
\end{align*}
where we use the fact that $\compV{U_{\syn}}$ uses the database in $D_{\adversary}$ and $\compV{U_{\verify}}$ does not record query for $\adversary$. Thus $\compV{U_{\syn}}$ and $\compV{U_{\verify}}$ commute. $\compV{U_{\syn}}$ does not act on $\reg{F}$ and $\reg{D}_{\ro}$, so it commutes with $O$.

Finally, we will show a randomized number of iterations of the update phase can not increase the number of bad queries made during $\compV{U_{\verify}}$. As both $\compV{U_{\verify}}$ and $O$ do not act on register $\reg{J}$, tracing out $\reg{J}$ after $\compV{U_{\update}}$ does not change the quantity. Recall that $\compV{U_{\update}}\ket{\psi} = \frac{1}{\sqrt{N(n)}}\sum_{j = 0}^{N(n) - 1}(\compV{U_{\verify}'}\compV{U_{\syn}})^j(\ket{\psi_0})\ket{j}_{\reg{J}}$. Thus
\begin{align*}
&\Tr(O\compV{U_{\update}} \ketbra{\psi}{\psi} \compV{U_{\update}}^{\dagger}) - \Tr(O \compV{U_{\verify}}\compV{U_{\update}} \ketbra{\psi}{\psi} \compV{U_{\update}}^{\dagger}\compV{U_{\verify}}^{\dagger})\\
= &\Tr(O\Tr_{\reg{J}}(\compV{U_{\update}} \ketbra{\psi}{\psi} \compV{U_{\update}}^{\dagger})) - \Tr(O \compV{U_{\verify}}\Tr_{\reg{J}}(\compV{U_{\update}} \ketbra{\psi}{\psi}\compV{U_{\update}}^{\dagger})\compV{U_{\verify}}^{\dagger})\\
= &\frac{1}{N(n)} \sum_{j = 0}^{N(n) - 1}\left(\Tr(O(\compV{U_{\verify}'}\compV{U_{\syn}})^j\ketbra{\psi_0}{\psi_0} (\compV{U_{\syn}}^{\dagger}\compV{U_{\verify}'}^{\dagger})^j) - \Tr(O\compV{U_{\verify}}(\compV{U_{\verify}'}\compV{U_{\syn}})^j\ketbra{\psi_0}{\psi_0} (\compV{U_{\syn}}^{\dagger}\compV{U_{\verify}'}^{\dagger})^j{\compV{U_{\verify}}^{\dagger}})\right)\\
\le & \Tr(O\ketbra{\psi_0}{\psi_0}) - \Tr(O \compV{U_{\verify}} \ketbra{\psi_0}{\psi_0} \compV{U_{\verify}}^{\dagger})\\
= & \Tr(O\ketbra{\psi}{\psi}) - \Tr(O \compV{U_{\verify}} \ketbra{\psi}{\psi} \compV{U_{\verify}}^{\dagger})
\end{align*}
\end{proof}
\section{Notation Tables}\label{NotationTable}

\begin{table}[H]
\caption{Parameters in \Cref{subsec:QuantumAnalysis}}
\begin{tabularx}{\textwidth}{ 
  | >{}p {0.11\textwidth}
  | >{\arraybackslash}X | }
\toprule
  $q(n), q'(n)$ & $n$ is the security number. $\verify$ makes $q(n)$ classical queries. $\keygen$ and $\mint$ make $q'(n)$ quantum queries in total. We sometimes omit $n$. \\
  $T(n), N(n)$ & two polynomials that will decide the maximal possible number of iterations we run in the test phase and the update phase.\\
\bottomrule
\end{tabularx}
\end{table}

\begin{table}[H]
\caption{Registers in \Cref{subsec:QuantumAnalysis}}
\begin{tabularx}{\textwidth}{ 
  | >{}p {0.11\textwidth}
  | >{\arraybackslash}X | }
\toprule
  $\reg{Pk}$ & the register storing the public key\\
  $\reg{S}$ & the register storing the serial number\\
  $\reg{M}$ & the register storing the alleged money state\\
  $\reg{D}_{\adversary}$ & the register storing the classical queries database maintained by $\adversary$\\
  $\reg{D}_{\ro}$ & the register storing the classical queries we made so far along with their answers (maintained by the oracle)\\
  $\reg{F}$ & the register storing the oracle if in the decompressed view or the register storing $D_F$ (the database for non-$\ket{\hat 0}$ elements) if in compressed view\\
  $\reg{G}$, $\reg{H}$ & the register storing unimportant things for the analysis. For example, it may include the secret key, working space for $\keygen$ and $\mint$, and some unused fresh ancillas.\\
  $\reg{T}$ & the register storing the number of iterations for test phase. \\
  $\reg{J}$ & the register storing the number of iterations for update phase.\\
  $\reg{Q}, \reg{Q}_1, \reg{Q}_2$ & the register storing the next query position.\\
  $\reg{A}, \reg{A}_1, \reg{A}_2$ & the register to store the next query answer.\\
\bottomrule
\end{tabularx}
\end{table}

\begin{table}[H]
\caption{States in \Cref{subsec:QuantumAnalysis}}
\begin{tabularx}{\textwidth}{ 
  | >{}p {0.12\textwidth}
  | >{\arraybackslash}X | }
\toprule
  $\ket{\phi}$ & A state in the following form (i.e. it's in the compressed view and the contents in $\reg{D}_{\ro}$ and $\reg{D}_{\adversary}$ are the same),
  $$\ket{\phi} = \sum_{\substack{pk, s, m, D, D_F, g \\ \text{s.t. }D \cap D_F = \emptyset}}\alpha_{pk, s, m, D, D_F, g}\ket{pk}_{\reg{Pk}}\ket{s}_{\reg{S}}\ket{m}_{\reg{M}}\ket{D}_{\reg{D}_{\adversary}}\ket{D_F}_{\reg{F}}\ket{D}_{\reg{D}_{\ro}}\ket{g}_{\reg{G}}.$$ In \Cref{claim:QuantumSimOnRho} and \Cref{claim:quantumSim}, we will instantiate it with the pure state we obtain by applying the unitaries $\compV{U_{\keygen}}$, $\compV{U_{\mint}}$ and $\compV{U_{\adversary}}$ to the state $\ket{1^n}\ket{\emptyset}_{\reg{D}_{\adversary}}\ket{\emptyset}_{\reg{F}}\ket{\emptyset}_{\reg{D}_{\ro}}$ along with enough ancillas. \\
  $\ket{\phi_j}$ &$\ket{\phi_j} = \compV{U_C}U_{j - 1}\cdots \compV{U_C} U_0\ket{\phi}$. It's the state when we run $\verify^{\ro, \ket{\pspace}}$ on $\ket{\phi}$ in the compressed view until we have answered the $j^{th}$ query. \\
  $\ket{\psi}$ & We abuse the notation. $\ket{\psi}$ is the pure state we obtain by running the first step in the compressed view in the case ${\verify^{\ro, \ket{\pspace}}(pk, (s, \cdot))}$ on ${\rho_s^{(t)}}$ until the end of the \textbf{test} phase (in \Cref{claim:QuantumSimOnRho}) or the \textbf{update} phase (in \Cref{claim:quantumSim}) of $\adversary$.\\
  $\ket{\varphi}$ & an arbitary state ready for two ``parallel'' classical queries on different registers. \\
  $\ket{\psi_0}$ & the pure state we obtain by running the first step in the compressed view in the case ${\verify^{\ro, \ket{\pspace}}(pk, (s, \cdot))}$ on ${\rho_s^{(t)}}$ until we finish the \textbf{test} phase and then truncating $\reg{J}$.\\
  $\ket{\psi^{(t)}}, \ket{\psi^{(j)}}$ & We abuse the notation. In \Cref{claim:QuantumSimOnRho}, $\ket{\psi^{(t)}}$ is the state after we run $t$ iterations in the test phase but not run the update phase in the compressed view. In \Cref{claim:quantumSim}, $\ket{\psi^{(j)}}$ is the state we obtain after we run a randomized number of iterations in the test phase and $j$ iterations in the update phase in the compressed view. \\
\bottomrule
\end{tabularx}
\end{table}

\begin{table}[H]
\caption{Observable and Unitaries in \Cref{subsec:QuantumAnalysis}}
\begin{tabularx}{\textwidth}{ 
  | >{}p {0.13\textwidth}
  | >{\arraybackslash}X | }
\toprule
  $O$ & the observable corresponding to the number of pairs in $\reg{F}$ (i.e. half of the nonempty length in $\reg{F}$). Formally, $O = \sum_{D_F}\length{D_F}\ketbra{D_F}{D_F}_{\reg{F}}$ where $\length{D_F}$ is the number of pairs in $D_F$. It will only be applied to those states in compressed view.\\
  $\decomp$ & the unitary defined in \Cref{subsec:compress}. It acts on $\reg{D}_{\ro}$ and $\reg{F}$ and decompresses two databases to one database and the oracle.\\
  $\comp$ & the inverse of the unitary $\decomp$. \\
  $\compV{U}$ & $\compV{U} = \comp \decompV{U} \decomp$. It's the compressed view version of $U$ for a general unitary $U$. See the figure in \Cref{figure:compress} for more details.\\
  $\decompV{U_Q}$ & the unitary corresponding to answering a \textbf{quantum} query with the real oracle.\\
  $\decompV{U_C}$ & the unitary corresponding to answering a \textbf{classical} query with the real oracle.\\
  $\decompV{U_R}$ & the same as $\decompV{U_C}$ except that it records the query-answer for $\adversary$ at the same time.\\
  $\decompV{U_D}$ & the unitary corresponding to answering a classical query with the database in register $\reg{D}_{\adversary}$ while recording the query-answer to $\reg{D}_{\adversary}$ for later use.\\
  $\decompV{U_D'}$ & the unitary corresponding to answering a classical query with the database in register $\reg{D}_{\ro}$ while recording the query-answer to $\reg{D}_{\ro}$ for later use.\\
  $U_i$ & When $0 \le i \le q - 1$, it's the unitary corresponding to the preparation of the ${(i + 1)}^{th}$ query of $\verify$. When $i = q$, it's the unitary after the final query of $\verify$.\\
  $U_{\keygen}, U_{\mint}$ & the unitary $U_{\keygen, n}, U_{\mint, n}$ defined in \Cref{subsec:purifiedView}.\\
  $U_{\verify}, U_{\syn}$ & the unitary $U_{\verify, n}, U_{\syn, n}$ defined in \Cref{subsec:purifiedView}, $U_{\verify} = U_q U_CU_{q - 1}\cdots U_C U_0$.\\
  $U_{\verify}'$ & the unitary corresponding to doing the verification while recording the query-answer pair for $\adversary$, $U_{\verify}' = U_q U_RU_{q - 1}\cdots U_R U_0$.\\
  $U_{\simulation}$ & the unitary corresponding to running $\verify^{D, \ket{\pspace}}$ where $D$ is the content in $\reg{D}_{\adversary}$, $U_{\simulation} = U_q U_DU_{q - 1}\cdots U_D U_0$.\\
  $U_\update$ & the unitary defined in \Cref{subsec:purifiedView} that describes our update phase. Formally, it's the unitary $\sum_{j = 0}^{N(n) - 1}(U_{\verify}U_{\syn})^j\ketbra{j}{j}_{\reg{J}}.$\\
\bottomrule
\end{tabularx}
\end{table}

\end{document}